\documentclass[12pt]{article}
\usepackage[T1]{fontenc}
\usepackage[dvips]{graphicx}
\usepackage{makeidx}\makeindex
\usepackage[round]{natbib}
\graphicspath{{images/}}
\usepackage[utf8]{inputenc}

\usepackage{amsmath}
\usepackage{geometry}
\usepackage{bbm}
\usepackage{authblk}
\usepackage{amsfonts}
\usepackage{xcolor}
\usepackage{tikz}
\usepackage{comment}
\usepackage{amsthm}
\usepackage{appendix}
\usepackage{bm}
\usepackage{bbm}
\usepackage{algorithm}
\usepackage{algpseudocode}

\usepackage{mathtools}
\usepackage{amssymb}
\usepackage{enumitem}
\usepackage{verbatim}
\usepackage{amsmath,amsthm}
\usepackage{xspace}
\usepackage{pifont}
\usepackage{graphicx}
\usepackage{amssymb}
\usepackage{epic, eepic}
\usepackage{dsfont}
\usepackage{amssymb}
\usepackage{makeidx}
\usepackage{mathrsfs}
\usepackage{exscale}
\usepackage{color} 
\usepackage{overpic} 
\usepackage{bm}
\usepackage{bbm}
\usepackage{booktabs} 
\usepackage{color, colortbl}
\usepackage{subcaption}
\usepackage[round]{natbib}

\RequirePackage[colorlinks,citecolor=blue,urlcolor=blue]{hyperref}

\definecolor{Gray}{gray}{0.9}

\usepackage{amsmath,afterpage}
\usepackage{epsf}
\usepackage{graphics,color}

\def\0{\mathbf{0}}

\def\rr{\rightarrow}

\def \< {\langle}
\def \> {\rangle}

\def\beqa{\begin{eqnarray}}
\def\eeqa{\end{eqnarray}}
\def\beqas{\begin{eqnarray*}}
\def\eeqas{\end{eqnarray*}}

\newtheorem{theorem}{Theorem}[section]
\newtheorem{lemma}[theorem]{Lemma}

\newtheorem{proposition}[theorem]{Proposition}

\newtheorem{corollary}[theorem]{Corollary}

\newtheorem{remark}[theorem]{Remark}

\newtheorem{definition}[theorem]{Definition}
\newtheorem{assumption}[theorem]{Assumption}
\numberwithin{equation}{section}
\newcommand{\hatd}[1]{{}}

\newcommand{\bd}{\begin{displaymath}}
\newcommand{\ed}{\end{displaymath}}
\newcommand{\be}{\begin{equation}}
\newcommand{\ee}{\end{equation}}
\newcommand{\bq}{\begin{eqnarray}}
\newcommand{\eq}{\end{eqnarray}}
\newcommand{\bn}{\begin{eqnarray*}}
\newcommand{\en}{\end{eqnarray*}}

\def\wt{\widetilde}

\usetikzlibrary{decorations.pathreplacing,positioning, arrows.meta}

\usepackage{authblk}

\author[1]{Rama Cont} 
\author[2]{ Alessandro Micheli\footnote{AM is supported by the EPSRC Centre for Doctoral Training in Mathematics of Random Systems: Analysis, Modelling and Simulation (EP/S023925/1)}}
\author[3]{ Eyal Neuman}

\affil[1]{Mathematical Institute, University of Oxford}
\affil[2,3]{Department of Mathematics, Imperial College London}

\title{Fast and Slow Optimal Trading with Exogenous Information}
\date{\today}

\begin{document}
\maketitle
\begin{abstract}
We model the interaction between a slow institutional investor and a high-frequency trader as a stochastic multiperiod Stackelberg game. 
The high-frequency trader exploits price information more frequently and is subject to periodic inventory constraints.  We first derive the optimal strategy of the high-frequency trader given any admissible strategy of the institutional investor. Then, we solve the problem of the institutional investor given the optimal  strategy of the high-frequency trader, in terms of the resolvent of a Fredholm integral equation, thus establishing the unique multi-period Stackelberg equilibrium of the game. Our results provide an explicit solution which shows that the high-frequency trader can adopt either predatory or cooperative strategies in each period, depending on the tradeoff between the order-flow and the trading signal. We also show that the institutional investor's strategy is  more profitable when the order-flow of the high-frequency trader is taken into account.  
\end{abstract}
\begin{description}
\item[Mathematics Subject Classification (2010):] 49N70, 49N90, 93E20, 60H30 
\item[JEL Classification:] C73, C02, C61, G11
\item[Keywords:] Optimal stochastic control, stochastic games, price impact, predictive signals, Stackelberg equilibrium

\end{description}

\bigskip

\section{Introduction}

Modern financial markets involve a  range of participants who place buy and sell orders across a wide spectrum of time scales: on one end, pension funds rebalance their portfolio on an annual basis and mutual fund managers rebalance typically on a monthly time scale while, on the other end of the spectrum, electronic market makers and high frequency trading firms submit several thousands of orders per second (see e.g \cite{cont2011}), while having strict inventory constraints (see p.4 of \cite{hf-rep}). Although this heterogeneity in time scales has been always present, the development of computerized trading in electronic markets has substantially widened the range of frequencies at which various market participants operate. 
The interaction between the flow of buy and sell orders from these different participants results in an aggregate order flow which is the superposition of components across a wide range of frequencies. The consequences of this phenomenon for market volatility, price dynamics and market stability have yet to be systematically explored. 

This heterogeneity of frequencies stands in contrast with mathematical models of market microstructure and price dynamics which are often formulated in terms of homogeneous agents operating at a single time scale as in  \citep{GARLEANU16,EvangelistaThamsten:20,voss.19,N-sch-22,M-M-N-21,DrapeauLuoSchiedXiong:19, CasgrainJaimungal:20, FuGraeweHorstPopier:20,N-V-2021} among others. Yet, the repeated occurrence of `flash crashes' (see e.g. \cite{kirilenko2017}) demonstrates that components at different frequencies may strongly interact and possibly lead to market disruption, calling for a modeling framework which incorporates the interaction of agents operating on different time scales.

As a first step to investigate these phenomena, we propose a model for the dynamics of prices and order flow in a market where participants of two different frequencies submit buy and sell orders on a risky asset. Specifically, we consider a stochastic game between an institutional investor and a high-frequency trader who are exploiting an exogenous signal which interacts with the price process in the drift term. The institutional investor and high-frequency trader, which will be referred to as \emph{major agent} and \emph{minor agent}, respectively, interact through their aggregated order-flow, which is resulting by their own trades. The trades of both agents create temporary and permanent price impact which affect the asset price process. We model this system by means of two coupled multi-period stochastic control problems over a fixed time horizon $T$, where the high-frequency trader exploits the exogenous information continuously, but is also subject to periodic inventory constraints at the end of any sub-period $0 <t_1<...<t_{n} = T$, for some $n\geq1$. On the other hand, the institutional investor has a limited access to the signal but she is only subject to inventory constraints at time $T$. Since in the setting that we wish to describe the minor agent has a clear advantage in terms of information exploitation, it is natural to look for a Stackelberg equilibrium in this game, where the minor agent takes advantage of the signal and the order-flow which is created by the major agent's transactions.

Our first result derives the unique optimal strategy of the high-frequency trader given any admissible strategy of the major agent (see Theorem \ref{thm-minor-optimal-strategy}). The challenging part in establishing a Stackelberg equilibrium is to derive the strategy of the player who plays first, namely the major agent. We develop a novel approach for this class of Stackelberg games in order to derive the major agent's optimal strategy given the optimal signal-adaptive strategy of the minor agent using tools from the theory of integral equations. Specifically, in Theorem  \ref{thm-major-solution} we describe the unique optimal major agent's strategy in terms of the resolvent of a Fredholm integral equation, thus establishing the unique multi-period Stackelberg equilibrium of the game. In Section \ref{sec-illustrations} we illustrate the solutions to the Stackelberg game and in Section \ref{sec-numerics} we derive the additional technical steps that are needed in order to obtain such explicit results directly from Theorems \ref{thm-minor-optimal-strategy} and \ref{thm-major-solution}. 

Our method for solving this Stackelberg game introduces concepts from the theory of integral equations, which according to our knowledge were  not used before in order to solve such equilibrium problems. These tools allow us to derive explicitly the equilibrium of this multi-period  Stackelberg model, which is quite a surprising as typically these models are highly intractable and can be solved only in some special cases. See the comparison with related papers by \cite{Carlinetal}, \cite{SchoenebornSchied} and \cite{fast_slow_informed} which is discussed later in this section. For this reason we put an effort to create a self-contained text, which introduces the required background for the application of our method.  In particular in Sections \ref{sec-numerics} and \ref{proof-convergence-numerics} and in Appendix \ref{sec-spectral-explicit} we describe in detail the numerical scheme which is used in order to plot the optimal strategy of the major agent in Theorem  \ref{thm-major-solution}, and we provide the proof of its convergence. These details are given in order demonstrate how the theoretical solution in Theorem  \ref{thm-major-solution}, which is given in terms of a spectral decomposition for a resolvent of an operator in \eqref{eq-resolvent-kernel-def}, can be approximated and plotted. The discussed operator $\mathcal G$ arises from the solution to the minor agent's problem (see \eqref{eq-kernel-caligraphic}).

From our main theoretical results we derive explicit expressions for both agents equilibrium strategies which have fascinating economic interpretation regarding the trading behaviour of high-frequency traders and on the best practices for  institutional investors who are executing large meta-orders. We summarise these insights in the following list and refer the reader for the comprehensive discussion in Section \ref{sec-illustrations}: 
\begin{itemize} 
\item[\textbf{(i)}] Our results suggest that the high-frequency trader can adopt either predatory or cooperative strategy with respect to the major agent in each period, depending on the tradeoff between the order-flow of the major agent and the trading signal during the period. See Figure \ref{fig-multiday-liquidation-trades} for specific realisations of such strategies.
\item[\textbf{(ii)}] We compare the revenues of the major agent's optimal order execution with a benchmark optimal strategy in which the agent is not taking into account of the minor agent's trading activity.  In Figure \ref{fig-benchmark-savings} we show that the major agent's optimal strategy on average considerably outperforms the benchmark strategy. This contrasts with the common belief that high-frequency traders order-flow can be regarded as noise. 

\item[\textbf{(iii)}] We show that the major agent's and minor agent's optimal trading strategies induce the well-known U-shaped pattern of intraday trading volume, where the traded volume peaks at the beginning and at the end of the day (see Figure \ref{fig-volume-curve-daily}). 
\end{itemize} 

Our model is related to a class of predatory trading models which was introduced by \cite{Carlinetal} for a single period and  further developed by \cite{SchoenebornSchied} for two-periods.  In \cite{Carlinetal} a single-period multi-agent game was introduced where traders are liquidating simultaneously where while creating both temporary and permanent price impact which affects the price process. In their model there are two types of agents: \emph{sellers} which start with a positive amount of assets and \emph{competitors} who have zero initial positions. All agents are seeking to maximise simultaneously similar revenue functionals, using strictly deterministic strategies. Their main results derive a Nash equilibrium for the game. In the single period case it is shown that, under some assumptions on the model parameters, if the seller is liquidating then the competitor is first selling and later buying her position back due to inventory constraints (see Figure 1 in  \cite{SchoenebornSchied}). In the two period model the seller can liquidate only in the first period, while the competitor can execute her strategy over two periods. Depending on the price impact parameters, there are two possible scenarios: either the competitor is buying in the first period and then selling in the second period, i.e. introducing cooperative strategies in the game (see Figure 8 therein), or doing a round trip of selling first and then closing the position all in the first period. 

Our model is different from the \cite{SchoenebornSchied} in a few critical points. First, we assume that minor agent (resp. competitor) is trading at a higher frequency than the major agent (resp. seller). This is reflected in the model as periodic inventory constraints in the minor agent's revenue functional. This term do not appear in the major agent's objective, who has a fuel constraint only at the end to the trading time horizon. The minor agent is also reacting continuously to exogenous information while the major agent has access to the information only at the beginning of the trade. This means also that the minor agent's optimal strategy is stochastic, unlike the deterministic game which was studied in \cite{SchoenebornSchied}. 
Another major difference between these models is in the type of equilibrium which is derived. In \cite{SchoenebornSchied} an open-loop Nash equilibrium was derived, which means that all traders optimise simultaneously. From market microstructure setting with various frequencies, it is essential to consider a Stackelberg equilibrium as the minor agent is indeed reacting to the major agent's selling strategy. As stated before, neither \cite{Carlinetal} nor  \cite{SchoenebornSchied} take into account exogenous information, therefore, their optimal strategies are always found to be deterministic. One of the main conclusions of our analysis is that this aspect has a prominent effect on the behaviour of the major agent and the minor agent, which is not captured in \cite{Carlinetal} and \cite{SchoenebornSchied}. Finally, despite the clear asymmetry in our model between the agents in the access to information, type of equilibrium and inventory constraints, which make the problem quite involved and required us to introduce new methods for Stackelberg games, we are able to derive explicit solutions for any number of time periods, in contrast to \cite{SchoenebornSchied}, where only the two period model is tractable.  

We briefly mention in this context that \cite{fast_slow_informed} studied a discrete-time model where fast traders, whose decisions depend on a market signal, trade simultaneously with slow traders, who can only observe a lagged version of that same signal. However besides this difference in the access to information, the fast agents do not have different objective functionals nor inventory constraints which differ them from the slow agents, which are some of the main ingredients in our model.

  \paragraph{Structure of the paper.} The rest of the paper is organised as follows. In Section \ref{sec-model-setup} we define the two player model. Our main results regarding the explicit solution to the Stackelberg game are presented in Section \ref{sec-main-result}. Section \ref{sec-illustrations} contains the illustrations and the financial interpretation of the main results. In Section \ref{sec-numerics} we derive rigorously the numerical scheme that we have used in order to plot the solutions in Section \ref{sec-illustrations}. The proofs of the results of this paper are given in Sections \ref{sec-proof-thm-minor-optimal}--\ref{proof-convergence-numerics} as well as in Appendices \ref{sec-spectral-explicit}--\ref{sec-proof-prop-ode}.

\section{Model Setup}\label{sec-model-setup}
We define the Stackelberg game between a major agent liquidating an initial amount of shares in a risky assets and a proprietary high frequency trader (HFT) trading on the same asset and who,  throughout this paper, we will regard as a minor agent.  
Let $T>0$ denote a finite deterministic time horizon and fix a filtered probability space $(\Omega, \mathcal F,(\mathcal F_{t})_{t\in[0,T]},\mathbb{P} )$ satisfying the usual conditions of right continuity and completeness. The set $\mathcal H^2$ represents the class of all (special) semimartingales~$P=(P_t)_{t\in[0,T]}$ whose canonical decomposition $P =   M + A$ into a (local) martingale $  M=(  M_t)_{t\in[0,T]}$ and a predictable finite-variation process $A=(A_t)_{t\in[0,T]}$ satisfies
\be \label{ass:P} 
\mathbb{E} \left[ \langle    M  \rangle_T \right] + \mathbb{E}\left[\left( \int_0^T |dA_s| \right)^2 \right] < \infty.  
\ee
We denote by $L^{2}([0,T])$ the space of square integrable functions $f:[0,T] \rr \mathbb{R}$ and by 
$\langle \cdot ,\cdot\rangle_{L^{2}}$ the inner product on $L^{2}([0,T])$, that is 
\begin{equation*}
\langle f,g\rangle_{L^{2}} = \int^{T}_{0}f(t)g(t)dt, \quad f,g\in L^{2}([0,T]),
\end{equation*}
and by $||\cdot||_{L^{2}}$ the associated norm.

 \paragraph{Admissible strategies and price impact.}
The major agent has an initial holding of $q_0 \in \mathbb{R}$ shares in a risky asset. Her trading rate $\nu^{0}= (\nu^{0}_{t})_{t\in[0,T]}$ is chosen from the class of fuel-constrained \textit{deterministic} admissible strategies $\mathcal{A}_{M}^{q_0}$, which is defined as 
\begin{equation}\label{eq-major-admissible-set}
\mathcal{A}_{M}^{q_0} :=\left\{ \nu\in L^{2}([0,T]) \text{ s.t. } \int^{T}_{0}\nu_{t}dt=q_0\right\}.
\end{equation}
Her trading rate $\nu^{0}$ affects her inventory process $  Q^{0,\nu^{0}}$ so that
\begin{equation}\label{eq-major-inventory-def}
    Q^{0,\nu^{0}}_{t} = q_0- \int^{t}_{0}\nu^{0}_{s} ds, \quad 0\leq t \leq T. 
\end{equation} 
The minor agent, being a proprietary high frequency trader, is assumed to have a zero initial position in the risky asset. Her trading rate $\nu^{1}= (\nu^{1}_{t})_{t\in[0,T]}$ is chosen from a class of \textit{adaptive} admissible strategies
\begin{equation} \label{eq-minor-admissible-set-def}
    \mathcal{A}_{m}  := \left\{\nu \textrm{ progressively measurable s.t. } \mathbb{E}\left[\int_{0}^{T} \nu_{s}^{2}ds\right]<\infty\right\}.
\end{equation}
Her trading rate $\nu^{1}$ affects her inventory process $  Q^{1,\nu^{1}}$ so that
\begin{equation}
\label{eq-minor-inventory}
Q^{1,\nu^{1}}_{t}=-\int^{t}_{0}\nu_{s}^{1} ds, \quad 0\leq t\leq T.
\end{equation}
Throughout, we use the notation $\nu=(\nu^{0},\nu^{1})$ for the major agent's control $\nu^{0}$ and the minor agent's control $\nu^{1}$. Once $\nu$ is fixed, the visible asset mid-price $P^{\nu}$ satisfies
\begin{equation}
\label{eq-P-nu-definition}
P_{t}^{\nu} =P_{t}-Y_{t}^{\nu}  , \quad 0\leq t\leq T, 
\end{equation}
where $P\in\mathcal{H}^{2}$ and where $Y^{\nu}$ is the permanent price impact price impact \`a la~\citet{OPTEXECAC00}, which is generated by both agents and which is given by 
\begin{equation}\label{eq-permanent-impact}
Y_{t}^{\nu} = \int^{t}_{0}(\kappa_{0}\nu_{s}^{0} + \kappa_{1}\nu_{s}^{1})ds,
\end{equation}
where $\kappa_{i},$ $i=1,2$, are positive constants. 

\paragraph{Major agent: institutional investor}
The major agent's execution price is affected instantaneously in an adverse manner through the presence of linear temporary price impact. The major agent's execution price is taken to be
\begin{equation}
\label{eq-execution-price-major}
S_{t}^{0,\nu}  =P_{t}^{\nu} -\lambda_{0}\nu_{t}^{0}, 
\end{equation}
where $\lambda_{0}$ is a positive constant measuring the magnitude of her temporary price impact. As a result, the major agent's cash holdings satisfy
\begin{equation}
\label{eq-major-cash}
    X^{0,\nu}_{t}  =x_{0} + \int^{t}_{0}S_{s}^{0,\nu}   \nu^{0}_{s}ds, \quad  0\leq t \leq T.
\end{equation}
 The major agent's objective is to optimally unwind her initial position $q_0$ by the trading horizon $T$, so to minimise her execution costs. This is equivalent to maximising the expected revenues from her liquidation, therefore, we take the major agent's performance functional to be
 \begin{equation}\label{eq-major-functional}
 H^{0}\left(\nu^{0};\nu^{1}\right) := \mathbb{E}\left[X_{T}^{0,\nu}\right].
\end{equation}

\paragraph{Minor agent: high-frequency trader.}

As in the case of the major agent, the transactions of the minor agent create temporary price impact, such that the execution price of her orders is given by
\begin{equation}\label{eq-execution-price-minor}
S_{t}^{1,\nu} = P_{t}^{\nu} - \lambda_{1}\nu_{t}^{1}. 
\end{equation} 
where $\lambda_{1}$ is a positive constant. Note that the temporary price impact parameter is likely to be smaller for the minor agent as HFTs can take advantage of the order-book realtime information in order to reduce their price impact. 

The minor agent's cash process is given by
\begin{equation}
\label{eq-minor-cash}
X_{t}^{1,\nu} = x_{1} + \int^{t}_{0}S_{s}^{1,\nu} \nu_{s}^{1}ds, \quad 0\leq t\leq T. 
\end{equation}

The minor agent wishes to maximise her cash, however as an HFT, she is inclined to avoid \textit{overnight risk}, specifically, in the form of non-zero overnight inventory. As an example, consider $T$ to be one business week, such that $[0,T]$ can be partitioned in five disjoint and contiguous intervals of equal duration $\tau$, where each intervals represents the market hours of each business day from Monday to Friday. Without loss of generality, in the context of our example we assume that the minor agent's intraday risk preferences are independent of the business day considered and we ignore the possibility of after-hours trading. Since the minor agent wishes to close her position by the end of each day, then as often done for terminal inventory penalties in the context of single-day liquidations, we can introduce a penalisation for non-zero inventory at the end of each day. These dynamic inventory preferences can be accounted by modelling the running inventory costs of the minor agent via a periodic function of period $\tau$ which drastically increases towards the end of each day (see e.g. \eqref{eq-periodic-phi1}), i.e. as $t$ approaches $\tau,2\tau,3\tau,4\tau,5\tau$ from the left.  Mathematically, in order to capture the minor agent's dynamic inventory preferences of our example and more general ones, we define the minor agent's running inventory costs in terms of a function $\phi^{1}:[0,T] \rr \mathbb R_+$  which we take to be piecewise continuous and locally bounded. 

The minor agent risk-revenue functional is therefore given by
\begin{equation}\label{eq-minor-functional}
 H^{1}(\nu^{1};\nu^{0}) := \mathbb{E}\left[ X_{T}^{1,\nu} + Q^{1,\nu^{1}}_{T}\left(P_{T}^{\nu}-\alpha Q^{1,\nu^{1}}_{T}\right)  - \int^{T}_{0}\phi^{1}_{t} \left(Q^{1,\nu^{1}}_{t}\right)^{2}dt\right].
\end{equation}
The first two terms in~\eqref{eq-minor-functional} represent the trader's terminal wealth; that is, her final cash position, accounting for the accrued trading costs which are induced by temporary price impact and the permanent price impact of both agents as prescribed in~\eqref{eq-execution-price-minor}, as well as the mark-to-market value of her terminal risky asset position. The third and fourth terms in~\eqref{eq-minor-functional} implement a penalty $\phi_t^{1}>0$ and $\alpha > 0$ on her running and terminal inventory, respectively. Also observe that $H^{1}(\nu^{1};\nu^{0}) < \infty$ for any pair of admissible strategies $\nu^0\in   \mathcal{A}^{q_0}_{M}$ and $\nu^{1} \in \mathcal{A}_{m}$.

\paragraph{The Stackelberg game.} We formulate the competition between the major agent and the minor agent as a stochastic Stackelberg game in which the minor agent is reacting to the major agent's trading. Mathematically, the game unfolds in two steps:
\begin{enumerate}
\item[(i)] \textit{Minor Agent's Problem:} for a given major agent's liquidation strategy $\nu^{0}\in\mathcal{A}_{M}^{q_0}$,  the minor agent chooses her own strategy $\nu^{1,*}(\nu^{0}) \in \mathcal{A}_{m}$ in order to maximise her objective functional $H^{1}$;
\item[(ii)]  \textit{Major Agent's Problem:} given the optimal minor agent's strategy $\nu^{1,*}$ established in (i), the major agent determines the optimal liquidation strategy $\nu^{0,*}\in \mathcal{A}_{M}^{q_0}$  in order to maximise her objective functional $H^{0}$.
\end{enumerate}
In the context of our model, we formalise the definition of Stackelberg equilibrium as follows. 
\begin{definition}[Stackelberg equilibrium] \label{def-stac} 
A pair $\nu^*:=(\nu^{0,*},\nu^{1,*}(\nu^{0,*}))$ where $\nu^{0,*}$ and $\nu^{1,*}(\nu^{0,*})$ solve the major and minor agent's problems, respectively, is called a Stackelberg equilibrium. 
\end{definition}


\section{Main Results}
\label{sec-main-result}
Our main results derive explicitly the unique Stackelberg equilibrium of the game. As stated at the end of Section \ref{sec-model-setup}, we start by solving the minor agent's problem.
\subsection{Solution to the Minor Agent's Problem}
\label{sec-minor-results}

We denote by $L^{2}([0,T]^{2})$ the space of measurable kernels $\mathcal{T}:[0,T]^{2}\to\mathbb{R}$ such that
\begin{equation} \label{bnd-T} 
\int^{T}_{0}\int^{T}_{0}\mathcal{T}(t,s)^{2}dtds <\infty.
\end{equation}
Henceforth, we make the following assumption.
\begin{assumption}
\label{ass-alpha-lambda}
We assume that the parameters $\alpha$ in \eqref{eq-minor-functional} and $\kappa_{1}$ in \eqref{eq-permanent-impact}  are chosen such that
\begin{equation*}
2\alpha \geq \kappa_{1}.
\end{equation*}
\end{assumption}

Let $r^{1}= (r^{1}_{t})_{t\in[0,T]}$ be the solution to the following Riccati equation with a time varying coefficient, 
\begin{equation} \label{eq-v1-ode}
\begin{cases}
\partial_{t}r^{1}_{t} &= \frac{1}{\lambda_{1}}\phi^{1}_{t} -(r^{1}_{t})^{2}, \\ 
r^{1}_{T}&=-\frac{2\alpha- \kappa_{1}}{2\lambda_{1}} .
\end{cases}
\end{equation}
Under Assumption \ref{ass-alpha-lambda}, the solution $r^{1}$ of \eqref{eq-v1-ode} exists and is unique over $[0,T]$ (see Proposition \ref{prop-results-ode}).
We further define
\begin{equation}\label{eq-def-xi-pm}
\xi_{t}^{\pm} := e^{\pm \int^{t}_{0}r^{1}_{z}dz}, \quad 0\leq t \leq T, 
\end{equation}
as well as the kernel $\mathcal{K}:[0,T]^{2}\to\mathbb{R}_{+}$ which is given by 
\begin{equation} \label{eq-kernel-minor}
\mathcal{K}(t,s):= \xi_{t}^{-}\xi_{s}^{+}, \quad 0\leq t,s  \leq T.
\end{equation}
Note that the kernel $\mathcal{K}$ is in $L^{2}([0,T]^{2})$ (see Lemma \ref{lemma-kernel-L2}). Moreover, for any $\nu^{0}\in\mathcal{A}_{M}^{q_0}$ we define the predictable process
\begin{equation} \label{eq-v0}
r^{0}_{t} := \frac{1}{2\lambda_{1}} \mathbb{E}_{t}\left[\int^{T}_{t} \mathcal{K}(t,s)(dA_{s}-\kappa_{0}\nu^{0}_{s}ds)
\right],  \quad 0\leq t \leq T. 
\end{equation}

The solution to the minor agent problem is given in the following theorem. 
\begin{theorem}[Solution to the minor agent's problem]\label{thm-minor-optimal-strategy}
Let $\nu^{0}\in\mathcal{A}_{M}^{q_0}$. Under Assumption \ref{ass-alpha-lambda}, there exists a unique optimal strategy $\nu^{1,*}(\nu^{0})\in\mathcal{A}_{m}$ that maximizes \eqref{eq-minor-functional}. This strategy is given by
\begin{equation} \label{eq-minor-agent-optimal}
\nu^{1,*}_{t}=-\left(r^{0}_{t}+ r^{1}_{t}\int^{t}_{0} \mathcal{K}(s,t)r^{0}_{s}ds\right), \quad 0\leq t \leq T. 
\end{equation}
\end{theorem}
The proof of Theorem \ref{thm-minor-optimal-strategy} is given in Section \ref{sec-proof-thm-minor-optimal}. 
\begin{remark}
\label{remark-ass-minor}
In Lemma \ref{lemma-minor-strictly-concave} we show that Assumption \ref{ass-alpha-lambda} is a sufficient condition to guarantee the strict concavity of the minor agent's functional \eqref{eq-minor-functional}, hence the uniqueness of the solution to the minor agent’s problem.
\end{remark}
\begin{remark}
\label{remark-minor-literature}
Note that minor agent's optimal control in \eqref{eq-minor-agent-optimal} can be written in feedback form as follows, 
\begin{equation*} \label{eq-minor-agent-optimal-feedback}
\nu^{1,*}_{t}=-\left(r^{0}_{t}+ r^{1}_{t}Q_{t}^{1,\nu^{1,*}}\right), \quad 0\leq t \leq T. 
\end{equation*}
In the special case where there is no permanent price impact, that is $\kappa_{i} = 0$, $i=1,2$ and the risk aversion function $\phi^{1}$ is a positive constant, \eqref{eq-minor-agent-optimal} coincides with the optimal strategy in \cite[Theorem 3.1]{BMO:19}.
\end{remark}
 
\subsection{Solution to the Major Agent's Problem}
\label{sec-main-resul-major}
Our next step is to derive the maximiser of the major agent's objective functional \eqref{eq-major-functional}, given the minor agent's optimal strategy $\nu^{1,*}$ in \eqref{eq-minor-agent-optimal}. As it is often the case in Stackelberg games, solving the second phase of the game is technically challenging and rarely achievable. In order to do so we make the following simplifying assumption on the signal $A$ in \eqref{ass:P}. We assume that the signal process $A$ is given by
\begin{equation*}
A_{t}=\int^{t}_{0}\mu_{s}ds, \quad 0\leq t \leq T. 
\end{equation*}
where $\mu=(\mu_{t})_{t\in[0,T]}$ is an $(\mathcal{F}_t)_{t\in[0,T]}$-adapted stochastic process  satisfying
\begin{equation}
\label{eq-mu-square-integrable}
\int^{T}_{0}\mathbb{E}[\mu_{t}^{2}]dt <\infty.
\end{equation}
Note that this assumption is an adaptation of the assumptions made in \cite{Car-Jiam-2016,Lehalle-Neum18} on the signal for single agent optimal execution problems to the present context. We further denote 
\be \label{exp-s} 
  \bar{\mu}_{t} := \mathbb{E}[\mu_{t}] , \quad 0\leq t \leq T. 
\ee
Next, we introduce some essential definitions related to linear operators in $ L^{2}([0,T])$. 
\paragraph{Definitions for linear operators in $ L^{2}([0,T])$.} 

For any linear operator $\mathsf{T}$ from $L^{2}([0, T])$ to $L^{2}([0, T])$ we define the operator norm
\begin{equation}
\label{eq-operator-norm-definition}
||\mathsf{T}|| := \sup \left\{||\mathsf{T}\psi||_{ L^{2}}: \psi\in L^{2}([0,T]), \ ||\psi||_{ L^{2}}\leq 1 \right\}, 
\end{equation}
and we denote by $B(L^{2}([0,T]))$ the space of all bounded linear operator from $L^{2}([0, T ])$ to $L^{2}([0, T ])$ with respect to the operator norm \eqref{eq-operator-norm-definition}. 

For any kernel $\mathcal{T}\in L^{2}([0,T]^{2})$ (see  \eqref{bnd-T}) we say that $\mathsf{T}$ is the integral operator generated by the kernel $\mathcal{T}$ if for any $\psi\in L^{2}([0,T])$, 
\begin{equation*}
(\mathsf{T}\psi)(t) = \int^{T}_{0}\mathcal{T}(t,s)\psi(s)ds, \quad 0\leq t \leq T.
\end{equation*}
Any integral operator generated by a kernel in $ L^{2}([0,T]^{2})$ is in $B(L^{2}([0,T]))$ by the Cauchy-Schwarz inequality. 

If $\mathsf{T}_{1}$ and $\mathsf{T}_{2}$ are two operators in $B(L^{2}([0,T]))$, then we denote by $\mathsf{T}_{2}\mathsf{T}_{1}$ the operator obtained by composing $\mathsf{T}_{2}$ with $\mathsf{T}_{1}$, that is for any $\psi\in L^{2}([0,T])$, 
\begin{equation*}
(\mathsf{T}_{2}\mathsf{T}_{1}\psi)(t) := (\mathsf{T}_{2}(\mathsf{T}_{1}\psi))(t), \quad 0\leq t \leq T.
\end{equation*}

\paragraph{Special operators for our setting.} 
Recall that $\mathcal K$ was defined in \eqref{eq-kernel-minor}. We introduce the kernel $\mathcal{G}:[0,T]^{2}\to\mathbb{R}_{+}$ defined as
\begin{equation}
\label{eq-kernel-caligraphic}
\mathcal{G}(t,s) := \int^{t\wedge s}_{0} \mathcal{K}(u,t)\mathcal{K}(u,s)du, \quad 0\leq t,s \leq T.
\end{equation}
Note that the kernel $\mathcal{G}$ is symmetric and in $L^{2}([0,T]^{2})$ (see Proposition \ref{prop-kernel-caligraphic-properties}). We define the operators $\mathsf{G}$ and $\mathsf{S}$ acting on any  $\psi\in L^{2}([0,T])$ as follows,
 \begin{align}
(\mathsf{G}\psi)(t)  &:= \int^{T}_{0} \mathcal{G}(t,s) \psi(s)ds, \label{eq-G-rep-integral}\\
(\mathsf{S}\psi)(t)&:= \frac{1}{2\lambda_{0}}\int^{T}_{0}\mathbbm{1}_{\{s\leq t\}}\psi(s)ds +\frac{\kappa_{1}}{4\lambda_{1}\lambda_{0}}(\mathsf{G}\psi)(t).\label{eq-chi-function}
\end{align}
Note that the both operators $\mathsf{G}$ and $\mathsf{S}$ are in $B(L^{2}([0,T]))$ (see Proposition \ref{prop-G-adjoint-compact} and Lemma \ref{lemma-chi-continuous}).  Moreover,  the operator $\mathsf{G}$ admits a spectral decomposition in terms of a sequence of positive eigenvalues $(\zeta_{n})_{n\geq 1}$ and a corresponding sequence of eigenfunctions $(\psi_{n})_{n\geq 1}$ in $L^{2}([0,T])$ (see Lemma \ref{lemma-spectral-theorem-G}). We define the \textit{resolvent} kernel $\mathcal{R}:[0,T]^{2}
\to\mathbb{R}$ as  
 \begin{equation}
 \label{eq-resolvent-kernel-def}
\mathcal{R}(t,s) = -\frac{\kappa_{1}\kappa_{0}}{2\lambda_{0}\lambda_{1}} \mathcal{G}(t,s) + \sum_{n\geq 1} \frac{1}{1+\frac{\kappa_{1}\kappa_{0}}{2\lambda_{0}\lambda_{1}}\zeta_{n}} \left(\frac{\kappa_{1}\kappa_{0}}{2\lambda_{0}\lambda_{1}}\zeta_{n}\right)^{2} \psi_{n}(t)\psi_{n}(s), 
\end{equation}
for all $t,s\in[0,T]$ and where the sum converges uniformly and uniformly-absolutely over $[0,T]^{2}$, see Remark \ref{remark-absolute-convergence} for details. Moreover, we define the \textit{resolvent} operator $\mathsf{R}$, acting on any $\psi\in L^{2}([0,T])$ as follows, 
\begin{equation}
\label{eq-resolvent-operator}
(\mathsf{R}\psi)(t) := \psi(t) + \int^{T}_{0}\mathcal{R}(t,s) \psi(s)ds, \quad 0\leq t \leq T. 
\end{equation}
The operator $\mathsf{R}$ is also in $B(L^{2}([0,T]))$, this is proved later in Proposition \ref{prop-properties-resolvent}.

\paragraph{Notation.} We denote by $\mathit{1}(t)$ the constant function which equals to $1$ everywhere on $[0,T]$.

We are ready to state our main result regarding the solution to major's agent problem conditional on the minor agent adopting the strategy $\nu^{1,*}$ given in \eqref{eq-minor-agent-optimal}. Recall that $\bar \mu$ was defined in \eqref{exp-s} and $\mathsf{S}$ was defined in \eqref{eq-chi-function}. 
\begin{theorem}[Solution to the major agent's problem] \label{thm-major-solution}
Assume that $\nu^{1,*}$ is given by \eqref{eq-minor-agent-optimal} and that Assumption \ref{ass-alpha-lambda} holds. Then, there exists a unique optimal strategy $\nu^{0,*}\in\mathcal{A}_{M}^{q_0}$ that maximizes  the major agent's objective functional \eqref{eq-major-functional}. It is given by
\begin{equation}
\label{eq-major-solution}
\nu^{0,*}_{t} =\frac{\eta}{2\lambda_{0}}(\mathsf{R}\mathit{1})(t)  + (\mathsf{R}\mathsf{S}\bar{\mu})(t), \quad 0\leq t \leq T, 
\end{equation}
where
\begin{equation}
\label{eq-constant-major-solution}
\eta = 2\lambda_{0}\frac{q_0 -  \langle\mathsf{R} \mathsf{S}\bar{\mu},\mathit{1} \rangle_{L^{2}} }{ \left\langle \mathsf{R} \mathit{1}, \mathit{1}\right\rangle_{L^{2}}}.
\end{equation}
Moreover, $\nu^{0,*}_{t}$ is continuous on $[0,T]$.
\end{theorem}
The proof of Theorem \ref{thm-major-solution} is given in Section \ref{proof-thm-major-solution}. In the proof of Theorem \ref{thm-major-solution} we also show that the constant $\eta$ in \eqref{eq-constant-major-solution} is well-defined, which is an ingredient in proving the admissibility of the optimal strategy \eqref{eq-major-solution}. 

The following corollary follows immediately from Theorem \ref{thm-minor-optimal-strategy} and Theorem \ref{thm-major-solution}.
\begin{corollary}
\label{corollary-stackelberg}
Let $\nu^{0,*}$ and $\nu^{1,*}(\nu^{0,*})$ as in Theorem \ref{thm-major-solution} and Theorem \ref{thm-minor-optimal-strategy}, respectively. Then, under Assumption \ref{ass-alpha-lambda}, the pair $(\nu^{0,*}, \nu^{1,*}(\nu^{0,*}))\in \mathcal{A}_{M}^{q_0}\times \mathcal{A}_{m}$ is the unique Stackelberg equilibrium in the sense of Definition \ref{def-stac}.
\end{corollary}
 The following remarks discuss the result of Corollary \ref{corollary-stackelberg}. 
\begin{remark} Note that $\nu^{0,*}$ in \eqref{eq-major-solution} is given in terms of the resolvent operator $\mathsf{R}$. In Section \ref{sec-numerics} we derive a numerical scheme that approximates $\nu^{0,*}$
by using finite dimensional projections of $\mathsf{G}$. The problem of computing $\mathsf{R}$ and hence $\nu^{0,*}$ is reduced to a finite-dimensional problem of matrix inversion. We refer to Proposition \ref{prop-degenerate-kernel-approximation} and Theorem \ref{thm-uniform-convergence-numerics} for the details. 
\end{remark} 
\begin{remark}
The most challenging step in obtaining a Stackelberg equilibrium is to derive the strategy of the player who acts first, namely the major agent. In our case we needed to develop a novel approach for deriving the optimal strategy in \eqref{eq-major-solution}, using tools from the theory of integral equations. In Section \ref{sec-illustrations} we illustrate the solutions to the Stackelberg game and in Section \ref{sec-numerics} we derive additional technical steps, which are needed in order to plot such explicit solutions directly from Theorems \ref{thm-minor-optimal-strategy} and \ref{thm-major-solution}. 
\end{remark}
\begin{remark}
Our illustrations in Section \ref{sec-illustrations} suggest that the minor agent can adopt either predatory or cooperative strategy with respect to the major agent, in each period, depending on the tradeoff between the order-flow of the major agent and the trading signal during the period (see Figure \ref{fig-multiday-liquidation-trades}). This qualitative behaviour can be compared with the deterministic model of \cite{SchoenebornSchied}, who showed that in the single period case the competitor is also selling and then buying her position back due to inventory constraints (see Figure 1 therein). In their two period model the seller is selling only the first period and then depending on the price impact parameters there are two possible scenarios: either the competitor is buying in the first period and then selling in the second period, i.e. introducing cooperative strategies in the game  (see Figure 8 therein), or doing a round trip of selling first and then closing the position, all in the first period. 
\end{remark} 

\begin{remark}
\label{remark-absolute-convergence}
We remark that the sum appearing in \eqref{eq-resolvent-kernel-def} satisfies the following convergence properties.  Define 
\begin{equation*}
\begin{aligned}
 \mathfrak{R}_{N}(t,s)&= \sum_{n=1}^{N} \frac{1}{1+\frac{\kappa_{1}\kappa_{0}}{2\lambda_{0}\lambda_{1}}\zeta_{n}} \left(\frac{\kappa_{1}\kappa_{0}}{2\lambda_{0}\lambda_{1}}\zeta_{n}\right)^{2} \psi_{n}(t)\psi_{n}(s), \quad t,s\in[0,T],\\ 
 \mathfrak{R}_{N}^{\text{abs}}(t,s) &=  \sum_{n=1}^{N}\left| \frac{1}{1+\frac{\kappa_{1}\kappa_{0}}{2\lambda_{0}\lambda_{1}}\zeta_{n}} \left(\frac{\kappa_{1}\kappa_{0}}{2\lambda_{0}\lambda_{1}}\zeta_{n}\right)^{2} \psi_{n}(t)\psi_{n}(s) \right|, \quad t,s\in[0,T].\\
\end{aligned}
\end{equation*}
Then, it follows from the proof of \cite[Theorem 4.27]{Porter1990} that $\mathfrak{R}_{N}$ converges uniformly to the sum in $\mathcal{R}$ on $(t,s)\in[0,T]^{2}$ ,
and that $ \mathfrak{R}_{N}^{\text{abs}}$ is uniformly convergent. 
The uniform convergence of $\mathfrak{R}_{N}^{\text{abs}}$ guarantees that the uniform convergence of $\mathfrak{R}_{N}$ is preserved even when the order of summation is changed.  Therefore, as it is natural to expect, the solution $\nu^{0,*}$ in \eqref{eq-major-solution} is independent of how one enumerates the eigenvalues and the corresponding eigenfunctions of $\mathsf{G}$.
\end{remark}

\section{Illustrations}\label{sec-illustrations}
In this section we illustrate the agents’ optimal equilibrium strategies, which were derived in Theorems \ref{thm-major-solution} and \ref{thm-minor-optimal-strategy}. Motivated by Section 4 of \cite{Lehalle-Neum18}, we consider the case where the signal $\mu$ in \eqref{eq-mu-square-integrable} follows an Ornstein-Uhlenbeck process, 
\begin{equation}
\label{eq-ou-dynamics}
d\mu_{t}= -\beta \mu_{t}dt + \sigma dW_{t}, \quad \mu_{0}=m_{0} ,
\end{equation}
where $W=(W_{t}) _{t \geq 0}$ is a standard Brownian motion and $\beta$ and $\sigma$ are positive constants.  Furthermore, we assume that $M$ in \eqref{ass:P} is given by  
\be \label{mart-num} 
M_t= M_0 +\sigma_{M}\wt W_t, 
\ee
where $\wt W$ is a standard Brownian motion independent from $W$, and $ M_0, \sigma_0$ are positive constants. We fix the values of the price impact parameters $\lambda_i, \kappa_i$, the initial inventory of the major agent $q_0$ and the terminal penalty parameter $\alpha$ in \eqref{eq-minor-functional} to be
\begin{equation} \label{param1} 
 \kappa_{0}=2, \quad \kappa_{1}= 2, \quad\lambda_{0}=1, \quad \lambda_{1}=1,\quad q_0=10, \quad \alpha =10, 
 \end{equation}
as well as the parameters of $M$ in \eqref{mart-num} and of $\mu$ in \eqref{eq-ou-dynamics},  
\begin{equation} \label{param2}
m_{0} =-0.5,  \quad \beta = 0.1,\quad \sigma =4, \quad  M_0=100, \quad  \sigma_{M} =1.
\end{equation}

The plots in this section are generated by using the numerical scheme which will described in detail in Section \ref{sec-numerics}. We choose as a complete orthonormal basis $(a_{i})_{i=1}^{\infty}$ the functions 
\begin{equation}
\label{eq-orthonormal-basis-l2}
 a_{i}(t) := \begin{cases}
1/\sqrt{T} \quad &i=1\\ 
\sqrt{2/T}\cos\left(\frac{(i-1)\pi  t}{T}\right) \quad &i=2,3,\ldots
 \end{cases}
\end{equation}
and such that each of corresponding degenerate kernel $\mathcal{G}_{n}$ defined in  \eqref{eq-convergence-kernel-series-G} represent the $n^{\text{th}}$-degree Fourier series approximation of the kernel $\mathcal{G}$ in \eqref{eq-kernel-caligraphic}. In order to strike a balance between numerical accuracy and computational efficiency, our simulations are generated by approximating the kernel $\mathcal{G}$ with the degenerate kernel $\mathcal{G}_{300}$.

The time dependence in the minor agent's inventory costs $\phi^{1}_t$ (see \eqref{eq-minor-functional}) can accommodate the setting of a liquidation carried out over several days. We take $T=k \tau$ for some positive integer $k$ and for $\tau>0$. Moreover, we choose the function $\phi^{1}$ to be given by the following parametric form
\begin{equation}
\label{eq-periodic-phi1}
\phi^{1}_{t}  = 
c_{0}\left(\frac{t}{\tau}-\left\lfloor \frac{t}{\tau} \right\rfloor\right)^{c_{1}}, \quad 0\leq t \leq T, 
\end{equation}
for two positive constants $c_{0}$ and $c_{1}$, which in the context of our simulations, we take to be $c_{0}=500$ and $c_{1}=15$. The function \eqref{eq-periodic-phi1} is periodic of period $\tau$ and increases to its maximum value as $t$ approaches $\tau, \ 2\tau, \ 3\tau, \ldots,\ k\tau$ from the left, forcing the minor agent to liquidate most of her position at the end of each period. We consider a liquidation carried over a business week, from Monday to Friday, such that $T=5$ (days) and $\tau=1$ (day). Figure \ref{fig-multiday-liquidation-trades} illustrates three examples of a multi-day liquidation. Specifically,  the top panel shows the major agent's deterministic optimal inventory (green line), deduced from \eqref{eq-major-solution}, as well as three different realisations of the minor agent's optimal inventories that one obtains from \eqref{eq-minor-agent-optimal} (blue, purple and red lines). The bottom panel shows the corresponding signal $\mu$ observed by the minor agent's while adopting the strategies at the top panel.

\begin{figure}[hbtp] 
\includegraphics[width=0.9\textwidth]{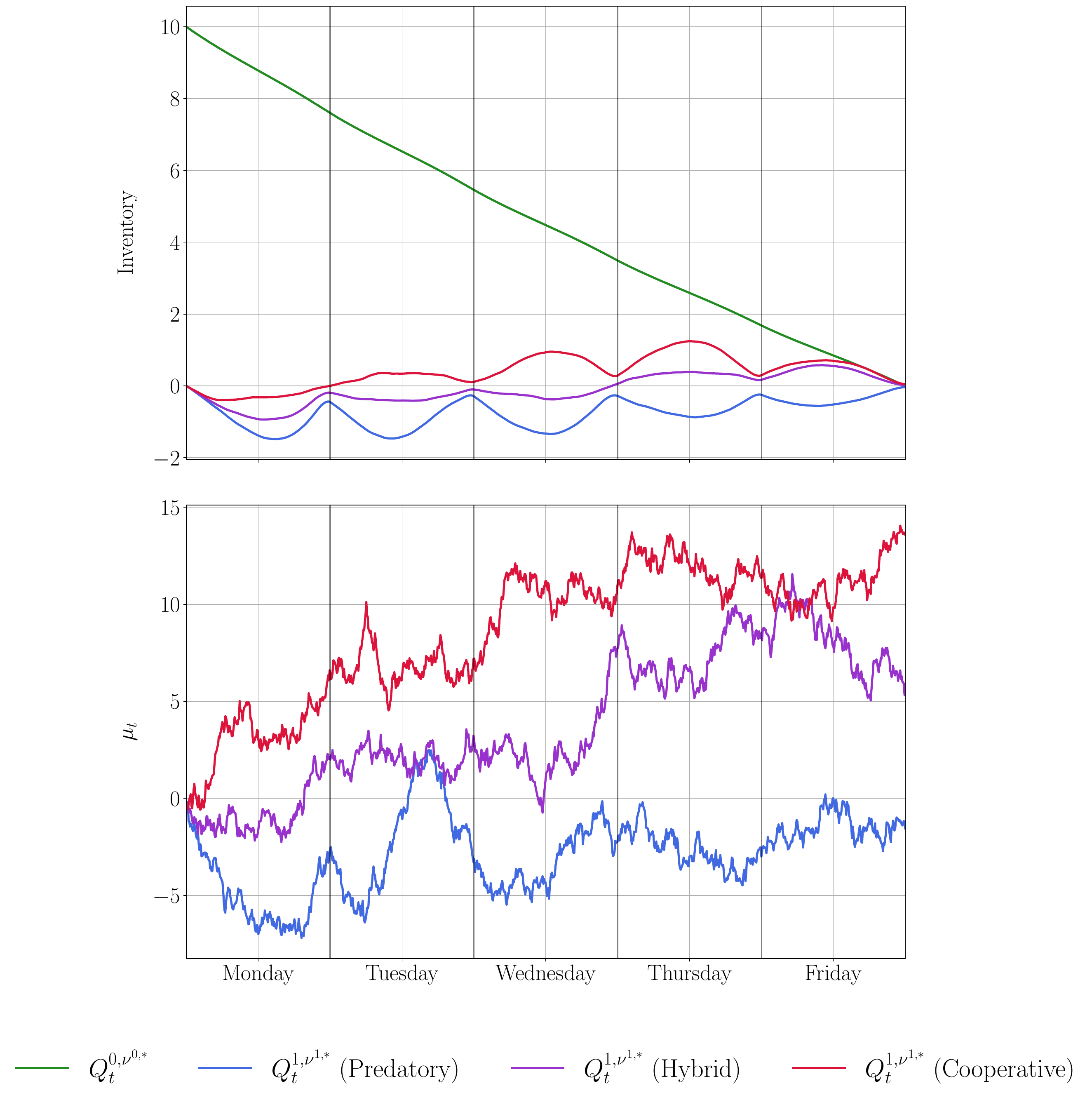}
\centering
\caption{In the top panel, the green line represent the major agent's optimal inventory while the remaining solid lines represent the minor agent's optimal inventory when the minor agent is adopting a predatory trading style (blue line) and cooperative trading style (red line) or an hybrid of both (purple line). In the bottom panel, we show the signal $\mu_{t}$ corresponding to the realisations of the minor agent's inventories in the top panel. }\label{fig-multiday-liquidation-trades}
\end{figure}
From \eqref{eq-P-nu-definition} it follows that the price impact generated by the major agent's  optimal strategy $\nu^{0,*}$ is perceived as a deterministic signal. The sell-off of shares by the major agent has the effect of pushing the price downwards, therefore, it generates opportunities which can be exploited by the minor agent. These considerations justify the fact that, as shown in \eqref{eq-minor-agent-optimal}, the minor agent adopts a trading strategy which tracks the ``impacted'' signal $\mu_{t} -\kappa_{0}\nu^{0,*}_{t}$ instead of the raw market signal $\mu_{t}$.  Hence, depending on her forecast on the impacted signal $\mu_{t} -\kappa_{0}\nu^{0,*}_{t}$, during each period the minor agent can decide whether to trade in the same direction of the major agent or not.  This has the effect that, over the interval $[0,T]$, the observed trading style of the minor agent can be \textit{predatory}, i.e. front running the major agent (blue line), \textit{cooperative} (red line) or a \textit{hybrid} of both (purple line).


To further understand several novel features of our model in the context of the multi-day liquidation we have just analysed, it is instructive to momentarily pause our discussion and consider the simpler case of a liquidation carried out over a single day. In particular, we wish to  benchmark the major agent's optimal strategy in \eqref{eq-major-solution} against the strategy $\nu^{0,\text{BM}}$ the major agent would use if she were unaware of the minor agent's trading activity. The strategy $\nu^{0,\text{BM}}$ can be found by solving the major agent's problem with $\kappa_{1}=0$ and it is given by
\begin{equation}
\label{eq-twap-strategy}
\nu^{0,\text{BM}}_{t} =\frac{q_0}{T}  + \frac{m_{0}}{2\lambda_{0}\beta} \left(\frac{1- \beta T e^{-\beta t}-e^{-\beta T}}{\beta T} \right), \quad 0\leq t \leq T. 
\end{equation}
Note that in the case of $m_{0} = 0$ in \eqref{eq-ou-dynamics}, $\nu^{0,\text{BM}}$ in \eqref{eq-twap-strategy} is a TWAP strategy. 

We assume that the major agent wishes to liquidate his initial position over a time 
horizon of six hours, from 10 AM to 4 PM, hence we set $T=6$ (hours).

In the present context, we slightly modify some of the parameters in \eqref{param1} and \eqref{param2}: $\sigma=1.5$, $\alpha=50$ and $\phi^{1}_{t} \equiv1$.
The top panel of Figure \ref{fig-benchmark-comparison-strat} shows the major agent's optimal trading rate $\nu^{0,*}$ (solid green line) and the benchmark trading rate $\nu^{0,\mathrm{BM}}$ (dashed green line). The bottom panel show 1000 realisations of minor agent's optimal trading rate $\nu^{1,*}$ (thin solid orange lines) and the cross-sectional average (thick solid brown line). 
We observe that the major agent's optimal strategy visibly deviate from the benchmark one in order to take into account the adverse effect of the minor agent's trading activity. We remark that since the major agent adopts a deterministic strategy, her decisions are based on the cross-sectional average of the minor agent's strategy, i.e. the solid brown line in the bottom panel of Figure \ref{fig-benchmark-comparison-strat}.   Initially, it is optimal to trade faster than the benchmark strategy in anticipation of the expected permanent price impact generated by the minor agent's reaction. Indeed, the early prices are more favourable to the major agent since they have not been affected yet by the extra price impact generated by presence of the minor agent. In the middle of the trading window the major agent's keeps trading but at a lower rate than the benchmark strategy. The explanation for this is that the major agent is aware that the minor agent could potentially trade in the same direction. Therefore, slowing down partially minimise the negative externality the minor agent's exerts on her via the aggregated permanent price impact. Finally, in the last section of the trading window two factors determine the behaviour of the major agent's optimal strategy. First,  the major agent must increase her trading rate to meet the terminal inventory constraint $Q_{T}^{0,\nu^{0,*}}=0$. Second, the major agent is aware that, on average, the minor agent will have to close her short position at the end of the time horizon, therefore, she will have to buy shares, generating market impact and pushing the price up again. Hence, the prices at the end of the trading window are more favourable for the major agent, therefore, a substantial portion of the liquidation is postponed to the last hour. 
Figure \ref{fig-benchmark-comparison-inventory} presents the major agent's and minor agent's inventories corresponding to the trading rates depicted in Figure \ref{fig-benchmark-comparison-strat}.

\begin{figure}[hbtp]
\includegraphics[width=0.75\textwidth]{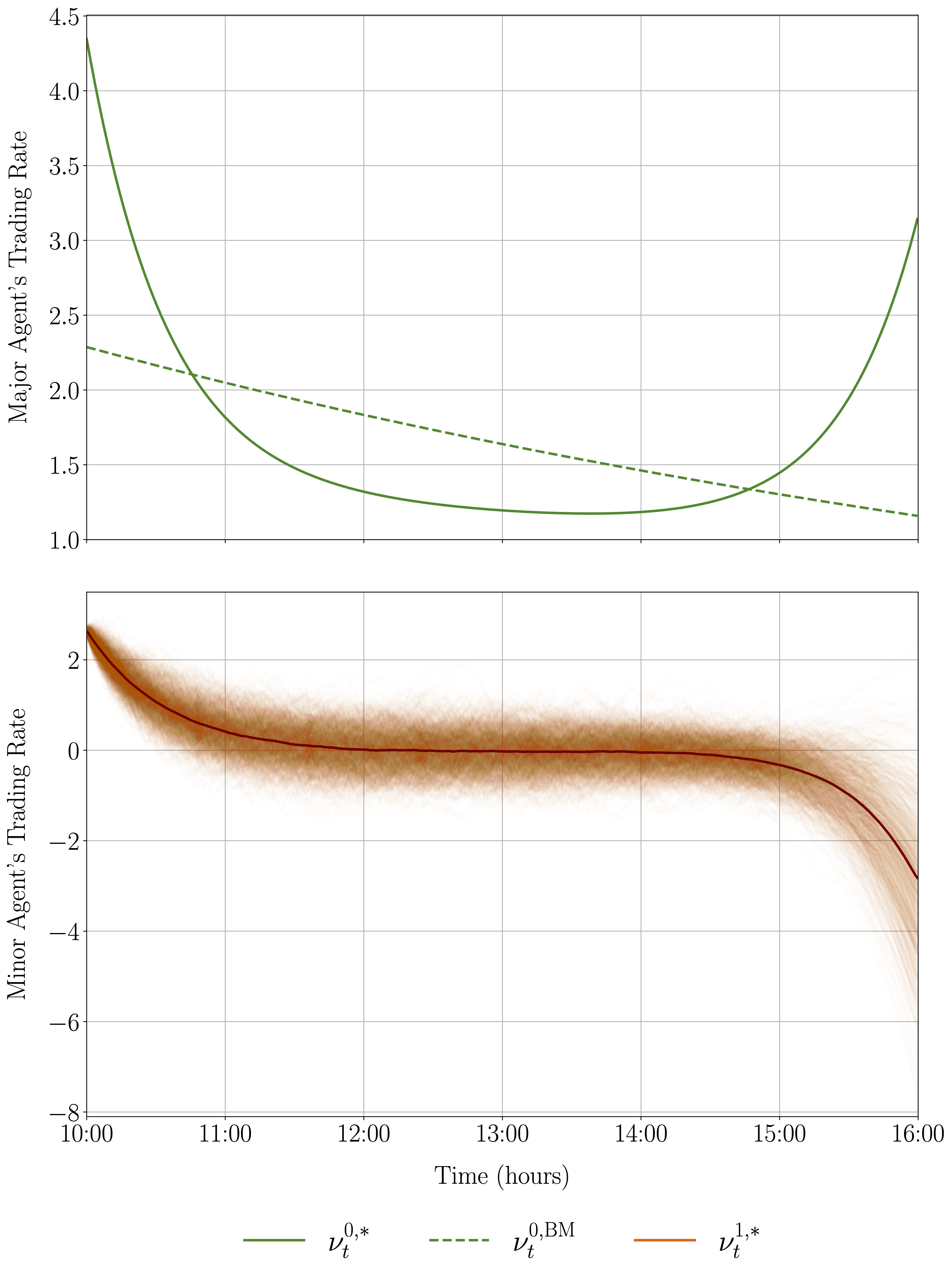}
\centering
\caption{ The major agent's and minor agent's optimal strategies in \eqref{eq-major-solution} and  \eqref{eq-minor-agent-optimal}, respectively, for a single-day liquidation. In the top panel, the green solid line shows the major agent's optimal strategy while the dashed green line shows the benchmark strategy of \eqref{eq-twap-strategy}. In the bottom panel, the thin orange solid lines depicts different realisations of the minor agent's optimal strategy. The brown solid line is the cross-sectional mean over the realisations.}
\label{fig-benchmark-comparison-strat}
\end{figure}

\begin{figure}[hbtp]
\includegraphics[width=0.75\textwidth]{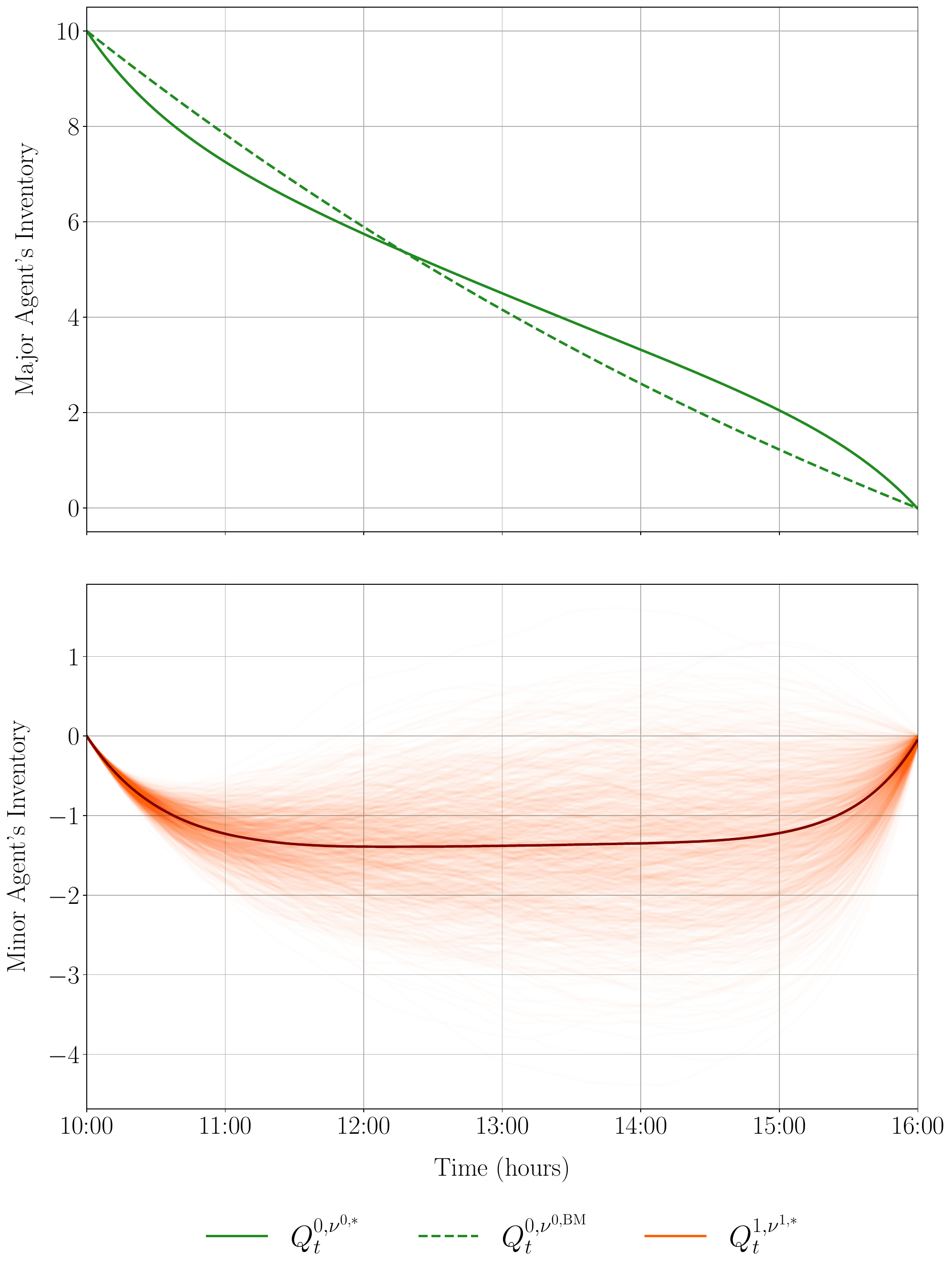}
\centering
\caption{The major agent's and minor agent's optimal inventories corresponding to the strategies in \eqref{eq-major-solution} and \eqref{eq-minor-agent-optimal}, respectively, for a single-day liquidation. In the top panel, the solid green line shows the major agent's optimal inventory corresponding to the major agent's optimal strategy  while the dashed green line shows the inventory corresponding to benchmark strategy in \eqref{eq-twap-strategy}. In the bottom panel, the thin solid orange lines represent different realisations of the minor agent's optimal inventories corresponding to the strategy  in \eqref{eq-minor-agent-optimal} while the solid brown line is the cross-sectional mean over the realisations.}
\label{fig-benchmark-comparison-inventory}
\end{figure}

Having established the major agent's and minor agent's trading patterns in the context of a single day liquidation, we turn again to the case of the multi-day liquidation presented in Figure \ref{fig-multiday-liquidation-trades}.  Analogously to Figure \ref{fig-benchmark-comparison-strat}, the top panel of Figure \ref{fig-multiday-comparison-strat} shows the major agent's optimal trading rate (solid green line) as well the trading rate of the benchmark strategy (dashed green line) in the context of the multiday-liquidation initially presented in Figure \ref{fig-multiday-liquidation-trades}.  Moreover, the bottom panel of Figure \ref{fig-multiday-comparison-strat} presents 1000 realisations of the minor agent's trading rates (thin orange lines) as well as the cross-sectional average (solid brown line). We recover analogous trading patterns to the one observed in the single-day liquidation: the major agent's speed, when compared to the benchmark strategy, greatly increases at the beginning and at the end of each day. Moreover, on average, the minor agent acquires a short position at the beginning of each day, pushing the price down, and then, in order to meet her terminal inventory constraint at the end of each day, she pushes the price up again by buying shares.  Note from Figure \ref{fig-multiday-liquidation-trades} that the predatory, cooperative and hybrid strategies share some common features. First, at the end of each day all the strategies have a very small inventory. This is because, by introducing the periodic running inventory costs of \eqref{eq-periodic-phi1}, the minor agent is strongly discouraged to hold a non-zero position at the end of each day, independently of her forecast for the impacted signal $\mu_{t} -\kappa_{0}\nu^{0,*}_{t}$. Secondly, from Figure \ref{fig-multiday-liquidation-trades} we observe that that the major agent is not liquidating at constant speed. Indeed, over the first day she liquidates at a speed visibly larger than, for example, the one employed over the last day. Such an intense liquidation in the first day generates an equally large alpha-signal through the corresponding price impact term $\kappa_{0}\nu^{0,*}_{t}$.  In the first day, the market impact-generated signal $\kappa_{0}\nu^{0,*}_{t}$ is large enough to outweigh any realistic realisation of the exogenous signal $\mu_{t}$, therefore, pushing the minor agent to trade in the same direction of the major agent, independently of the trading style she will adopt later on in the remaining days. 

\begin{figure}[hbtp]
\includegraphics[width=0.75\textwidth]{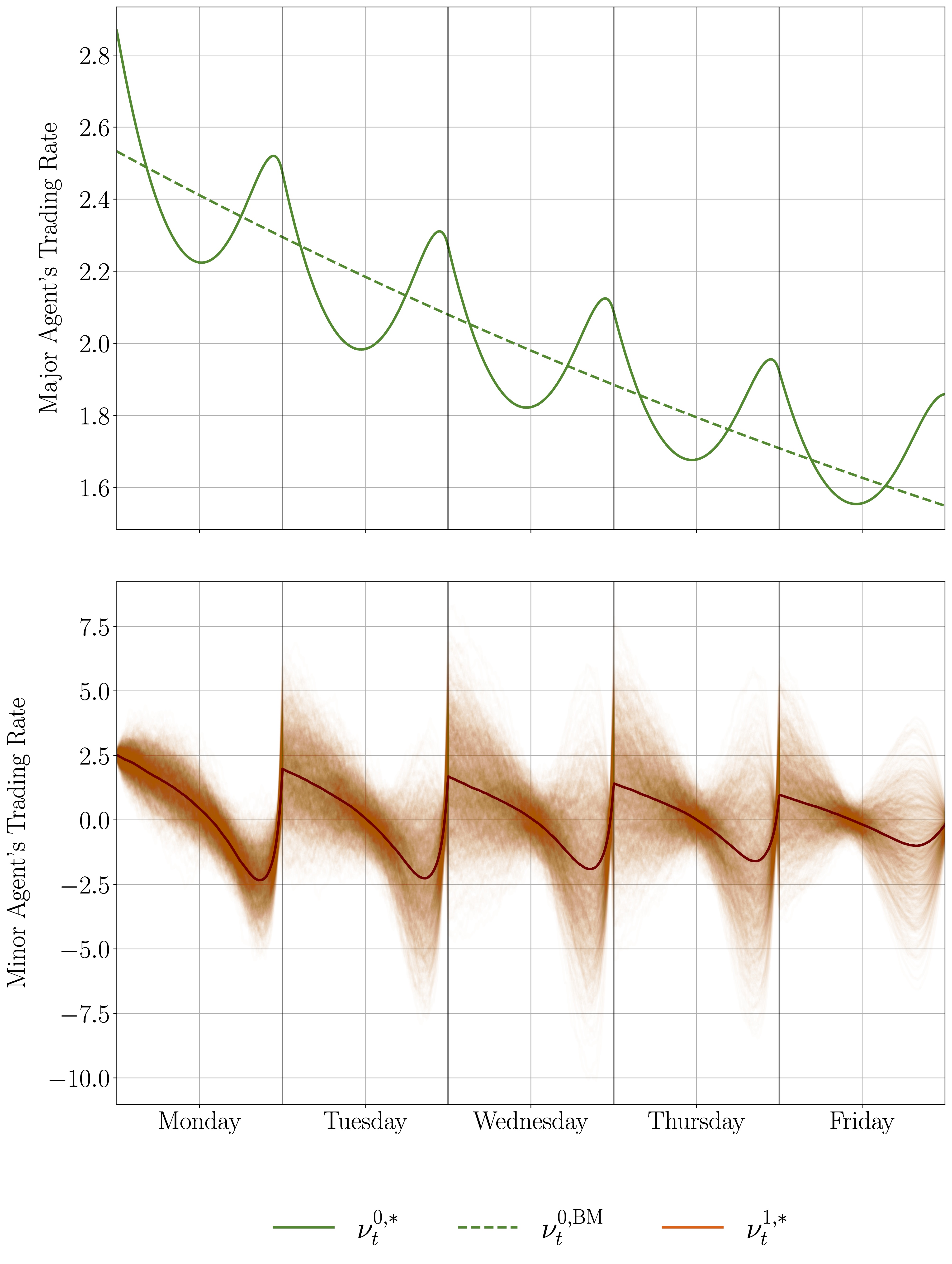}
\centering
\caption{ We present the major agent's and minor agent's optimal strategy in \eqref{eq-major-solution} and  \eqref{eq-minor-agent-optimal}, respectively, for a multi-day liquidation. In the top panel, the green solid line shows the major agent's optimal strategy while the dashed green line shows the benchmark strategy of \eqref{eq-twap-strategy}. In the bottom panel, the thin orange solid lines depicts different realisations of the minor agent's optimal strategy while the brown solid line represents the cross-sectional mean over the realisations.}
\label{fig-multiday-comparison-strat}
\end{figure}

It is of practical interest to compare the financial performance of the major agent's optimal strategy $\nu^{0,*}$ against those of the benchmark strategy $\nu^{0,\text{BM}}$. In the interest of brevity, we limit ourselves to the case of the single-day liquidation presented in Figure \ref{fig-benchmark-comparison-inventory}. In Figure \ref{fig-benchmark-savings} we present a histogram of the empirical probability distribution of the performance of the major agent's optimal strategy in \eqref{eq-major-solution} relative to the benchmark strategy in \eqref{eq-twap-strategy} generated using 1,000 simulations. We compare the profit-and-loss (PnL) of the strategies in basis points (bps) through the following formula:
\begin{equation}
\label{eq-savings-formula}
\frac{X^{0,\nu^{0,*}}_{T}-X^{0,\nu^{0,\text{BM}}}_{T}}{X^{0,\nu^{0,\text{BM}}}_{T}}\times 10^{4}, 
\end{equation}
where $X^{0,\nu^{0,\text{BM}}}_{T}$ is the terminal cash obtained from employing the benchmark strategy $\nu^{0,\text{BM}}$ and $X^{0,\nu^{0,*}}_{T}$ is the terminal cash obtained from employing the optimal strategy $\nu^{0,*}_{t}$. The mean of the distribution in Figure \ref{fig-benchmark-savings} is strictly positive, hence the major agent's optimal strategy on average outperforms the benchmark strategy.

\begin{figure}[t!]
\includegraphics[width=0.8\textwidth]{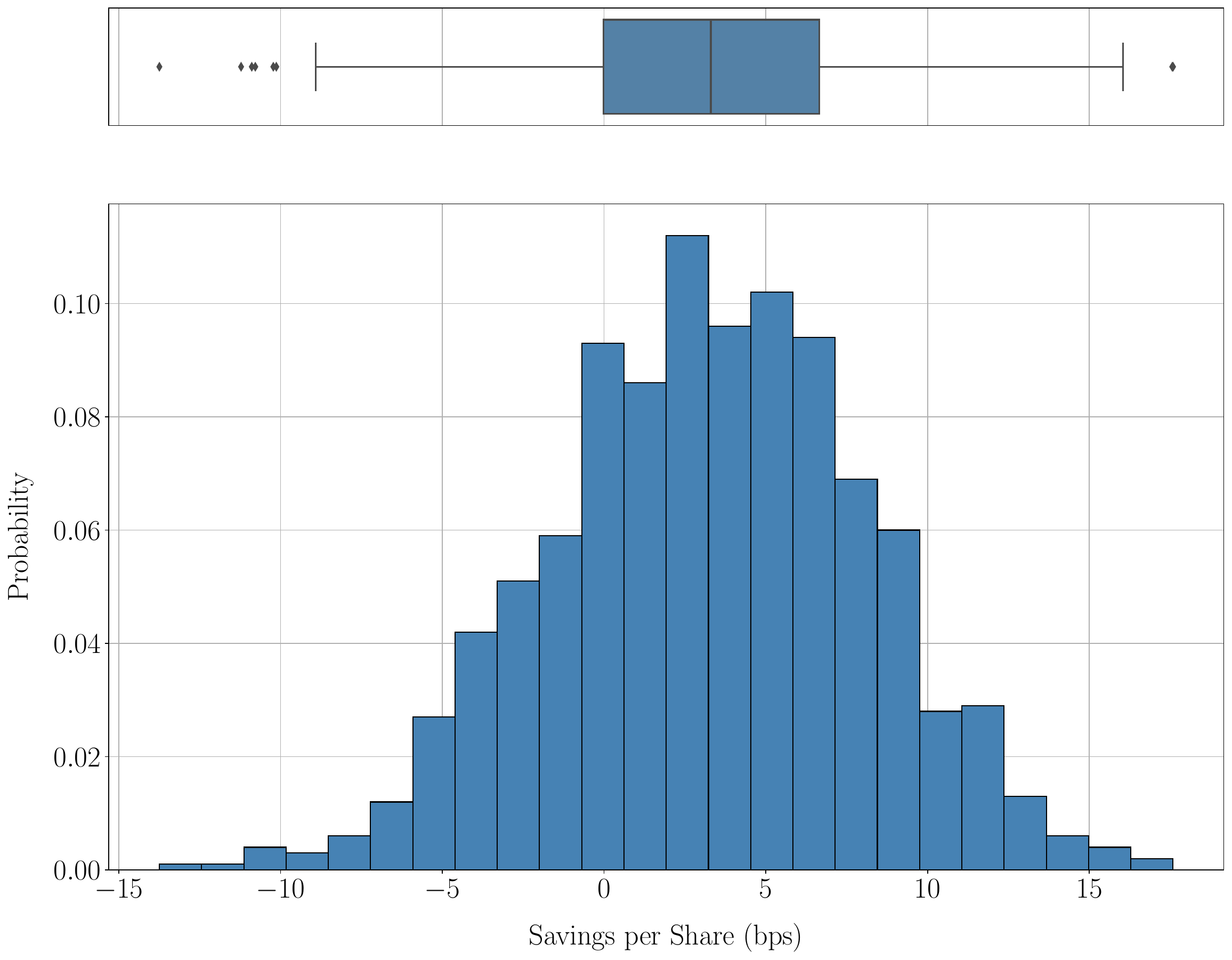}
\centering
\caption{The savings per share, computed using \eqref{eq-savings-formula}, measured in bps from 
 following the major agent's optimal strategy relative to the benchmark strategy in \eqref{eq-twap-strategy}. The top panel shows the box-plot corresponding to the distribution in the bottom panel.}
 \label{fig-benchmark-savings}
\end{figure}

Finally, we show that the major agent's and minor agent's trading behaviour induces noteworthy patterns in the intraday volume. A well-known empirical pattern of intraday volume is that it follows a U-shaped curve, where the traded volume peaks at the beginning of the day and at the end of the day, see for example \cite[Chapter 4, Figure 4.2]{Cartea-Book}. In Figure \ref{fig-volume-curve-daily} we present the intraday volume curve implied by the major agent's and minor agent's trading behaviour. Specifically, we simulate 1000 realisation of the minor agent's trading strategy and for each realisation we consider the absolute number of shares traded by the minor agent and the major agent over 1 minute bins from 10 AM to 4 PM. We call this quantity the ``volume'' traded in each minute bin. Then, we compute the natural logarithm of $1+ \mathrm{volume}$, where adding $1$ allows to consider assets whose traded volume is a fraction of a share. Our procedure is completely analogous to the one described in \cite{Cartea-Book}. 
Each blue line in Figure \ref{fig-volume-curve-daily} represents a realisation of the log-volume, while the magenta line is the median value of each 1-minute bin. The volume curves of Figure \ref{fig-volume-curve-daily} visibly present a U-shaped pattern analogous, for example, to those empirically observed and reported in \cite{Cartea-Book}.

\begin{figure}[t!]
\includegraphics[width=0.9\textwidth]{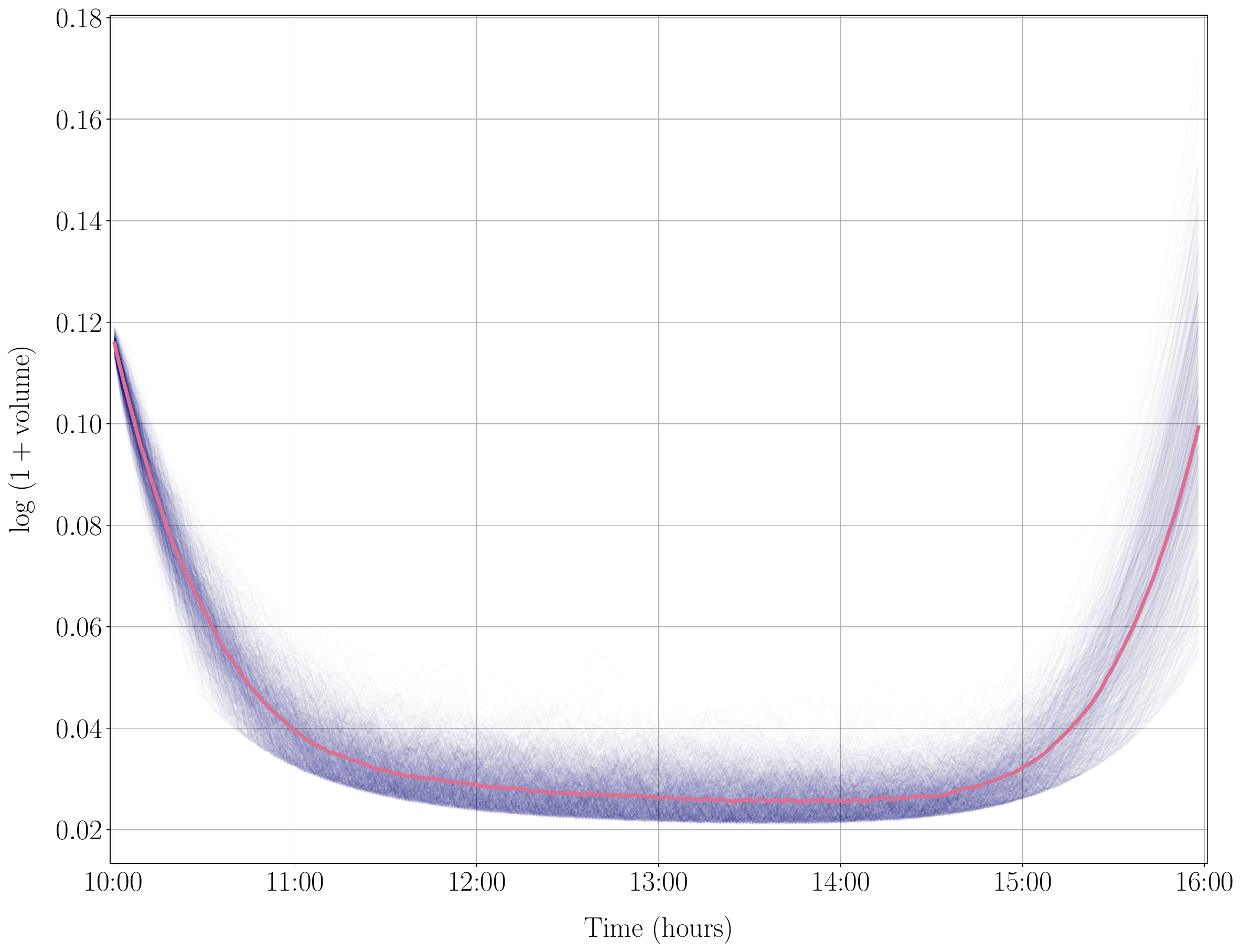}
\centering
\caption{Intraday volume curve as a function of the time of day. Realisations of the volume curves are represented by blue lines, while the cross-sectional median is drawn in pink.}
 \label{fig-volume-curve-daily}
\end{figure}

\section{Numerical Scheme}
\label{sec-numerics}
Theorem \ref{thm-major-solution} presents the unique major agent optimal strategy $\nu^{0,*}$ in closed-form. The optimal strategy $\nu^{0,*}$ is expressed in terms of the resolvent operator $\mathsf{R}$, defined in \eqref{eq-resolvent-kernel-def}, which in turn relies on the eigenvalues $(\zeta_{n})_{n\geq 1}$ and eigenfunctions $(\psi_{n})_{n\geq 1}$ of the operator $\mathsf{G}$ defined in \eqref{eq-G-rep-integral}.  In several simple cases these eigenfunctions and eigenvalues can be computed explicitly. 

For example, in the case of $\phi^{1} \equiv 0$, $(\zeta_{n})_{n\geq 1}$ and $(\psi_{n})_{n\geq 1}$ can be explicitly determined in terms of the roots of a transcendental equation (see Appendix \ref{sec-spectral-explicit}).  Nevertheless, a closed-form representation for the eigenvalues $(\zeta_{n})_{n\geq 1}$ and the eigenfunctions $(\psi_{n})_{n\geq 1}$ might be unattainable when $\phi^{1}$ is a generic non-negative piecewise continuous function. Therefore, we dedicate this section to developing a numerical scheme to compute the major agent optimal strategy $\nu^{0,*}$ which fully bypass the need of determining these eigenvalues and eigenfunctions.  As a by product, such numerical scheme will also determine the Stackelberg equilibrium of Corollary \ref{corollary-stackelberg}. 

We denote by $\mathsf{I}$ the identity operator on $L^{2}([0,T])$, that is 
\be \label{id} 
(\mathsf{I}\psi)(t)=\psi(t) \quad  \textrm{for all } 0\leq t \leq T, \  \psi \in L^{2}([0,T]). 
\ee 
As usual, the resolvent operator in \eqref{eq-resolvent-operator} can be written as 
$$\mathsf{R} = \left(\mathsf{I} + \frac{\kappa_{1}\kappa_{0}}{2\lambda_{0}\lambda_{1}}\mathsf{G}\right)^{-1}.$$
This is proved rigorously in Proposition \ref{prop-properties-resolvent}.  
It follows that the major agent optimal strategy $\nu^{0,*}$ in \eqref{eq-major-solution} satisfies to the following integral operator equation
\begin{equation}
\label{eq-operator-strategy-major}
\left(\mathsf{I} + \frac{\kappa_{1}\kappa_{0}}{2\lambda_{0}\lambda_{1}}\mathsf{G}\right) \nu^{0,*} = \mathsf{S}\bar{\mu} + \frac{\eta}{2\lambda_{0}}, 
\end{equation}
 Here, the constant $\eta$ is defined as in \eqref{eq-constant-major-solution} and the operators $\mathsf{G}$ and $\mathsf{S}$ are defined as in \eqref{eq-G-rep-integral} and \eqref{eq-chi-function}, respectively. The idea is to replace \eqref{eq-operator-strategy-major} with a sequence of approximate equations (see \eqref{eq-operator-strategy-major-approx}) whose solutions converge to the desired optimal strategy $\nu^{0,*}$.

For the discussion that follows it is convenient to recall the definition of a finite-rank operator and that of a compact operator, which we will provide in Definition \ref{def-compact-operator}. In Proposition \ref{prop-G-adjoint-compact} we will show that the operator $\mathsf{G}$ is compact. Therefore, there exists a sequence of finite-rank operator $(\mathsf{G}_{n})_{n\geq 1}$ in $B(L^{2}([0,T]))$ satisfying the approximation property
\begin{equation*}
\lim_{n\to\infty} ||\mathsf{G}_{n}-\mathsf{G}||=0, 
\end{equation*}
where $\|\cdot\|$ refers to the operator norm in \eqref{eq-operator-norm-definition}. 

In order to construct such sequence, we consider a complete orthonormal basis $(a_{i})_{i=1}^{\infty}$ in $L^{2}([0,T])$. A possible choice of such complete orthonormal basis in $L^{2}([0,T])$ is given by \eqref{eq-orthonormal-basis-l2}. 

Let the kernel $\mathcal{G}$ be defined as in \eqref{eq-kernel-caligraphic} and let the functions $(b_{i})_{i=1}^{\infty}$ be defined as 
\begin{equation}\label{eq-bj-integral}
b_{i}(t) := \int^{T}_{0}\mathcal{G}(t,s)a_{i}(s)ds. 
\end{equation}
We recall the definition of a degenerate kernel \cite[Definition 3.1]{Porter1990} which will be useful in the following. 
\begin{definition}[Degenerate Kernel]
\label{def-degenerate-kernel}
Let $n\geq 1$ and suppose there are finitely many functions $(a_{i})_{i=1}^{n}$ and $(b_{i})_{i=1}^{n}$ such that $a_{i}:[0,T]\to\mathbb{R}$ and $b_{i}:[0,T]\to\mathbb{R}$ for $i=1,\ldots,n$. Assume further that $\mathcal{T}:[0,T]^{2}\to \mathbb{R}$ is a kernel such that
\begin{equation*}
\mathcal{T}(t,s) = \sum_{i=1}^{n} a_{i}(t)b_{i}(s), \quad t,s\in[0,T].
\end{equation*}
Then, the kernel $\mathcal{T}$ is said to be \emph{degenerate}.
\end{definition}

Define the sequence of degenerate kernels $(\mathcal{G}_{n})_{n\geq 1}$ as the partial sums
\begin{equation}
\label{eq-convergence-kernel-series-G}
\mathcal{G}_{n}(t,s) := \sum_{i=1}^{n}a_{i}(t)b_{i}(s),\quad n\geq 1.
\end{equation}
Since $\mathcal{G}$ is a kernel in $L^{2}([0,T]^2)$ (see Proposition \ref{prop-kernel-caligraphic-properties})  
then, as shown in the proof of  \cite[Theorem 3.4]{Porter1990},   the sequence $(\mathcal{G}_{n})_{n\geq 1}$ converges to $\mathcal{G}$ in the sense
\begin{equation}
\label{eq-convergence-degenerate-kernel}
\lim_{n\to\infty}\int^{T}_{0}\int^{T}_{0}(\mathcal{G}(t,s)-\mathcal{G}_{n}(t,s))^{2}dsdt =0.
\end{equation}
Given the degenerate kernels $(\mathcal{G}_{n})_{n\geq 1}$ we can define a corresponding sequence of so-called finite rank integral operators $(\mathsf{G}_{n})_{n\geq 1}$ as
\begin{equation}
\label{eq-operator-degenerate-kernel}
(\mathsf{G}_{n}\psi)(t) := \int^{T}_{0}\mathcal{G}_{n}(t,s) \psi(s)ds, \quad \psi\in L^{2}([0,T]). 
\end{equation}
The following proposition, which is proved in Section \ref{proof-convergence-numerics}, gives the convergence result for the sequence $(\mathsf{G}_{n})_{n\geq 1}$. 
\begin{proposition}
\label{prop-convergence-Gn}
Under Assumption \ref{ass-alpha-lambda}, let $(\mathsf{G}_{n})_{n\geq 1}$ be defined as in \eqref{eq-operator-degenerate-kernel} and let $\mathsf{G}$ be defined as in \eqref{eq-G-rep-integral}. Then the finite rank operators $\mathsf{G}_{n}$ are in $B(L^{2}([0,T]))$. Moreover, we have that
\begin{equation}
\lim_{n\to\infty} ||\mathsf{G}_{n}-\mathsf{G}||=0.
\end{equation}
\end{proposition}
Next, we consider the following sequence of approximate equations to \eqref{eq-operator-strategy-major}, 
\begin{equation}
\label{eq-operator-strategy-major-approx}
\left(\mathsf{I} + \frac{\kappa_{1}\kappa_{0}}{2\lambda_{0}\lambda_{1}}\mathsf{G}_{n}\right) \nu^{0,(n)} = \mathsf{S}\bar{\mu} + \frac{\eta_{n}}{2\lambda_{0}},\quad n\geq 1,
\end{equation}
for a suitably defined sequence of constants $(\eta_{n})_{n\geq 1}$. 
\begin{remark}
We remark that in \eqref{eq-operator-strategy-major-approx} we continue to take the operator $\mathsf{S}$ to be defined in terms of $\mathsf{G}$, as in \eqref{eq-chi-function}, and not in terms of the sequence $(\mathsf{G}_{n})_{n\geq 1}$. It is not necessary to approximate the operator $\mathsf{S}$ since it can be explicitly expressed in terms of the kernel $\mathcal{G}$ in \eqref{eq-kernel-caligraphic}, via the operator $\mathsf{G}$, therefore, it can be computed explicitly via numerical integration (see also Remark \ref{remark-exact-numerics}).
\end{remark}
A solution to \eqref{eq-operator-strategy-major-approx} exists if the inverse of the operator $\mathsf{I} + \frac{\kappa_{1}\kappa_{0}}{2\lambda_{0}\lambda_{1}}\mathsf{G}_{n}$ exists, with the candidate solution $\nu^{0,(n)}$ being given by
\begin{equation}
\label{eq-integral-approximate-operator}
 \nu^{0,(n)} =\left(\mathsf{I} + \frac{\kappa_{1}\kappa_{0}}{2\lambda_{0}\lambda_{1}} \mathsf{G}_{n}\right)^{-1} \left( \mathsf{S}\bar{\mu} + \frac{\eta_{n}}{2\lambda_{0}}\right). 
\end{equation}
The next result shows that for sufficiently large $n$, the inverse of $\mathsf{I} + \frac{\kappa_{1}\kappa_{0}}{2\lambda_{0}\lambda_{1}}\mathsf{G}_{n}$ exists. Moreover, we show that the problem of finding such inverse is reduced to the finite dimensional problem of matrix inversion.

To state our results it is convenient to introduce the sequence of matrices $(\mathbb{G}_{n})_{n\geq 1}$ where  $\mathbb{G}_{n}\in\mathbb{R}^{n\times n}$ and whose entries are defined as
\begin{equation}
\label{eq-matrix-G-frak}
(\mathbb{G}_{n})_{ij} :=\langle a_{i}, b_{j}\rangle_{L^{2}}
\end{equation} 
for all $1\leq i,j\leq n$ and with $a_i$ and $b_j$ defined as  in \eqref{eq-bj-integral}. Moreover, we will denote by $I_{n}$ the $n$-dimensional identity matrix, that is $I_{n}:=\text{diag}(1,\ldots,1)\in\mathbb{R}^{n\times n}$. \\We are now ready to state our next proposition, which is proved in Section \ref{proof-convergence-numerics}.  

\begin{proposition}
\label{prop-degenerate-kernel-approximation}
Under Assumption \ref{ass-alpha-lambda}, let $(\mathsf{G}_{n})_{n\geq 1}$ be defined as in \eqref{eq-operator-degenerate-kernel} and let $(\mathbb{G}_{n})_{n\geq 1}$ be defined as in \eqref{eq-matrix-G-frak}. Then, there exists $N\geq 1$ such that for all $n\geq N$ the operator 
$\mathsf{I} + \frac{\kappa_{1}\kappa_{0}}{2\lambda_{0}\lambda_{1}}\mathsf{G}_{n}$ and the matrix  $I_{n} +  \frac{\kappa_{1}\kappa_{0}}{2\lambda_{0}\lambda_{1}} \mathbb{G}_{n}$ are both invertible. In particular, for all $n\geq N$ it holds that 
\begin{equation}
\label{eq-inverse-matrix-degenerate}
\left(\mathsf{I} + \frac{\kappa_{1}\kappa_{0}}{2\lambda_{0}\lambda_{1}}\mathsf{G}_{n}\right)^{-1}\psi = \psi -\frac{\kappa_{1}\kappa_{0}}{2\lambda_{0}\lambda_{1}} \sum_{i,j=1}^{n} \left(I_{n} +  \frac{\kappa_{1}\kappa_{0}}{2\lambda_{0}\lambda_{1}} \mathbb{G}_{n}\right)^{-1}_{i,j}\langle \psi,b_{j} \rangle_{L^{2}} \ a_{i}, 
\end{equation}
for any $\psi\in L^{2}([0,T])$.
\end{proposition}

Note that both operators $\mathsf{I} + \frac{\kappa_{1}\kappa_{0}}{2\lambda_{0}\lambda_{1}}\mathsf{G}$ and $\mathsf{I} + \frac{\kappa_{1}\kappa_{0}}{2\lambda_{0}\lambda_{1}}\mathsf{G}_{n}$ are invertible (see Proposition \ref{eq-R-inverse-operator}), nevertheless, only in the case of the latter the inverse operator can be computed via matrix inversion by exploiting the corresponding degenerate kernel decomposition, see also Remark \ref{remark-exact-numerics} for additional discussion. 
 
 The next result shows that that the candidate solutions in \eqref{eq-integral-approximate-operator} converge \textit{in mean} to the optimal strategy $\nu^{0,*}$ of Theorem \ref{thm-major-solution}.
Henceforth, we take the sequence of constants $(\eta_{n})_{n\geq 1}$ to be defined as
\begin{equation}
\label{eq-eta-n}
\eta_{n} := 2\lambda_{0}\frac{q_0- \left\langle \left( \mathsf{I} + \frac{\kappa_{1}\kappa_{0}}{2\lambda_{0}\lambda_{1}}\mathsf{G}_{n}\right)^{-1}\mathsf{S}\bar{\mu},\mathit{1}\right\rangle_{L^{2}}}{\left\langle \left( \mathsf{I} + \frac{\kappa_{1}\kappa_{0}}{2\lambda_{0}\lambda_{1}}\mathsf{G}_{n}\right)^{-1}\mathit{1},\mathit{1}\right\rangle_{L^{2}}}, \quad n\geq 1. 
\end{equation}
\begin{proposition}
\label{prop-convergence-numerics-mean}
Under Assumption \ref{ass-alpha-lambda}, let $\nu^{0,*}$ and $\nu^{0,(n)}$ be defined as in \eqref{eq-major-solution} and \eqref{eq-integral-approximate-operator}, respectively.
Then, there exists $N\geq 1$ such that for all $n\geq N$ the functions $\nu^{0,(n)}$ are well-defined and are in $L^{2}([0,T])$.
Moreover, 
\begin{equation}
\label{eq-error-bound-mean}
\lim_{n\to\infty }\left\|\nu^{0,*}-\nu^{0,(n)}\right\|_{L^{2}} = 0.
\end{equation}
\end{proposition}
The proof of Proposition \ref{prop-convergence-numerics-mean} is postponed to Section \ref{proof-convergence-numerics}. Lemma \ref{lemma-constant-convergence} shows that for sufficiently large $n$, the constants $\eta_{n}$ in \eqref{eq-eta-n} are well-defined.

In order to obtain an approximating sequence which converges \textit{uniformly} to the optimal control $\nu^{0,*}$ we introduce the sequence of candidate functions $(\hat{\nu}^{0,(n)})_{n\geq 1}$ defined as 
\begin{equation}
\label{eq-uniform-numerical-solution}
\hat{\nu}^{0,(n)}_{t}  :=  - \frac{\kappa_{1}\kappa_{0}}{2\lambda_{0}\lambda_{1}} \left(\mathsf{G} \nu^{0,(n)}\right)(t) +  (\mathsf{S}\bar{\mu})(t) + \frac{\eta_{n}}{2\lambda_{0}}
\end{equation}
for all $t\in[0,T]$ and for all $n\geq1$. 
Our main result for this section is the following convergence theorem.
\begin{theorem}
\label{thm-uniform-convergence-numerics}
Under Assumption \ref{ass-alpha-lambda}, let $\nu^{0,*}$, $\hat{\nu}^{0,(n)}$ and $\nu^{1,*}$ be defined as in \eqref{eq-major-solution}, \eqref{eq-uniform-numerical-solution} and \eqref{eq-minor-agent-optimal}, respectively.
Then, there exists an $N\geq 1$ such that for all $n\geq N$ the functions $\hat{\nu}^{0,(n)}$ are in $L^{2}([0,T])$ and the controls $\nu^{1,*}(\hat{\nu}^{0,(n)})$ are in $\mathcal{A}_{m}$. Furthermore, we have that:
\begin{itemize} 
\item[\textbf{(i)}] 
\begin{equation*}
 \lim_{n\to\infty}\sup_{t\in[0,T]} \left|\nu^{0,*}_{t}-\hat{\nu}^{0,(n)}_{t}\right| =0, 
 \end{equation*}
\item[\textbf{(ii)}] 
\begin{equation*}
 \lim_{n\to\infty}\sup_{t\in[0,T]} \left|\nu^{1,*}_{t}\left(\nu^{0,*}\right)-\nu^{1,*}_{t}\left(\hat{\nu}^{0,(n)}\right)\right| =0, \quad \mathbb{P}-\rm{a.s.} 
 \end{equation*}
 \end{itemize} 
\end{theorem}
The proof of Theorem \ref{thm-uniform-convergence-numerics} is postponed to Section \ref{proof-convergence-numerics}. 

Proposition \ref{prop-degenerate-kernel-approximation} and Theorem \ref{thm-uniform-convergence-numerics} show that the infinite-dimensional problem of determining the solution to \eqref{eq-operator-strategy-major}, can be reduced to the finite-dimensional problem of matrix inversion.
\begin{remark}
The proofs of the results of Proposition \ref{prop-convergence-numerics-mean} and Theorem \ref{thm-uniform-convergence-numerics} do not rely on the existence of the orthonormal expansion in \eqref{eq-convergence-kernel-series-G} and the corresponding convergence \eqref{eq-convergence-degenerate-kernel}. Indeed, our result can be extended to any generic sequence of operators $(\mathsf{G}_{n})_{n\geq 1}$ in $B(L^{2}([0,T]))$ satisfying the approximation property of Proposition \ref{prop-convergence-Gn} and which do not necessarily enjoy a integral representation of the form in \eqref{eq-operator-degenerate-kernel}. 
\end{remark}
\begin{remark}
\label{remark-exact-numerics}
The matrix entries in \eqref{eq-matrix-G-frak} must be computed numerically.  The use of a numerical evaluation in \eqref{eq-matrix-G-frak}  will lead to numerical errors in the entries of the matrix $I_{n} +  \frac{\kappa_{1}\kappa_{0}}{2\lambda_{0}\lambda_{1}} \mathbb{G}_{n}$. As shown in \cite[Chapter 2.3.4]{Atkinson1997}, for a sufficiently accurate estimation the numerical error arising from these computations is negligible. A similar discussion applies to, among others, the numerical evaluation of the inverse of the matrix $I_{n} +  \frac{\kappa_{1}\kappa_{0}}{2\lambda_{0}\lambda_{1}} \mathbb{G}_{n}$ and of the integral $\mathsf{S}\bar{\mu}$. These are all elementary and well-understood convergence problems in numerical analysis and the corresponding convergence rate could be easily incorporated in the convergence results of this section. Hence, our discussion assumes that the aforementioned quantities are taken to be exact and that the corresponding numerical errors are negligible. 
\end{remark}

\begin{remark}
The numerical scheme we have presented has an advantage from an implementation standpoint too. Specifically, if one were to determine the major agent's optimal strategy by using the result of Theorem \ref{thm-major-solution}, she would need to mathematically determine the eigenvalues $(\zeta_{n})_{n\geq 1}$ and eigenfunctions $(\psi_{n})_{n\geq 1}$ each time she wishes to change the function $\phi^{1}$, as shown, for example, in  Appendix \ref{sec-spectral-explicit}. On the other hand, with the numerical scheme of Theorem \ref{thm-uniform-convergence-numerics}, to achieve the same goal it is sufficient to change the expression of $\phi^{1}$ in the numerical solver of the Riccati ODE \eqref{eq-v1-ode}, which usually amounts to change solely few lines of code.
\end{remark}

\section{Proof of Theorem \ref{thm-minor-optimal-strategy}}
\label{sec-proof-thm-minor-optimal}
We show how the Stackelberg equilibrium can be found by \textit{backward induction}, that is by first solving the minor agent's problem and then the major agent's problem. We determine the minor agent's optimal strategy via a calculus of variations argument, as similarly done in \cite{NeumanVoss:20}. The following results also borrow ideas from  \cite{casgrain.jaimungal.19}.

Henceforth, we assume that $\nu^{0}\in \mathcal{A}_{M}^{q_0}$ is a fixed major agent liquidation strategy and with a slight abuse of notation we write $H^{1}(\nu)$ for $H^{1}(\nu;\nu^{0})$.
We start by determining an alternative representation for the minor agent's objective.
\begin{lemma}
\label{lemma-alternative-minor}
The minor agent's objective $H^{1}$ in \eqref{eq-minor-functional} can be alternatively represented as
\begin{equation}
\label{eq-alternative-minor}
\begin{aligned}
H^{1}(\nu^{1}) &= x_{1}- \mathbb{E}\Bigg[   \lambda_{1}\int^{T}_{0}(\nu^{1}_{t})^{2}dt +\alpha \left(Q^{1,\nu^{1}}_{T}\right)^{2} +\int^{T}_{0}\phi^{1}_{t} \left(Q^{1,\nu^{1}}_{t}\right)^{2}dt\\ &\quad\quad\quad\quad+\int^{T}_{0} Q^{1,\nu^{1}}_{t}(\kappa_{0}\nu^{0}_{t}dt + \kappa_{1}\nu^{1}_{t}dt -dA_{t})  \Bigg],
\end{aligned}
\end{equation}
for any $\nu^{1}\in\mathcal{A}_{m}$.
\end{lemma}
\begin{proof}
We use \eqref{eq-minor-cash}, \eqref{eq-execution-price-minor} and the It\^o's product rule on $Q^{1,\nu^{1}}_{T}P_{T}^{\nu}$  to get
\begin{equation}
\label{eq-integration-cash}
\begin{aligned}
& \mathbb{E}\left[X_{T}^{1,\nu^{1}} + Q^{1,\nu^{1}}_{T}P_{T}^{\nu} \right] \\
 &=  x_{1}  +\mathbb{E}\Bigg[ \int^{T}_{0} (P_{t}^{\nu}- \lambda_{1}\nu^{1}_{t})\nu^{1}_{t}dt +\int^{T}_{0} Q^{1,\nu^{1}}_{t}dP_{t}^{\nu} + \int^{T}_{0} P_{t}^{\nu}dQ^{1,\nu^{1}}_{t} \Bigg],
 \end{aligned}
\end{equation}
where  we also used $Q^{1,\nu^{1}}_{0}=0$ by \eqref{eq-minor-inventory}. Recall that $P =   M + A$. We apply \eqref{eq-minor-inventory},  \eqref{eq-P-nu-definition} and \eqref{eq-permanent-impact} to \eqref{eq-integration-cash} in order to obtain 
\begin{equation}
\label{eq-integration-cash-2}
\begin{aligned}
  \mathbb{E}\left[X_{T}^{1,\nu^{1}} + Q^{1,\nu^{1}}_{T}P_{T}^{\nu} \right]&=x_{1} + \mathbb{E}\Bigg[ - \lambda_{1}\int^{T}_{0} (\nu^{1}_{t})^{2}dt +\int^{T}_{0} Q^{1,\nu^{1}}_{t}dP_{t}^{\nu}\Bigg] \\
 &= x_{1} -  \mathbb{E}\Bigg[\lambda_{1}\int^{T}_{0} (\nu^{1}_{t})^{2}dt -\int^{T}_{0} Q^{1,\nu^{1}}_{t}dM_{t} \\ &\quad\quad\quad\quad+\int^{T}_{0} Q^{1,\nu^{1}}_{t}(\kappa_{0}\nu^{0}_{t}dt + \kappa_{1}\nu^{1}_{t}dt -dA_{t})\Bigg].
 \end{aligned}
\end{equation}
Since $\nu^{1}\in\mathcal{A}_{m}$ (see \eqref{eq-minor-admissible-set-def}), then we can drop the martingale term in \eqref{eq-integration-cash-2} and obtain
  \begin{equation}
 \label{eq-integration-cash-3}
 \begin{aligned}
  \mathbb{E}\left[X_{T}^{1,\nu^{1}} + Q^{1,\nu^{1}}_{T}P_{T}^{\nu} \right] 
 &= x_{1} -  \mathbb{E}\Bigg[\lambda_{1}\int^{T}_{0} (\nu^{1}_{t})^{2}dt  \\ &\quad\quad\quad\quad+\int^{T}_{0} Q^{1,\nu^{1}}_{t}(\kappa_{0}\nu^{0}_{t}dt + \kappa_{1}\nu^{1}_{t}dt -dA_{t})\Bigg].
 \end{aligned}
 \end{equation}
Substituting \eqref{eq-integration-cash-3} in \eqref{eq-minor-functional} returns \eqref{eq-alternative-minor}.
\end{proof}

In the next result we use the representation of  \eqref{eq-alternative-minor} to show that the minor agent's objective is strictly concave. 
\begin{lemma}
\label{lemma-minor-strictly-concave}
Under Assumption \ref{ass-alpha-lambda}, the functional $H^{1}$ defined in  \eqref{eq-minor-functional} is strictly concave for $\nu^1 \in \mathcal{A}_{m}$.
\end{lemma}
\begin{proof}
In order to prove that the functional $H^{1}$ is strictly concave, we must show that for any $0 < \rho <1$ and $\nu,\omega\in \mathcal{A}_{m}$, such that $\nu$ and $\omega$ are $d\mathbb{P}\otimes dt$ distinguishable, it holds that
\begin{equation}
\label{eq-concavity-minor}
\mathcal{I}^{1}(\rho,\nu,\omega) :=H^{1}(\rho\nu +(1-\rho)\omega)-\rho H^{1}(\nu)-(1-\rho)H^{1}(\omega) >0.
\end{equation}
It is convenient to introduce the constant $\theta := \frac{2\alpha -\kappa_{1}}{2}$ as well as to define the function $\Gamma^{1}$ as follows 
\begin{equation}
\label{eq-func-gamma-1}
\Gamma^{1}_{t} = \begin{pmatrix}
\lambda_{1} & -\theta \\
-\theta & \phi^{1}_{t}
\end{pmatrix}, \quad 0\leq t \leq T. 
\end{equation} 
Note that under Assumption \ref{ass-alpha-lambda} it holds that $\theta\geq 0$. From \eqref{eq-minor-inventory} and integration by parts we get
\begin{equation}
\label{eq-q1-squared-identity}
\left(Q^{1,\nu^{1}}_{T}\right)^{2} = -2\int^{T}_{0} Q^{1,\nu^{1}}_{t}\nu^{1}_{t}dt.
\end{equation}
Using \eqref{eq-q1-squared-identity} we rewrite the minor agent's objective in \eqref{eq-alternative-minor} in terms of the function $\Gamma^{1}$ as 
\begin{equation}
\label{eq-alternative-minor-2}
\begin{aligned}
H^{1}(\nu^{1}) &= x_{1}- \mathbb{E}\Bigg[  \int^{T}_{0}\begin{pmatrix}
\nu^{1}_{t}\\
Q^{1,\nu^{1}}_{t}
\end{pmatrix}^{\intercal} \Gamma^{1}_{t }\begin{pmatrix}
\nu^{1}_{t}\\
Q^{1,\nu^{1}}_{t}
\end{pmatrix}dt +\int^{T}_{0} Q^{1,\nu^{1}}_{t}(\kappa_{0}\nu^{0}_{t}dt -dA_{t})  \Bigg]. 
\end{aligned}
\end{equation}
Note that, given the representation in \eqref{eq-alternative-minor-2}, in the case of $\theta>0$ and $\phi^{1}_{t}>0$ for all $t\in[0,T]$ the strict concavity of $H^{1}(\nu^{1})$ follows from \eqref{eq-minor-inventory} and the fact that $\Gamma^{1}_{t}$ is a positive-definite matrix for all $t\in[0,T]$.  In what follows we will use \eqref{eq-alternative-minor-2} to show that $ H^{1}(\nu^{1}) $ is strictly concave also under the assumption that $\phi_t^{1} \geq 0$ and $\theta \geq 0$.

We observe that $Q^{1,\nu}$ is linear with respect to $\nu$, that is
\begin{equation*}
Q^{1,\rho\nu +(1-\rho)\omega}_{t}  = \rho Q^{1,\nu}_{t} +(1-\rho) Q^{1,\omega}_{t} \quad \textrm{for all }\rho\in[0,1], \ \nu,\omega\in\mathcal{A}_{m}.
\end{equation*}
We substitute \eqref{eq-alternative-minor-2} in \eqref{eq-concavity-minor} and we use the linearity of $Q^{1,\cdot}$ to cancel out the $Q^{1,\cdot}_{t}(\kappa_{0}\nu^{0}_{t}dt -dA_{t})$ terms. This yields
\begin{equation*}
\begin{aligned}
\mathcal{I}^{1}(\rho,\nu,\omega) &= \mathbb{E}\Bigg[\int^{T}_{0}\rho\begin{pmatrix}
\nu_{t}\\
Q^{1,\nu}_{t}
\end{pmatrix}^{\intercal} \Gamma^{1}_{t }\begin{pmatrix}
\nu_{t}\\
Q^{1,\nu}_{t}
\end{pmatrix} dt + (1-\rho)\begin{pmatrix}
\omega_{t}\\
Q^{1,\omega}_{t}
\end{pmatrix}^{\intercal} \Gamma^{1}_{t }\begin{pmatrix}
\omega_{t}\\
Q^{1,\omega}_{t}
\end{pmatrix} dt
\\ &\quad-\left( \rho\begin{pmatrix}
\nu_{t}\\
Q^{1,\nu}_{t}
\end{pmatrix} + (1-\rho)\begin{pmatrix}
\omega_{t}\\
Q^{1,\omega}_{t}
\end{pmatrix} \right)^{\intercal} \Gamma^{1}_{t }\left( \rho\begin{pmatrix}
\nu_{t}\\
Q^{1,\nu}_{t}
\end{pmatrix} + (1-\rho)\begin{pmatrix}
\omega_{t}\\
Q^{1,\omega}_{t}
\end{pmatrix} \right) dt\Bigg],
\end{aligned}
\end{equation*}
where after multiplying out all the terms we get 
\begin{multline}
\label{eq-left-hand side-concavity-minor}
\mathcal{I}^{1}(\rho,\nu,\omega)  = \\ \mathbb{E}\Bigg[\int^{T}_{0}\rho(1-\rho)\left(\begin{pmatrix}
\nu_{t}\\
Q^{1,\nu}_{t}
\end{pmatrix}- \begin{pmatrix}
\omega_{t}\\
Q^{1,\omega}_{t}
\end{pmatrix}\right)^{\intercal} \Gamma^{1}_{t }\left(\begin{pmatrix}
\nu_{t}\\
Q^{1,\nu}_{t}
\end{pmatrix}- \begin{pmatrix}
\omega_{t}\\
Q^{1,\omega}_{t}
\end{pmatrix}\right)dt \Bigg].
\end{multline}
It is convenient to introduce the function $\delta_{t} := \nu_{t}- \omega_{t}$ for all $t\in[0,T]$. From \eqref{eq-minor-inventory} it follows that $Q^{1,\delta}_{t} = Q^{1,\nu}_{t} - Q^{1,\omega}_{t}$ for all $t\in[0,T]$. We can rewrite \eqref{eq-left-hand side-concavity-minor} in terms of $\delta$ and $Q^{1,\delta}$ as
\begin{equation}
\begin{aligned}
\label{eq-left-hand side-concavity-minor-delta}
\mathcal{I}^{1}(\rho,\nu,\omega) &= \rho(1-\rho) \Big(\mathbb{E}\left[\int^{T}_{0} \lambda_{1} \delta_{t}^{2}dt\right] + \mathbb{E}\left[\int^{T}_{0} \phi^{1}_{t}\left(Q^{1,\delta}_{t} \right)^{2}dt\right]  \\ &\quad\quad\quad\quad\quad- \mathbb{E}\left[\int^{T}_{0}2\theta\delta_{t}Q^{1,\delta}_{t}  dt \right] \Big),
\end{aligned}
\end{equation}
where we have used \eqref{eq-func-gamma-1}.  

Since $\phi^{1}_{t}\geq0$ for all $t\in[0,T]$, we have
 $$\ \mathbb{E}\left[\int^{T}_{0} \phi^{1}_{t}\left(Q^{1,\delta}_{t} \right)^{2}dt\right] \geq 0.$$ 
 
 From \eqref{eq-minor-inventory} it holds that $Q^{1,\delta}_{t} = -\int^{t}_{0}\delta_{t}dt$, therefore by \eqref{eq-q1-squared-identity} we get
\begin{equation}
\label{eq-rightmost-term-minor}
-\mathbb{E}\left[\int^{T}_{0} 2\delta_{t}Q^{1,\delta}_{t} dt \right] =\mathbb{E}\left[\left(Q^{1,\delta}_{T} \right)^{2} \right] \geq 0.
\end{equation}
Finally, notice that since $\nu$ and $\omega$ are $d\mathbb{P}\otimes dt$ distinguishable then 
$$ \mathbb{E}\left[\int^{T}_{0}  \delta_{t}^{2}dt\right]>0.$$
This shows that $\mathcal{I}^{1}(\rho,\nu,\omega) >0$  for any $\theta \geq 0$, $0 < \rho <1$ and any $\nu,\omega\in \mathcal{A}_{m}$, such that $\nu$ and $\omega$ are $d\mathbb{P}\otimes dt$ distinguishable.
\end{proof}

As similarly show in \cite{NeumanVoss:20}, a probabilistic and convex analytic calculus of variations approach can be readily applied to derive a system of coupled linear FBSDEs which characterises the unique solution to the minor agent's problem.

Since under Assumption  \ref{ass-alpha-lambda} the map $\nu^{1}\to H^{1}(\nu^{1})$ in \eqref{eq-alternative-minor} is strictly concave, then it admits a unique maximiser characterised by the critical point at which the G\^ateaux derivative
\begin{equation}
\label{eq-def-minor-G\^ateaux}
\langle\mathcal{D}H^{1}(\nu^{1}),\omega\rangle := \lim_{\epsilon  \to 0} \frac{H^{1}(\nu^{1}+\epsilon\omega)-H^{1}(\nu^{1})}{\epsilon}
\end{equation}
vanishes.
In the following lemma we obtain an explicit expression for the G\^ateaux derivative of $H^{1}$.
\begin{lemma} \label{lemma-gat} 
The G\^ateaux derivative of  $H^{1}$ in \eqref{eq-alternative-minor}, in direction $\omega\in\mathcal{A}_{m}$ is given by
\begin{multline}
\label{eq-minor-G\^ateaux}
\left\langle\mathcal{D}H^{1}(\nu^{1}),\omega\right\rangle =\mathbb{E}\Bigg[ \int^{T}_{0}\omega_{t}\Big(-2\lambda_{1}\nu_{t}^{1} +2\alpha Q^{1,\nu^{1}}_{T} -\kappa_{1}Q^{1,\nu^{1}}_{t}  + A_{t}\\ 
+ \int^{T}_{t}\left(2\phi^{1}_{s}Q_{s}^{1,\nu^{1}} + \kappa_{0}\nu^{0}_{s}+ \kappa_{1}\nu_{s}^{1}\right)ds- A_{T}\Big)dt\Bigg],
\end{multline}
for any $\nu^{1}\in\mathcal{A}_{m}$.
\end{lemma}
The proof of Lemma \ref{lemma-gat} is given in Appendix \ref{sec-gat-pf}. 

From the explicit expression of the G\^ateaux derivative in \eqref{eq-minor-G\^ateaux}  we can derive a first order optimality condition. It takes the form of a coupled system of linear forward backward stochastic differential equations (FBSDE), as described in following lemma which is proved in Appendix \ref{sec-gat-pf}.

\begin{lemma} \label{lemma-fbsde-minor}
Under Assumption \ref{ass-alpha-lambda}, the control $\nu^{1,*} \in \mathcal{A}_{m}$ is the unique maximiser  to the minor agent's objective functional $H^{1}$ in \eqref{eq-alternative-minor} if and only if the process $(Q^{1,\nu^{1,*}},\nu^{1,*})$ satisfies the following coupled linear FBSDE system
\begin{equation}\label{eq-minor-fbsde}
\begin{aligned}
dQ_{t}^{1,\nu^{1,*}}&= - \nu^{1,*}_{t}dt, \quad Q_{0}^{1,\nu^{1,*}}  = 0,\\
d\nu_{t}^{1,*}&=\frac{1}{2\lambda_{1}} d\mathcal{N}_{t}+
\frac{1}{2\lambda_{1}}d\mathcal{M}_{t}-\frac{\phi^{1}_{t}}{\lambda_{1}}Q_{t}^{1,\nu^{1,*}}dt -\frac{\kappa_{0}}{2\lambda_{1}}\nu^{0}_{t}dt +\frac{1}{2\lambda_{1}}dA_{t}, \\
\nu_{T}^{1,*}&=\frac{2\alpha -\kappa_{1}}{2\lambda_{1}}Q_{T}^{1,\nu^{1,*}}
\end{aligned}
\end{equation}
$d\mathbb{P}\otimes dt$-a.e. on $\Omega\times[0,T]$ where $\mathcal{M}=(\mathcal{M}_{t})_{t\in[0,T]}$ and $\mathcal{N}=(\mathcal{N}_{t})_{t\in[0,T]}$ are two suitable square integrable martingales. 
\end{lemma}
For the remainder of this section we focus on the derivation of the explicit solution to \eqref{eq-minor-fbsde}. We first describe the heuristics of the proof.  

\paragraph{Heuristics for the solution to \eqref{eq-minor-fbsde}.} The solution to the FBSDE system \eqref{eq-minor-fbsde} determines the solution to the minor agent's problem. The main obstacle in solving the system \eqref{eq-minor-fbsde} is that it presents a general time dependent coefficient $\phi^{1}_{t}$. 
In order to solve this equation, we formulate an ansatz for the minor agent's optimal strategy $\nu^{1,*}$. Then, we demonstrate that the ansatz solution for $\nu^{1,*}$ is the unique solution to \eqref{eq-minor-fbsde} and therefore the solution to the minor agent's problem. Due to the linear structure of the system \eqref{eq-minor-fbsde}, we make the ansatz that there are two progressively measurable processes $r^{0} =(r^{0}_{t})_{t\in[0,T]}$
 and $r^{1} =(r^{1}_{t})_{t\in[0,T]}$ such that $\nu^{1,*}$ can be expressed as
 \begin{equation}
 \label{eq-ansatz-minor}
 \nu^{1,*}_{t}=-\left(r^{0}_{t}+r^{1}_{t}Q_{t}^{1,\nu^{1,*}}\right), \quad 0\leq t \leq T. 
 \end{equation}
 We differentiate \eqref{eq-ansatz-minor} via It\^o’s lemma and by using $dQ_{t}^{1,\nu^{1,*}} = -\nu^{1,*}_{t}dt$ to get
 \begin{equation}
 \label{eq-differential-nu1}
 \begin{aligned}
 d \nu^{1,*}_{t} =- dr^{0}_{t} - dr^{1}_{t}Q_{t}^{1,\nu^{1,*}} + \nu^{1,*}_{t}r^{1}_{t}dt.
  \end{aligned}
 \end{equation} 
 We plug-in \eqref{eq-differential-nu1} into \eqref{eq-minor-fbsde} and we arrive at
 \begin{equation}
 \label{eq-collecting-q}
 \begin{aligned}
 0&=\left(2\lambda_{1}dr^{1}_{t} - 2\phi^{1}_{t}dt +2\lambda_{1}(r^{1}_{t})^{2}dt \right)Q_{t}^{1,\nu^{1,*}}\\
  &\quad+
 \left(2\lambda_{1}dr^{0}_{t} +2\lambda_{1}r^{1}_{t}r^{0}_{t} dt-\kappa_{0}\nu^{0}_{t}dt +dA_{t} +d\mathcal{M}_{t} + d\mathcal{N}_{t}\right).
 \end{aligned}
 \end{equation}
Equation \eqref{eq-collecting-q} must hold $d\mathbb{P}\otimes dt$ almost everywhere for all values $Q_{t}^{1,\nu^{1,*}}$. We conjecture that the terms within each of the brackets must vanish independently.   
The terms from \eqref{eq-collecting-q} yield to two coupled differential equations for $r^{0}$ and $r^{1}$ independent of the process $Q^{1,\nu^{1,*}}$, where we determine the terminal conditions from  \eqref{eq-minor-fbsde}.
Specifically, the process $r^{1}$ must satisfy $dt$-a.e. the following non-autonomous Riccati ODE 
\begin{equation}
\label{eq-v1-ode-proof}
 \partial_{t}r^{1}_{t} = \frac{1}{\lambda_{1}}\phi^{1}_{t} -(r^{1}_{t})^{2},  \quad r^{1}_{T}=-\frac{2\alpha- \kappa_{1}}{2\lambda_{1}}, 
\end{equation}
 while $r^{0}$ must satisfy the following BSDE
\begin{equation}
\label{eq-v0-bsde}
-dr^{0}_{t}=  r^{1}_{t}r^{0}_{t} dt-\frac{\kappa_{0}}{2\lambda_{1}}\nu^{0}_{t}dt + \frac{1}{2\lambda_{1}}dA_{t} +\frac{1}{2\lambda_{1}}d\mathcal{M}_{t} + \frac{1}{2\lambda_{1}} d\mathcal{N}_{t}, \quad 
r^{0}_{T}= 0.
\end{equation}
An explicit formula for the solution of \eqref{eq-v1-ode-proof} does not exist when $\phi^{1}_{t}$ is a general piecewise continuous function as in the case at hand. Nevertheless, we will prove that the solution to \eqref{eq-v1-ode-proof} exists and it is unique. Once a solution to \eqref{eq-v1-ode-proof} is found, then we can plug it to \eqref{eq-v0-bsde} and derive $r^0$. 

In the following proposition, which is proved in Appendix \ref{sec-proof-prop-ode}, we derive the existence and uniqueness of the solutions to \eqref{eq-v1-ode-proof} and \eqref{eq-v0-bsde}.  
\begin{proposition}
 \label{prop-results-ode}
Under Assumption \ref{ass-alpha-lambda}, there exists a unique continuous function $r^{1}$ that satisfies the non-autonomous Riccati ODE \eqref{eq-v1-ode-proof} $dt$-a.e. on $[0,T]$.
Furthermore, the BSDE \eqref{eq-v0-bsde} admits a closed form solution $r^{0}$ given by \eqref{eq-v0}. Moreover, 
\be \label{sq-r0} 
\mathbb{E}\left[\int^{T}_{0}(r^{0}_{t})^{2}dt\right]<\infty.
\ee.
 \end{proposition}
 \begin{remark} 
 As stated in Proposition \ref{prop-results-ode}, the function $r^{1}$ satisfies the Riccati ODE \eqref{eq-v1-ode-proof} only $dt$ almost everywhere. This is to be expected since $\phi^{1}$ is assumed to be piecewise continuous and the derivatives of $r^{1}$ may not exists at the points of discontinuity. Nevertheless, as we will show in the proof of Theorem \ref{thm-minor-optimal-strategy}, this is sufficient for our needs as we wish to solve the FBSDE system \eqref{eq-minor-fbsde} only $d\mathbb{P}\otimes dt$ almost everywhere.
 \end{remark} 
 
 In order to prove Theorem \ref{thm-minor-optimal-strategy} we will need the following lemma, which is also proved in Appendix \ref{sec-proof-prop-ode}. 
 \begin{lemma}
\label{lemma-kernel-L2}
Let $r^{1}$ be the unique solution of \eqref{eq-v1-ode-proof}. Then, the kernel $\mathcal{K}$ is jointly continuous over $[0,T]^{2}$. In particular, $\mathcal{K}$ is bounded over $[0,T]^{2}$ and is in $L^{2}([0,T]^{2})$.
\end{lemma}

We are now ready to prove  Theorem \ref{thm-minor-optimal-strategy}. In order to simplify the notation, we will often denote the process $\nu^{1,*}(\nu^{0}) $ by $\nu^{1,*}$.
 \begin{proof}[Proof of Theorem \ref{thm-minor-optimal-strategy}] Let $\nu^{1,*}$ as in  \eqref{eq-minor-agent-optimal} with $r^1$ as in \eqref{eq-v1-ode-proof} and $r^0$ as in \eqref{eq-v0}.
 
\emph{Step 1:} We determine an explicit expression for $Q^{1,\nu^{1,*}}$. We argue that
 \begin{equation}
\label{eq-minor-inventory-explicit}
Q^{1,\nu^{1,*}}_{t} = \int^{t}_{0}\mathcal{K}(s,t)r^{0}_{s}ds. 
\end{equation}
This could verified by using \eqref{eq-kernel-minor} so we can write \eqref{eq-minor-inventory-explicit} as
\begin{equation}
\label{eq-Q1-exponential}
Q^{1,\nu^{1,*}}_{t} = \xi_{t}^{+}\int^{t}_{0} \xi_{s}^{-}r^{0}_{s}ds.
\end{equation} 
From \eqref{eq-def-xi-pm} we note that $\xi^{+}$ satisfies the following ODE,  
\begin{equation}
\label{eq-ode-xi-plus-inventory}
\frac{d\xi_{t}^{+}}{dt} = r^{1}_{t} \xi_{t}^{+} \quad 0\leq t \leq T. 
\end{equation}
Taking the derivative in \eqref{eq-Q1-exponential} and using \eqref{eq-def-xi-pm}, \eqref{eq-kernel-minor} and \eqref{eq-ode-xi-plus-inventory} we arrive at
\begin{equation}
\label{eq-derivative-inventory}
\begin{aligned}
dQ^{1,\nu^{1,*}}_{t} &= r^{1}_{t} \left(\xi_{t}^{+}\int^{t}_{0} \xi_{s}^{-}r^{0}_{s}ds\right) dt + r^{0}_{t}dt \\ 
 &= \left(r^{0}_{t} + r^{1}_{t} \int^{t}_{0} \mathcal{K}(s,t)r^{0}_{s}ds  \right)dt .
\end{aligned}
\end{equation}
From \eqref{eq-derivative-inventory} and \eqref{eq-minor-inventory} we get \eqref{eq-minor-inventory-explicit}. 

\emph{Step 2:} We show that $\nu^{1,*}$ solves \eqref{eq-minor-fbsde}. Note that from \eqref{eq-minor-inventory} and \eqref{eq-minor-agent-optimal} we can rewrite $\nu^{1,*}$ as 
\begin{equation}
\label{eq-minor-agent-optimal-feedback-proof}
\nu^{1,*}_{t} = -\left(r^{0}_{t}  + r^{1}_{t} Q^{1,\nu^{1,*}}_{t}\right), \quad 0\leq t \leq T. 
\end{equation}
By plugging in \eqref{eq-minor-agent-optimal-feedback-proof} into \eqref{eq-minor-fbsde} we conclude that it is enough to prove that \eqref{eq-collecting-q} holds. Since $r^{0}_{t}$ satisfies  \eqref{eq-v0-bsde} and $ r^{1}_{t}$ satisfies \eqref{eq-v1-ode-proof} $dt$-a.e. on $[0,T]$, then \eqref{eq-collecting-q} holds $d\mathbb{P}\otimes dt$ almost everywhere. Next, using the terminal conditions of  $r^{0}$ and $r^{1}$  from \eqref{eq-v0-bsde}  and \eqref{eq-v1-ode-proof} we deduce from \eqref{eq-minor-agent-optimal-feedback-proof} that the terminal condition in \eqref{eq-minor-fbsde} is satisfied. Recall that the forward component in \eqref{eq-minor-fbsde} is satisfied by \eqref{eq-minor-inventory}. Therefore, $(Q^{1,\nu^{1,*}},\nu^{1,*})$ solve the system \eqref{eq-minor-fbsde}, $d\mathbb{P}\otimes dt$ almost everywhere.

\emph{Step 3:} We show that $\nu^{1,*}\in\mathcal{A}_{m}$. From \eqref{eq-minor-admissible-set-def} it follows that we need to verify that 
\begin{equation}
\mathbb{E}\left[\int^{T}_{0} (\nu_{t}^{1,*})^{2}dt\right] <\infty.
\end{equation}
Form \eqref{eq-minor-inventory-explicit}, Proposition \ref{prop-results-ode}, Lemma \ref{lemma-kernel-L2} and Cauchy-Schwartz inequality we get
\begin{equation}
\label{eq-bound-q1-squared}
\begin{aligned}
\mathbb{E}\left[\int^{T}_{0}\left(Q^{1,\nu^{1,*}}_{t}\right)^{2} dt\right] &\leq \mathbb{E}\left[\int^{T}_{0} \left(\int^{t}_{0}\mathcal{K}(s,t)^{2}ds\right)   \left(\int^{t}_{0}(r^{0}_{s})^{2}ds\right) dt\right] \\
 &<\infty. 
\end{aligned}
\end{equation}
From Proposition \ref{prop-results-ode} it follows that $r^{1}$ is continuous hence bounded on $[0,T]$ and that $r^0$ is square-integrable. Using this, \eqref{eq-minor-agent-optimal-feedback-proof} and Cauchy-Schwartz inequality gives 
\begin{equation}
\begin{aligned}
\mathbb{E}\left[\int^{T}_{0} (\nu_{t}^{1,*})^{2}dt\right]  &\leq 2\mathbb{E}\left[\int^{T}_{0} (r^{0}_{t})^{2}dt\right]  + 2\mathbb{E}\left[\int^{T}_{0} (r^{1}_{t})^{2}\left(Q^{1,\nu^{1,*}}_{t}\right)^{2}dt\right]\\
&< \infty. 
\end{aligned} 
\end{equation}
Therefore, $\nu^{1,*}$ is admissible and it solves \eqref{eq-minor-fbsde}, hence by Lemma \eqref{lemma-fbsde-minor} it is the unique maximiser  to the minor agent's objective functional $H^{1}$ in \eqref{eq-alternative-minor}.
\end{proof}

\section{Proof of Theorem \ref{thm-major-solution}}
\label{proof-thm-major-solution}

In this section we derive the major agent's optimal strategy via a calculus of variation argument. 
We start be defining operators which are essential to our proofs. Then we derive an equivalent representation 
the major agent's objective $H^{0}$ which is more convenient to our method of proof. 
Throughout this section we assume that Assumption \ref{ass-alpha-lambda} holds and that the minor agent is adopting the strategy $\nu^{1,*}$ from Theorem \ref{thm-minor-optimal-strategy}. Henceforth, with a slight abuse of notation we write $H^{0}(\nu^{0})$ for $H^{0}(\nu^{0},\nu^{1,*}(\nu^{0}))$. 

\paragraph{Essential definitions of $L^2([0,T])$ operators.}  
We denote by $\mathcal{T}^{*}$ the adjoint kernel of $\mathcal{T}$ for $\langle \cdot,\cdot \rangle_{L^{2}}$, that is
\begin{equation}
\label{eq-definition-adjoint-kernel}
\mathcal{T}^{*}(t,s) := \mathcal{T}(s,t), \quad s,t\in [0,T],
\end{equation}
and by $\mathsf{T}^{*}$ the corresponding adjoint integral operator.

We define the kernel $\mathcal{K}_{1}:[0,T]^{2}\to \mathbb{R}_{+}$ as
\begin{equation}
\label{eq-kernel-K1}
\mathcal{K}_{1}(t,s) := \mathcal{K}(s,t)\mathbbm{1}_{\{s\leq t\}}, \quad s,t\in [0,T],
\end{equation}
where $\mathcal{K}$ is given in \eqref{eq-kernel-minor}. We let $ \mathsf{K}_{1}$ to be the integral operator generated by the kernel $\mathcal{K}_{1}$, that is
\begin{equation}
\label{eq-operator-T1-T2}
(\mathsf{K}_{1}\psi)(t) := \int^{T}_{0}\mathcal{K}_{1}(t,s)\psi(s)ds,\quad  t\in [0,T],  \ \psi\in L^{2}([0,T]),
\end{equation}
The following lemma, which is proved in Section \ref{sec-integral-operators}, outlines some useful properties of $ \mathsf{K}_{1}$. Recall that the class of operators $B(L^{2}([0,T]))$ was defined after \eqref{eq-operator-norm-definition}. 
 \begin{lemma}
\label{lemma-adjoint-operator}
The operator $\mathsf{K}_{1}$ is in $B(L^{2}([0,T]))$.  Moreover,  $\mathsf{K}^{*}_{1}\in B(L^{2}([0,T]))$ is given by
\begin{equation}
\label{eq-adjoint-k1-operator}
(\mathsf{K}^{*}_{1}\psi)(t) = \int^{T}_{0}\mathcal{K}(t,s) \psi(s)\mathbbm{1}_{\{t\leq s\}}ds, \quad  t\in [0,T],  \ \psi\in L^{2}([0,T]). 
\end{equation}
\end{lemma}
We recall the definition of a compact operator from \cite{Porter1990}, (see Definition 3.2 therein).  \begin{definition}
\label{def-compact-operator}
An operator $\mathsf{T}:L^{2}([0,T])\to L^{2}([0,T])$ is said to have finite rank if its image $\{\mathsf{T}\psi:\psi\in L^{2}([0,T])\}$ has finite dimension. An operator $\mathsf{L}\in B(L^{2}([0,T]))$ is said to be compact if there is a sequence $(\mathsf{L}_{n})_{n\geq 1}$ of finite-rank operators in $B(L^{2}([0,T]))$ such that $||\mathsf{L}_{n}-\mathsf{L}||\to 0$ as $n\to\infty$.
\end{definition}
A particularly important result that we will use states that any operator generated by a kernel in $L^{2}([0,T]^{2})$ is compact (see Theorem 3.4 in \cite{Porter1990}). We remark that not all operators in $B(L^{2}([0,T]))$ are compact operators as pointed out in example 3.6 in \cite{Porter1990}. 

Next we define non-negative and positive operators in $B(L^{2}([0,T]))$ as in \cite[Definition 6.1]{Porter1990}.
\begin{definition}
\label{def-positive-operator}
Let $\mathsf{T}\in B(L^{2}([0,T]))$ be a self-adjoint operator. $\mathsf{T}$ is said to be non-negative if and only if $\langle \mathsf{T}\psi,\psi\rangle_{L^{2}} \geq0$ for all $\psi\in L^{2}([0,T])$. It is said to be positive if and only if $\langle \mathsf{T}\psi,\psi\rangle_{L^{2}}> 0$ for all $\psi \neq 0$ in $L^{2}([0,T])$. If there a positive constant $m$ for which $\langle \mathsf{T}\psi,\psi\rangle_{L^{2}}\geq m||\psi||^{2}_{L^{2}}$ for all $\psi\in L^{2}([0,T])$, then $\mathsf{T}$ is said to be positive and bounded below. 
\end{definition}

Our next result outlines several important properties of the operator $\mathsf{G}$ in \eqref{eq-G-rep-integral}. Recall that $\mathsf{K}_1$ was defined in \eqref{eq-operator-T1-T2}. 
\begin{proposition}
\label{prop-G-adjoint-compact}
Let $\mathsf{G}$ be defined as in \eqref{eq-G-rep-integral}. Then $\mathsf{G}\in B(L^{2}([0,T]))$ is a positive, compact and self-adjoint operator. Moreover, it satisfies
\begin{equation}
\label{eq-G-rep-K}
(\mathsf{G}\psi)(t)  =  (\mathsf{K}_{1}\mathsf{K}_{1}^{*}\psi)(t), \quad \textrm{for all } t\in[0,T], \ \psi\in L^{2}([0,T]). 
\end{equation}
\end{proposition}
The proof of Proposition \ref{prop-G-adjoint-compact} is given in Section \ref{sec-integral-operators}.

In the following lemma we determine an alternative representation for the major agent's objective functional.
\begin{lemma}
\label{lemma-alternative-major}
Let $H^{0}$ be the major agent's objective functional in \eqref{eq-major-functional}. Then, for any $\nu^{0}\in\mathcal{A}_{M}^{q_0}$ it holds that
\begin{equation}
\begin{aligned}
\label{eq-alternative-major}
H^{0}(\nu^{0}) &=x_{0} + M_{0}q_0  -\kappa_{0}\frac{q_0^{2}}{2}- \frac{\kappa_{1}\kappa_{0}}{2\lambda_{1}}\int^{T}_{0}\nu^{0}_{t}(\mathsf{G}\nu^{0})(t)  dt   \\ &\quad+\int^{T}_{0}\left(\frac{\kappa_{1}}{2\lambda_{1}}\nu^{0}_{t}(\mathsf{G}\bar{\mu})(t)  +  Q^{0,\nu^{0}}_{t}\bar{\mu}_{t}\right)dt -\lambda_{0}\int^{T}_{0}(\nu^{0}_{t})^{2}dt .
\end{aligned}
\end{equation}
\end{lemma}
The proof of Lemma \ref{lemma-alternative-major} is postponed to Section \ref{sec-proof-alternative-major}. 

The following proposition establishes the uniqueness of the maximiser of $H^{0}$.
\begin{proposition}
\label{prop-uniqueness-major}
There exists at most one admissible maximiser to the major agent's objective functional  $H^{0}$ in \eqref{eq-major-functional}
\end{proposition}
\begin{proof}
To show the desired result, it is sufficient to show that the objective functional $H^{0}$ is strictly concave over $\mathcal{A}_{M}^{q_0}$.

Notice that since $\nu^{1,*}\in\mathcal{A}_{m}$, $\nu^{0}\in\mathcal{A}_{M}^{q_0}$ and $\mu$ satisfies \eqref{eq-mu-square-integrable}, then $H^{0}$ is finite.  Therefore, in order to show that the objective $H^{0}$ is strictly concave over $\mathcal{A}_{M}^{q_0}$ we must verify  that 
\begin{equation}
\label{eq-strictly-concave-major}
\mathcal{I}^{0}(\rho,\nu,\omega) := H^{0}(\rho\nu + (1-\rho)\omega) - \rho H^{0}(\nu) -(1-\rho) H^{0}(\omega)>0 
\end{equation}
 for any $\rho\in(0,1)$ and for any $dt$-distinguishable $\nu,\omega\in\mathcal{A}_{M}^{q_0}$. 

We now fix such $\rho$ and $\nu,\omega$. From \eqref{eq-major-inventory-def} and \eqref{eq-G-rep-integral} it follows that $Q^{0,\nu}$ and $\mathsf{G}\nu$ are linear in $\nu$. Then from \eqref{eq-alternative-major} and \eqref{eq-strictly-concave-major} we get 
\begin{equation}
\label{eq-left-hand side-concavity-major-2}
\begin{aligned}
\mathcal{I}^{0}(\rho,\nu,\omega) 
&= \rho(1-\rho)\lambda_{0}\int^{T}_{0}
(\nu_{t} -\omega_{t})^{2}dt  \\
&\quad+\rho(1-\rho)\frac{\kappa_{1}\kappa_{0}}{2\lambda_{1}} \int^{T}_{0}  (\nu_{t} - \omega_{t} )( \mathsf{G}(\nu -\omega))(t) dt.
\end{aligned}
\end{equation}
Note that the left to show that the right-hand side of \eqref{eq-left-hand side-concavity-major-2} is strictly positive since $\mathsf{G}$ is a positive operator by Proposition \ref{prop-G-adjoint-compact} and $\nu$ and $\omega$ are $dt$-distinguishable.
\end{proof}

As similarly done in Section \ref{sec-proof-thm-minor-optimal}, we can apply a convex analytic calculus of variation approach to derive a first order optimality conditions which characterizes the unique solution to the major agent's problem.

Since the map $\nu^{0}\to H^{0}(\nu^{0})$ in \eqref{eq-major-functional} is strictly concave, then $ H^{0}$ admits a unique maximiser characterised by the critical point at which the G\^ateaux derivative
\begin{equation}
\label{eq-def-major-G\^ateaux}
\langle\mathcal{D}H^{0}(\nu^{0}),\omega\rangle = \lim_{\epsilon  \to 0} \frac{H^{0}(\nu^{0}+\epsilon\omega)-H^{0}(\nu^{0})}{\epsilon}
\end{equation}
vanishes. We remind the reader that  the minor agent optimal strategy $\nu^{1,*}$ was fixed to be the one in \eqref{eq-minor-agent-optimal}, therefore, making the major agent's objective functional $H^{0}$  only a function of the major agent control $\nu^{0}$. In the next lemma we explicitly compute the first-order G\^ateaux derivative of $H^{0}$. Recall that $\mathcal{A}_{M}^{q_{0}}$ was defined in \eqref{eq-major-admissible-set}.
\begin{lemma}
\label{lemma-gateaux-major}
The G\^ateaux derivative of $H^{0}$ in \eqref{eq-alternative-major}, in a direction $\omega\in\mathcal{A}_{M}^{0}$ is given by
\begin{equation}
\label{eq-major-G\^ateaux}
\langle\mathcal{D} H^{0}(\nu^{0}),\omega\rangle =  \int^{T}_{0}\omega_{t}\Bigg(-2\lambda_{0}\nu^{0}_{t}  - \frac{\kappa_{1}\kappa_{0}}{\lambda_{1}} (\mathsf{G}\nu^{0})(t)  +\frac{\kappa_{1}}{2\lambda_{1}}(\mathsf{G}\bar{\mu})(t) - \int^{T}_{t}\bar{\mu}_{s}ds\Bigg) dt,
\end{equation}
for any $\nu^{0}\in\mathcal{A}_{M}^{q_0}$.
\end{lemma}
\begin{proof}
Let $\epsilon>0$, $\nu^{0}\in\mathcal{A}_{M}^{q_0}$ and $\omega \in \mathcal{A}_{M}^{0}$.  From \eqref{eq-alternative-major} we get
\begin{equation}
\label{eq-difference-gateaux}
\begin{aligned}
&H^{0}(\nu^{0}+\epsilon\omega)-H^{0}(\nu^{0})\\ &=\epsilon\Bigg(  \int^{T}_{0}\omega_{t}\left(-2\lambda_{0}\nu^{0}_{t} +\frac{\kappa_{1}}{2\lambda_{1}}(\mathsf{G}\bar{\mu})(t)\right) dt - \frac{\kappa_{1}\kappa_{0}}{2\lambda_{1}} \int^{T}_{0}\omega_{t}(\mathsf{G}\nu^{0})(t)dt\\ 
& \qquad- \frac{\kappa_{1}\kappa_{0}}{2\lambda_{1}}\int^{T}_{0} \nu^{0}_{t}(\mathsf{G}\omega)(t)dt - \int^{T}_{0}\bar{\mu}_{t}\left(\int^{t}_{0}\omega_{s}ds\right)dt \Bigg)\\
&\quad+ \epsilon^{2}\Bigg(-\lambda_{0}\int^{T}_{0}\omega_{s}^{2}ds  -\frac{\kappa_{1}\kappa_{0}}{2\lambda_{1}}\int^{T}_{0}\omega_{t}(\mathsf{G}\omega)(t)dt  \Bigg).
\end{aligned}
\end{equation}
Notice that all the terms in \eqref{eq-difference-gateaux} are finite as a consequence of  $\omega \in\mathcal{A}_{M}^{0}$, $\nu^{0}\in\mathcal{A}_{M}^{q_0}$,  $\mu$ satisfying \eqref{eq-mu-square-integrable} and $\mathsf{G}\in B(L^{2}([0,T]))$ as shown in Proposition \ref{prop-G-adjoint-compact}.  It follows that \begin{equation}
\label{eq-minor-gateaux-proof}
\begin{aligned}
\langle\mathcal{D}H^{0}(\nu^{0}),\omega\rangle &= \int^{T}_{0}\omega_{t}\Big(-2\lambda_{0}\nu^{0}_{t} - \frac{\kappa_{1}\kappa_{0}}{2\lambda_{1}} (\mathsf{G}\nu^{0})(t) +\frac{\kappa_{1}}{2\lambda_{1}}(\mathsf{G}\bar{\mu})(t)\Big)dt \\ 
&\qquad- \frac{\kappa_{1}\kappa_{0}}{2\lambda_{1}}\int^{T}_{0}\nu^{0}_{t}(\mathsf{G}\omega)(t)dt- \int^{T}_{0}\left(\bar{\mu}_{t}\right)\left(\int^{t}_{0}\omega_{s}ds\right)dt .
\end{aligned}
\end{equation}
Proposition \ref{prop-G-adjoint-compact} shows that $\mathsf{G}\in B(L^{2}([0,T]))$. Since $\mu$ satisfies \eqref{eq-mu-square-integrable}, then $\bar{\mu}$ is in $L^{2}([0,T])$ by Jensen's inequality. 
Since $\omega\in\mathcal{A}_{M}^{0}$ and $\bar{\mu}\in L^{2}([0,T])$, then we can apply Fubini's theorem to obtain
\begin{equation}
\label{eq-fubini-major-1}
 \int^{T}_{0}\left(\bar{\mu}_{t}\right)\left(\int^{t}_{0}\omega_{s}ds\right)dt 
 = \int^{T}_{0}\omega_{s}\int^{T}_{s}\bar{\mu}_{t}dt ds.
\end{equation}
As claimed in Proposition \ref{prop-G-adjoint-compact}, $\mathsf{G}$ is self-adjoint, therefore, it holds that
\begin{equation}
\begin{aligned}
\label{eq-self-adjoint-nu0}
\int^{T}_{0}\nu^{0}_{t}(\mathsf{G}\omega)(t)dt &= \langle\nu^{0},(\mathsf{G}\omega)\rangle_{L^{2}} \\ &= \langle\omega,(\mathsf{G}\nu^{0})\rangle_{L^{2}} \\ 
&=\int^{T}_{0}\omega_{t}(\mathsf{G}\nu^{0})(t)dt .
\end{aligned}
\end{equation}
Finally, we use \eqref{eq-fubini-major-1} and  \eqref{eq-self-adjoint-nu0} in \eqref{eq-minor-gateaux-proof} to get \eqref{eq-major-G\^ateaux}. 

\end{proof}
Recall the definition of the operator $\mathsf{S}$ in \eqref{eq-chi-function}.  In the following lemma we derive an optimality condition that takes the form of an integral equation.
\begin{proposition}
\label{prop-integral-eq-major}
A strategy $\nu^{0,*}\in\mathcal{A}^{q_0}_{M}$ maximises the major agent's performance functional \eqref{eq-major-functional} if there exists a constant $\eta \in \mathbb{R}$ such that, 
\begin{equation}
\label{eq-integral-eq-major}
\nu^{0,*}_{t}  + \frac{\kappa_{1}\kappa_{0}}{2\lambda_{0}\lambda_{1}} (\mathsf{G}\nu^{0,*})(t) =(\mathsf{S}\bar{\mu})(t)  +\frac{\eta}{2\lambda_{0}}, \quad \textrm{for all } 0\leq t \leq T. 
\end{equation}

\end{proposition}

\begin{proof}
In the proof of Proposition \ref{prop-uniqueness-major} we have shown that the objective functional $H^{0}$ is strictly concave over $\mathcal{A}_{M}^{q_0}$. Therefore, by Proposition 2.1 in Chapter 2 in \cite{EkelTem:99} if an admissible strategy $\nu^{0,*}$ satisfies 
\begin{equation}
\label{eq-vanishing-derivative-major}
\langle\mathcal{D}H^{0}(\nu^{0,*}),\omega\rangle =0,\quad  \textrm{for all } \omega \in\mathcal{A}_{M}^{0},
\end{equation}
then it is the maximiser of $H^{0}$. 
 
Recalling the result of Lemma \ref{lemma-gateaux-major}, we plug \eqref{eq-integral-eq-major} into \eqref{eq-major-G\^ateaux} to get
\begin{equation}
\label{eq-gateaux-first-order}
\begin{aligned}
\langle\mathcal{D}H^{0}(\nu^{0,*}),\omega\rangle  
&= -\int^{T}_{0}\omega_{t}\left( 2\lambda_{0}(\mathsf{S}\bar{\mu})(t) + \eta  -\frac{\kappa_{1}}{2\lambda_{1}}(\mathsf{G}\bar{\mu})(t)+\int^{T}_{t}\bar{\mu}_{s}ds\right)dt \\
&= -\left(\int^{T}_{0}\bar{\mu}_{s}ds+\eta\right)\int^{T}_{0}\omega_{t}dt 
\end{aligned}
\end{equation}
where we used \eqref{eq-chi-function} in the second equality. Recall that $\omega \in \mathcal{A}_{M}^{0}$, then from \eqref{eq-major-admissible-set} we get \eqref{eq-vanishing-derivative-major}. 
\end{proof}

The following proposition, which is proved in Section \ref{sec-integral-operators}, derives some properties of the operator $\mathsf{R}$ in \eqref{eq-resolvent-operator} that are crucial to the proof of Theorem \ref{thm-major-solution}.

\begin{proposition}
\label{prop-properties-resolvent}
Let $\mathsf{R}$ be defined as in \eqref{eq-resolvent-operator} and $\mathsf{G}$ be defined as in \eqref{eq-G-rep-integral}.  Then the following holds:
\begin{itemize}
\item[\textbf{(i)}] The inverse operator of $\mathsf{I} + \frac{\kappa_{1}\kappa_{0}}{2\lambda_{0}\lambda_{1}}\mathsf{G}$  exists and it satisfies 
\begin{equation}
\label{eq-R-inverse-operator}
\mathsf{R}=\left(\mathsf{I} + \frac{\kappa_{1}\kappa_{0}}{2\lambda_{0}\lambda_{1}}\mathsf{G}\right)^{-1}.
\end{equation}
\item[\textbf{(ii)}] $\mathsf{R}$ is positive, bounded from below and is in $B(L^{2}([0,T]))$. 
\item[\textbf{(iii)}]  If $f\in C([0,T])$, then $(\mathsf{R}f)\in C([0,T])$.
\end{itemize}
\end{proposition}

The following lemma derive some essential properties of  the operator $\mathsf{S}$ in \eqref{eq-chi-function}. 
\begin{lemma}
\label{lemma-chi-continuous}
Let $\mathsf{S}$ be as in \eqref{eq-chi-function}. Then $\mathsf{S}$ is in $B(L^{2}([0,T]))$ and for any $\psi\in L^{2}([0,T])$,$(\mathsf{S}\psi)(t)$ is continuous on $0\leq t \leq T$.
\end{lemma}
\begin{proof}
Proposition \ref{prop-G-adjoint-compact} proves that $\mathsf{G}$ is in $B(L^{2}([0,T]))$, hence it follows from \eqref{eq-chi-function} that also $\mathsf{S}$ is in $B(L^{2}([0,T]))$.

Next, let $\psi\in L^{2}([0,T])$. From Proposition \ref{proposition-operator-differential-eq} it follows that $\mathsf{G}\psi$ is continuously differentiable on $[0,T]$, hence by \eqref{eq-chi-function} the continuity of $(\mathsf{S}\psi)(\cdot)$ follows. 
\end{proof}
We are now ready to prove Theorem \ref{thm-major-solution}. 
\begin{proof}[Proof of Theorem \ref{thm-major-solution}]
Recall that $\nu^{0,*}$ and $\eta$ were defined in \eqref{eq-major-solution} and \eqref{eq-constant-major-solution}, respectively. We split the proof into the following steps. 

\textbf{Step 1.} We show that $(\nu^{0,*},\eta)$ is the unique solution to \eqref{eq-integral-eq-major}. Recall that $\mathsf{R}$ was defined in \eqref{eq-resolvent-operator}. From \eqref{eq-mu-square-integrable} and Lemma \ref{lemma-chi-continuous} it follows that $\mathsf{S}\bar{\mu}(\cdot)$ is in $L^{2}([0,T])$. Then, from Proposition \ref{prop-properties-resolvent} we get that the unique solution to \eqref{eq-integral-eq-major} is given by
\begin{equation}
\label{eq-solution-major-proof}
\begin{aligned}
\nu^{0,*}_{t} &=\left( \left(\mathsf{I} + \frac{\kappa_{1}\kappa_{0}}{2\lambda_{0}\lambda_{1}}\mathsf{G}\right)^{-1}\left(\frac{\eta}{2\lambda_{0}} + \mathsf{S}\bar{\mu} \right)\right) (t)\\
&= \frac{\eta}{2\lambda_{0}}(\mathsf{R}\mathit{1})(t)  + (\mathsf{R}\mathsf{S}\bar{\mu})(t), \quad 0\leq t \leq T. 
\end{aligned}
\end{equation}

\textbf{Step 2.} We verify that $\nu^{0,*}$ satisfies the fuel constraint in \eqref{eq-major-admissible-set}. Note that we need to show that $\langle \nu^{0,*},\mathit{1}\rangle_{L^{2}} = q_0$. From \eqref{eq-major-solution} and \eqref{eq-constant-major-solution} we get
\begin{equation}
\begin{aligned}
\langle \nu^{0,*},\mathit{1}\rangle_{L^{2}} &= \frac{\eta}{2\lambda_{0}} \langle \mathsf{R}\mathit{1},\mathit{1}\rangle_{L^{2}} + \langle \mathsf{R}\mathsf{S}\bar{\mu},\mathit{1}\rangle_{L^{2}} \\
&=\left(\frac{q_0-\langle \mathsf{R}\mathsf{S}\bar{\mu},\mathit{1}\rangle_{L^{2}} }{\langle \mathsf{R}\mathit{1},\mathit{1}\rangle_{L^{2}}}\right)\langle \mathsf{R}\mathit{1},\mathit{1}\rangle_{L^{2}} + \langle \mathsf{R}\mathsf{S}\bar{\mu},\mathit{1}\rangle_{L^{2}} \\
&= q_0. 
\end{aligned}
\end{equation}
Note that by Proposition \ref{prop-properties-resolvent}(ii) the operator $\mathsf{R}$ is positive and bounded from below. Therefore the denominator in \eqref{eq-constant-major-solution} is strictly positive and the constant $\eta$ is well-defined.

\textbf{Step 3.} We verify that $\nu^{0,*}$ is in $L^2([0,T])$. From Lemma \ref{lemma-chi-continuous} it follows that $(\mathsf{S}\bar{\mu})(t)$ is continuous on $[0,T]$. Proposition \ref{prop-properties-resolvent}(iii) and \eqref{eq-solution-major-proof} then imply that $\nu^{0,*}_{t}$ is continuous and therefore bounded on $[0,T]$. 
This concludes the proof. 
\end{proof}

\section{Proofs of Lemma \ref{lemma-adjoint-operator} and Propositions \ref{prop-G-adjoint-compact} and \ref{prop-properties-resolvent}}
\label{sec-integral-operators}

\begin{proof}[Proof of Lemma \ref{lemma-adjoint-operator}]
Recall that kernels $\mathcal{K}$ and $\mathcal{K}_{1}$ were defined in \eqref{eq-kernel-minor} and \eqref{eq-kernel-K1} respectively. Let $\psi$ be in $L^{2}([0,T])$. From \eqref{eq-definition-adjoint-kernel} we get 
\begin{equation} \label{k-1-st} 
\mathcal{K}^{*}_{1}(t,s) =  \mathcal{K}(t,s)\mathbbm{1}_{\{t\leq s\}},  \quad \textrm{for all } 0\leq t,s \leq T,
\end{equation}
which verifies \eqref{eq-adjoint-k1-operator}.  From Lemma \ref{lemma-kernel-L2} we get that $\mathcal{K}$ is in $L^{2}([0,T]^{2})$. Then by \eqref{eq-kernel-K1} and \eqref{k-1-st} also $\mathcal{K}_{1}^*$ and $\mathcal{K}_{1}$ are in $L^{2}([0,T]^{2})$, therefore, $\mathsf{K}_{1}$ and $\mathsf{K}_{1}^*$ are in $B(L^{2}([0,T]))$.  
\end{proof}
The following lemma describes several properties of the kernel $\mathcal{G}$ in \eqref{eq-kernel-caligraphic} which are essential to the proof of Proposition \ref{prop-properties-resolvent}. 
\begin{lemma}
\label{prop-kernel-caligraphic-properties}
Let $\mathcal{G}$ as in \eqref{eq-kernel-caligraphic}. Then, $\mathcal{G}$ is symmetric and jointly continuous on $[0,T]^{2}$. Moreover, $\mathcal{G}$ is in $L^{2}([0,T]^{2})$.
\end{lemma}
\begin{proof} From \eqref{eq-kernel-caligraphic} it follows that $\mathcal{G}$ is symmetric.
Recall that $\xi^{+}$ and $\xi^{-}$ were defined as in \eqref{eq-def-xi-pm}. Note that $\xi^{+}$ and $\xi^{-}$ are continuous on $[0,T]$ and that $\xi^{-}$ belongs to $L^{2}([0,T])$. It follows that $ t \mapsto \int^{t}_{0} (\xi_{u}^{-})^{2}du$
is continuous on $[0,T]$. 
From \eqref{eq-kernel-minor} and  \eqref{eq-kernel-caligraphic} it follows that the kernel $\mathcal{G}$ can be rewritten as follows, 
\begin{equation*}
\mathcal{G}(t,s) = \xi_{t}^{+}\xi_{s}^{+}\int^{t\wedge s}_{0} (\xi_{u}^{-})^{2}du,  \quad \textrm{for all } 0\leq t,s \leq T. 
\end{equation*}
Therefore, $\mathcal{G}(t,s)$  is jointly continuous on $[0,T]^{2}$, hence it is in $L^{2}([0,T]^{2})$.
\end{proof}
The following lemma is needed for the proof of Proposition \ref{prop-G-adjoint-compact}.  
\begin{lemma}
\label{lemma-T1-zero} 
Let $\mathsf{K}_{1}$ be defined as in \eqref{eq-operator-T1-T2} and let $\psi\in L^{2}([0,T])$. If $(\mathsf{K}_{1}\psi)(t) =0$ for all $0\leq t \leq T$, then $\psi(t) =0$ a.e. on $[0,T]$.
\end{lemma}

\begin{proof} Let $\psi\in L^{2}([0,T])$ and assume that $(\mathsf{K}_{1}\psi)(t) =0$ for all $0\leq t \leq T$. From \eqref{eq-kernel-minor}, \eqref{eq-kernel-K1} and \eqref{eq-operator-T1-T2} we get
\begin{equation} \label{k-i} 
(\mathsf{K}_{1}\psi)(t) 
=\xi_{t}^{+}\left(\int^{t}_{0}\xi^{-}_{s} \psi(s)ds\right) =0, \quad \textrm{for all } 0\leq t \leq T. 
\end{equation}
From \eqref{eq-def-xi-pm} it follows that $\xi_{t}^{\pm}>0$ for all $0\leq t \leq T$, therefore 
\begin{equation}
\label{eq-integral-zero-psi}
\int^{t}_{0}\xi^{-}_{s} \psi(s)ds = 0, \quad \textrm{for all } 0\leq t \leq T. 
\end{equation}
Since $\xi_{t}^{\pm }$ are also continuous for all $t\in[0,T]$ and $\psi\in L^{2}([0,T])$, then $\xi^{-} \psi$ is integrable. Then from Lebesgue differentiation theorem and \eqref{eq-integral-zero-psi} we conclude that
\begin{equation*}
\xi^{-}_{t} \psi(t) =0, \quad \textrm{$dt$-a.e. on $[0,T]$}.
\end{equation*}
But again since $\xi_{t}^{-}>0$ on $[0,T]$ we get that 
\begin{equation*}
\psi(t) =0,  \quad \textrm{$dt$-a.e. on $[0,T]$},
\end{equation*}
which proves the result.
\end{proof}

\begin{proof}[Proof of Proposition \ref{prop-G-adjoint-compact}]
Let $\psi\in L^{2}([0,T])$. From \eqref{eq-kernel-caligraphic} we get
\begin{equation*}
\begin{aligned}
(\mathsf{G}\psi)(t) &= \int^{T}_{0}\mathcal{G}(t,u)  \psi(u)du \\
&=  \int^{T}_{0}\psi(u)\int^{t\wedge u}_{0} \mathcal{K}(s,t) \mathcal{K}(s,u)ds du \\ 
&= \int^{T}_{0}\int^{T}_{0} \mathcal{K}(s,t) \mathcal{K}(s,u) \psi(u)\mathbbm{1}_{\{s\leq t\}} \mathbbm{1}_{\{s\leq u\}}du ds. 
\end{aligned}
\end{equation*}
Together with \eqref{eq-kernel-K1} and \eqref{k-1-st} it follows that
\begin{equation}
\label{eq-G-integral-before-indicator}
\begin{aligned}
(\mathsf{G}\psi)(t) &= \int^{T}_{0}\mathcal{K}_{1}(t,s)\int^{T}_{0}  \mathcal{K}_{1}^{*}(s,u) \psi(u) du ds\\
&= (\mathsf{K}_{1}\mathsf{K}_{1}^{*}\psi)(t), \quad \textrm{for all } 0\leq t \leq T, 
\end{aligned}
\end{equation}
which proves \eqref{eq-G-rep-K}.
Since by Lemma \ref{lemma-adjoint-operator}, $\mathsf{K}_{1},\mathsf{K}_{1}^{*}\in B(L^{2}([0,T]))$, a straightforward application of the Cauchy-Schwartz inequality and \eqref{eq-G-integral-before-indicator} shows that  $\mathsf{G} \in B(L^{2}([0,T]))$. 
Note that \eqref{eq-G-integral-before-indicator} also implies that $\mathsf{G}$ is self-adjoint, that is
\begin{equation}
\begin{aligned}
 \mathsf{G}^{*}  &=( \mathsf{K}_{1}\mathsf{K}_{1}^{*})^{*} \\ 
&=(\mathsf{K}_{1}^{*})^{*} \mathsf{K}_{1}^{*} \\
&=\mathsf{G}.
\end{aligned}
\end{equation}
Next, using \eqref{eq-G-rep-integral} we prove that $\mathsf{G}$ is compact (see Definition \ref{def-compact-operator}). From Lemma \ref{prop-kernel-caligraphic-properties} it follows that $\mathcal{G}$ is in $L^{2}([0,T]^{2})$. Then the result follows from Theorem 3.4 in \cite{Porter1990}, which shows that any integral operator generated by a kernel in $L^{2}([0,T]^{2})$ is compact. 

Finally, we prove that $\mathsf{G}$ is a positive operator in the sense of Definition \ref{def-positive-operator}. We have shown that $\mathsf{G}=\mathsf{K}_{1}\mathsf{K}_{1}^{*}$. Since for any $\psi\in L^{2}([0,T])$ we have 
\bd
\langle \mathsf{G}\psi, \psi\rangle =\langle   \mathsf{K}_{1}\mathsf{K}_{1}^{*}\psi, \psi\rangle = ||\mathsf{K}_{1}\psi||^{2}\geq 0, 
\ed
it follows that $\mathsf{G}$ is non-negative. By Lemma \ref{lemma-T1-zero} we have $(\mathsf{K}_{1}\psi)(t)=0$ for all $t\in[0,T]$ only for $\psi =0 $ a.e. on $[0,T]$. Therefore, $\mathsf{G}$ is positive. 
\end{proof}

In the following we present a sequence of results which are essential to the proof of Proposition \ref{prop-properties-resolvent}. Using the results of Proposition \ref{prop-G-adjoint-compact} we are in a position to fully characterise the spectral properties of the integral operator $\mathsf{G}$.  
\begin{lemma}[Spectral Decomposition of $\mathsf{G}$]
\label{lemma-spectral-theorem-G}
Let $\mathsf{G}$ be defined as in \eqref{eq-G-rep-integral}. Then $\mathsf{G}$ has a sequence of positive eigenvalues $(\zeta_{n})_{n\geq 1}$ and a corresponding orthonormal sequence $(\psi_{n})_{n\geq 1 }$ of eigenfunctions in $L^{2}([0,T])$, such that for each $\varphi\in L^{2}([0,T])$, we have that
\begin{equation*}
\mathsf{G}\varphi = \sum_{n \geq 1}  \zeta_{n}\langle\varphi,\psi_{n}\rangle_{L^{2}}\psi_{n}.
\end{equation*} 
Moreover, for all $N\geq 1$, define $\mathsf{G}_{N},\mathsf{G}_{N}^{\text{abs}}\in B(L^{2}([0,T]))$ by 
\begin{equation}
\label{eq-GN-operator}
\begin{aligned}
\mathsf{G}_{N}\varphi &= \sum_{n=1}^{N} \zeta_{n}\langle\varphi,\psi_{n}\rangle_{L^{2}}\psi_{n}, \\
\mathsf{G}_{N}^{\text{abs}}\varphi &= \sum_{n=1}^{N}\left| \zeta_{n}\langle\varphi,\psi_{n}\rangle_{L^{2}}\psi_{n}\right|.
\end{aligned}
\end{equation}
Then, $\mathsf{G}_{N}\varphi$ converges uniformly to $\mathsf{G}\varphi$,
 \begin{equation}
\label{eq-GN-uniform-absolute-convergence}
\lim_{N\to\infty}\sup_{t\in[0,T]}\left|(\mathsf{G}_{N}\varphi)(t)-(\mathsf{G}\varphi)(t)\right| =0, 
\end{equation}
and $\mathsf{G}_{N}^{\text{abs}}\varphi$ is uniformly convergent, that is, there exists a function $\Phi\in L^{2}([0,T])$ such that
\be \label{abs-con} 
\lim_{N\to\infty}\sup_{t\in[0,T]}\left|(\mathsf{G}_{N}^{\text{abs}}\varphi)(t)-\Phi(t)\right| =0.
\ee
\end{lemma}
\begin{proof}
From Proposition \ref{prop-G-adjoint-compact} it follows that $\mathsf{G}$ is a self-adjoint compact operator in $B(L^{2}([0,T]))$. Therefore, from Theorem 4.15 of \cite{Porter1990} there is a sequence $(\zeta_{n})_{n\geq 1 }$ of non-zero eigenvalues of $\mathsf{G}$  and a corresponding orthonormal sequence $(\psi_{n})_{n \geq 1}$ of eigenfunctions in $L^{2}([0,T])$ such that $\mathsf{G}\psi = \sum_{n \geq 1}^{}\zeta_{n}\langle\psi,\psi_{n}\rangle_{L^{2}}\psi_{n}$. Moreover, the operator $\mathsf{G}_{N}$  converges to $\mathsf{G}$ in mean, i.e. $||\mathsf{G}-\mathsf{G}_{N}||\to 0$ as $N\to \infty$. From Proposition \ref{prop-G-adjoint-compact} it follows that $\mathsf{G}$ is positive and self-adjoint, hence from Lemma 6.1 in \cite{Porter1990} we get that all of its eigenvalues $(\zeta_{n})_{n\geq 1}$ are positive. Since by Proposition \ref{prop-kernel-caligraphic-properties} $\mathcal{G}$ is continuous and symmetric, then from Theorem 4.22 of \cite{Porter1990} it follows that $\mathsf{G}_{N}\psi$ and $\mathsf{G}_{N}^{\text{abs}}$ satisfy the convergence in \eqref{eq-GN-uniform-absolute-convergence} and \eqref{abs-con}. 
\end{proof} 

\begin{remark}
In Appendix \ref{sec-spectral-explicit} we provide an example for the spectral decomposition of $\mathsf{G}$ in Lemma \ref{lemma-spectral-theorem-G}. 
\end{remark}

\begin{lemma}
\label{lemma-operator-positive-bounded}
Let $\mathsf{G}$ be defined as in \eqref{eq-G-rep-integral}. Then, the operator $\mathsf{I} + \frac{\kappa_{1}\kappa_{0}}{2\lambda_{0}\lambda_{1}} \mathsf{G}$ is positive and bounded from below in the sense of Definition \ref{def-positive-operator}.
\end{lemma}
\begin{proof}
Let $\psi\in L^{2}([0,T])$. Since $\mathsf{G}$ is positive by Proposition \ref{prop-G-adjoint-compact} and $\kappa_{0},\kappa_{1},\lambda_{0},\lambda_{1}>0$, we get that  
\begin{equation}
\left\langle\left(\mathsf{I} + \frac{\kappa_{1}\kappa_{0}}{2\lambda_{0}\lambda_{1}} \mathsf{G}\right)\psi,\psi\right\rangle_{L^{2}} \geq  ||\psi||^{2}_{L^{2}}. 
\end{equation}
Therefore, $\mathsf{I} + \frac{\kappa_{1}\kappa_{0}}{2\lambda_{0}\lambda_{1}} \mathsf{G}$ is positive and bounded from below.
\end{proof}

Recall that we assume that constants $\lambda_{0}, \lambda_{1}, \kappa_{1}$ and  $\kappa_{0}$ are strictly positive and  that we proved in Lemma \ref{lemma-spectral-theorem-G} that the eigenvalues of $\mathsf{G}$ are all positive. The following lemma is therefore an easy corollary. 
\begin{lemma}
\label{lemma-eigenvalue-integral-operator}
Let $\zeta^{*}=-\frac{2\lambda_{0}\lambda_{1}}{\kappa_{1}\kappa_{0}}$ and let $\mathsf{G}$ be defined as in \eqref{eq-G-rep-integral}. Then, $\zeta^{*}$ is not an eigenvalue of the integral operator $\mathsf{G}$.
\end{lemma}
We recall Theorem 4.27 of \cite{Porter1990} which will be useful in the proof of Proposition \ref{prop-properties-resolvent}.
\begin{proposition}[{\cite[Theorem 4.27]{Porter1990}}]
\label{proposition-porter-theorem427}
Let $f\in L^{2}([0,T])$. Suppose that $\mathcal{T}:[0,T]^{2}\to\mathbb{R}$ is a continuous symmetric kernel and that $\mathsf{T}$ is the integral operator generated by $\mathcal{T}$. Let $(\mu_{n})_{n\geq 1}^{}$ and $(\varphi_{n})_{n\geq 1}^{}$ be the sequence of eigenvalues and eigenfunctions of the operator $\mathsf{T}$. Moreover, suppose that $\frac{1}{\lambda}$ is not an eigenvalue of $\mathsf{T}$. Then the unique solution to the integral equation
\begin{equation}
\label{eq-prop-integral-porter}
\psi(t)- \lambda\int^{T}_{0}\mathcal{T}(t,s)\psi(s)dt = f(t), \quad t\in[0,T],
\end{equation}
is given by
\begin{equation*}
\psi(t) = f(t) + \int^{T}_{0}{\wt{\mathcal{R}}}(t,s)f(s)ds, \quad t\in[0,T],
\end{equation*}
where 
\begin{equation*}
{\wt{\mathcal{R}}}(t,s) = \lambda \mathcal{T}(t,s) + \sum_{n \geq 1}^{}\frac{\lambda\mu_{n}}{1-\lambda\mu_{n}}\varphi_{n}(t)\varphi_{n}(s), 
\end{equation*}
is jointly continuous on $[0,T]^{2}$.
\end{proposition}

We are now ready to prove Proposition \ref{prop-properties-resolvent}.
\begin{proof}[Proof of Proposition \ref{prop-properties-resolvent}]  (i)  Note that the operator $\mathsf{G}$ and the corresponding kernel $\mathcal{G}$ satisfy the assumptions of Proposition \ref{proposition-porter-theorem427}.  Specifically, from Lemma \ref{lemma-eigenvalue-integral-operator} it follows that $-\frac{2\lambda_{0}\lambda_{1}}{\kappa_{1}\kappa_{0}}$ is not an eigenvalue of $\mathsf{G}$.  Moreover, in  Proposition \ref{prop-kernel-caligraphic-properties} we have shown that $\mathcal{G}$ is continuous and symmetric on $[0,T]^{2}$.  Therefore, we can apply the result of Proposition \ref{proposition-porter-theorem427} to the following integral equation,
\begin{equation}
\label{eq-integral-equation-proof}
\psi(t)+ \frac{\kappa_{1}\kappa_{0}}{2\lambda_{0}\lambda_{1}}\int^{T}_{0}\mathcal{G}(t,s)\psi(s)ds = f(t),\quad t\in[0,T],
\end{equation}
and determine that the unique solution to \eqref{eq-integral-equation-proof} is given by 
\begin{equation}
\label{eq-psi-continuous}
\psi(t) = f(t) + \int^{T}_{0}\mathcal{R}(t,s)f(s)ds, 
\end{equation}
with $\mathcal{R}$ as in \eqref{eq-resolvent-kernel-def}. Moreover, it follows from Proposition \ref{proposition-porter-theorem427} that the kernel $\mathcal{R}$ is jointly continuous on $[0,T]^{2}$. 

Next, we show that the inverse of $\mathsf{I} + \frac{\kappa_{1}\kappa_{0}}{2\lambda_{0}\lambda_{1}}\mathsf{G}$ is given by $\mathsf{R}$. Since by Lemma \ref{lemma-eigenvalue-integral-operator}, $-\frac{2\lambda_{0}\lambda_{1}}{\kappa_{1}\kappa_{0}}$ is not an eigenvalue of $\mathsf{G}$, the operator $\left(\mathsf{I} + \frac{\kappa_{1}\kappa_{0}}{2\lambda_{0}\lambda_{1}}\mathsf{G}\right)^{-1}$ exists. Let $\psi$ be the solution of \eqref{eq-integral-equation-proof}. Since $\mathsf{I} + \frac{\kappa_{1}\kappa_{0}}{2\lambda_{0}\lambda_{1}}\mathsf{G}$ is invertible it follows from \eqref{eq-integral-equation-proof} that $\psi$ can be written as follows, 
\begin{equation}
\label{eq-inverse-implicit}
\psi = \left(\mathsf{I} + \frac{\kappa_{1}\kappa_{0}}{2\lambda_{0}\lambda_{1}}\mathsf{G}\right)^{-1}f.
\end{equation}
On the other hand, from \eqref{eq-resolvent-operator} and \eqref{eq-psi-continuous} we have that
\begin{equation}
\label{eq-inverse-explicit}
\psi = \mathsf{R}f. 
\end{equation}
Therefore, by comparing \eqref{eq-inverse-implicit} and \eqref{eq-inverse-explicit} we find that
\begin{equation}
\label{eq-comparison-inverse}
\left(\mathsf{I} + \frac{\kappa_{1}\kappa_{0}}{2\lambda_{0}\lambda_{1}}\mathsf{G}\right)^{-1}f =  \mathsf{R}f.
\end{equation}
Since \eqref{eq-comparison-inverse} holds for any $f\in L^{2}([0,T])$, (i) follows.

(ii) From Proposition \ref{prop-G-adjoint-compact} it follows that $\mathsf{G}$ is a compact operator. Since $\mathsf{G}$ is a compact and the inverse of $\mathsf{I} + \frac{\kappa_{1}\kappa_{0}}{2\lambda_{0}\lambda_{1}}\mathsf{G}$ exists by (i), we get from the remark below the proof of Theorem 3.3 of \cite{Porter1990}, that the inverse of $\mathsf{I} + \frac{\kappa_{1}\kappa_{0}}{2\lambda_{0}\lambda_{1}}\mathsf{G}$ is in $B(L^{2}([0,T]))$. From \eqref{eq-R-inverse-operator} it follows that  $\mathsf {R}$ is also in $B(L^{2}([0,T]))$.
Recall that Lemma \ref{lemma-operator-positive-bounded} shows that $\mathsf{I} + \frac{\kappa_{1}\kappa_{0}}{2\lambda_{0}\lambda_{1}} \mathsf{G}$ is positive and bounded from below, hence from Lemma 6.2 of \cite{Porter1990} it follows that its inverse is positive and bounded from below. From \eqref{eq-R-inverse-operator} we conclude that $\mathsf{R}$ is positive and bounded from below in the sense of Definition \ref{def-positive-operator}.

(iii) Assume that $f\in C([0,T])$ then (iii) follows from \eqref{eq-inverse-explicit} and since $\mathcal{R}$ is jointly continuous on $[0,T]^{2}$ by (i). 
\end{proof}

\section{Proof of Lemma \ref{lemma-alternative-major}}
\label{sec-proof-alternative-major}

Throughout this section we assume Assumption \ref{ass-alpha-lambda} such that minor agent's optimal control $\nu^{1,*}$ is well-defined. Before proving Lemma \ref{lemma-alternative-major} we prove several intermediate results. \begin{lemma}
\label{lemma-expectation-r0}
Let $r^{0}$ be defined as in \eqref{eq-v0} and $\mathsf{K}_{1}$ be defined as in \eqref{eq-operator-T1-T2}. Then the following holds for any $\nu^{0}\in\mathcal{A}_{M}^{q_0}$, 
\begin{equation}
\label{eq-expectation-v0}
\mathbb{E}[r^{0}_{t}(\nu^{0})]  = \frac{1}{2\lambda_{1}}(\mathsf{K}^{*}_{1}\bar{\mu})(t) -\frac{\kappa_{0}}{2\lambda_{1}}(\mathsf{K}^{*}_{1}\nu^{0})(t), \quad \textrm{for all }0\leq t\leq T,
\end{equation}
 Moreover $(\mathbb{E}[r^{0}_t(\nu^{0})])_{t \in [0,T]}$ is in $L^{2}([0,T])$. 
\end{lemma}
\begin{proof}
Let $\nu^{0}\in\mathcal{A}_{M}^{q_0}$ and recall that $\mathcal{K}$ was defined in \eqref{eq-kernel-minor}.   
From Lemma \ref{lemma-kernel-L2}, \eqref{eq-mu-square-integrable} and \eqref{exp-s} it follows that the conditions of  Fubini's theorem are satisfied and we get that 
\begin{equation}
\label{eq-fubini-k-mu}
\mathbb{E}\left[\int^{T}_{t} \mathcal{K}(t,s)\mu_{s}ds\right] =\int^{T}_{t}  \mathcal{K}(t,s) \bar{\mu}_{s}ds, \quad \textrm{for all }0\leq t\leq T.
\end{equation}
Using \eqref{eq-v0}, \eqref{eq-fubini-k-mu} and the tower property we get
\begin{equation}
\label{eq-expectation-v0-proof}
\begin{aligned}
\mathbb{E}[r^{0}_{t}(\nu^{0})] &= \frac{1}{2\lambda_{1}}\mathbb{E}\left[\mathbb{E}_{t}\left[\int^{T}_{t} \mathcal{K}(t,s)(\mu_{s}-\kappa_{0}\nu^{0}_{s})ds\right]\right] \\
&=  \frac{1}{2\lambda_{1}}\int^{T}_{0} \mathcal{K}(t,s)(\bar{\mu}_{s}-\kappa_{0}\nu^{0}_{s})\mathbbm{1}_{\{t\leq s\}}ds. 
\end{aligned}
\end{equation}
Using the expression for $\mathsf{K}_{1}^{*}$ from \eqref{eq-adjoint-k1-operator} in \eqref{eq-expectation-v0-proof} we arrive at \eqref{eq-expectation-v0}. 

From Lemma \ref{lemma-adjoint-operator} it follows that the operators $\mathsf{K}_{1}$ and $\mathsf{K}_{1}^{*}$ are in $B(L^{2}([0,T]))$.  By assumption $\nu^{0},\bar{\mu}\in L^{2}([0,T])$, then from  \eqref{eq-expectation-v0} it follows that $(\mathbb{E}[r^{0}_t(\nu^{0})])_{t \in [0,T]}$ is in  $L^{2}([0,T])$.
\end{proof}

\begin{lemma}
\label{lemma-expectation-v1}
Let $\nu^{1,*}$ be defined as in \eqref{eq-minor-agent-optimal}. Then or any $\nu^{0}\in\mathcal{A}_{M}^{q_0}$ the following holds 
\begin{equation*}
\label{eq-expectation-nu1-star}
\mathbb{E}[\nu_{t}^{1,*}(\nu^{0})] =  \frac{\kappa_{0}}{2\lambda_{1}}[(\mathsf{K}^{*}_{1}\nu^{0})(t) +r^{1}_{t}(\mathsf{G}\nu^{0})(t) ] - \frac{1}{2\lambda_{1}}[(\mathsf{K}^{*}_{1}\bar{\mu})(t) +r^{1}_{t}(\mathsf{G}\bar{\mu})(t)], 
\end{equation*}
for all $0\leq t\leq T$. 
\end{lemma}
\begin{proof}
Lemmas \ref{lemma-expectation-r0} and \ref{lemma-kernel-L2} prove that $\mathbb{E}\left[ r_{\cdot}^{0}(\nu^{0})\right]\in L^{2}([0,T])$ and $\mathcal{K}\in L^{2}([0,T]^{2})$, respectively. Then from \eqref{eq-minor-agent-optimal} and Fubini's theorem it follows that
\begin{equation*}
\label{eq-expectation-v1-before-sub}
\mathbb{E}[\nu_{t}^{1,*}(\nu^{0})] = - \mathbb{E}\left[r^{0}_{t}(\nu^{0})\right] -  r^{1}_{t}\int^{t}_{0}\mathcal{K}(s,t)\mathbb{E}\left[ r^{0}_{s} (\nu^{0})\right]ds.
\end{equation*}
Together with \eqref{eq-expectation-v0} we get
\begin{equation}
\label{eq-expectation-v1-sub}
\begin{aligned}
\mathbb{E}[\nu_{t}^{1,*}(\nu^{0})]  &= -\frac{1}{2\lambda_{1}}(\mathsf{K}_{1}^{*}\bar{\mu})(t) + \frac{\kappa_{0}}{2\lambda_{1}}(\mathsf{K}_{1}^{*}\nu^{0})(t)\\ &\quad-r^{1}_{t}\int^{t}_{0}\mathcal{K}(s,t) \left( \frac{1}{2\lambda_{1}}(\mathsf{K}_{1}^{*}\bar{\mu})(s)-\frac{\kappa_{0}}{2\lambda_{1}}(\mathsf{K}_{1}^{*}\nu^{0})(s)\right)ds.
\end{aligned}
\end{equation} 
By using \eqref{eq-kernel-K1}, \eqref{eq-operator-T1-T2} and \eqref{eq-G-rep-K} in \eqref{eq-expectation-v1-sub} we get the result. 
\end{proof}
The following lemma simply follows from \eqref{eq-major-admissible-set} and integration by parts hence we omit the proof.  
\begin{lemma}
\label{lemma-integral-identity-3}
Let $M$ be a martingale as in \eqref{ass:P}. Then for any $\nu^{0}\in\mathcal{A}_{M}^{q_0}$, we have 
\begin{equation*}
 \mathbb{E}\left[\int^{T}_{0} M_{t}\nu^{0}_{t}dt\right] = M_{0}q_0 \label{eq-F-identity}.
\end{equation*}
\end{lemma}
 
In the following proposition we derive an operator differential equation which is satisfied by the operator $\mathsf{G}$ in \eqref{eq-G-rep-integral}.  
\begin{proposition}
\label{proposition-operator-differential-eq}
For any  $\psi\in  L^{2}([0,T])$ the operator $\mathsf{G}$ satisfies the following differential equation, 
\begin{equation}
\label{eq-G-ODE}
\frac{d}{dt} (\mathsf{G}\psi)(t)  = r^{1}_{t} (\mathsf{G}\psi)(t) + (\mathsf{K}_{1}^{*}\psi)(t) , \quad (\mathsf{G}\psi)(0) = 0, \quad 0\leq t\leq T.
\end{equation}
Moreover, $(\mathsf{G}\psi)(t) $ is continuously differentiable on $[0,T]$.
\end{proposition}

\begin{proof} 
Let $\psi\in L^{2}([0,T])$. From Proposition \ref{prop-kernel-caligraphic-properties} it follows that $\mathcal{G}$ is jointly continuous on $[0,T]^{2}$ hence by \eqref{eq-G-rep-integral} we get that $(\mathsf{G}\psi)(t)$ is continuous on $[0,T]$.  

Note that by \eqref{eq-def-xi-pm} we have
\begin{equation}
\label{eq-ode-xi-plus}
\frac{d\xi^{\pm}}{dt}=\pm r^{1}_{t}\xi_{t}^{\pm}, \quad \textrm{for all } 0\leq t\leq T. 
\end{equation}
Since by Proposition \ref{prop-results-ode} $ r^{1}$ is continuous over $[0,T]$, it follows from \eqref{eq-ode-xi-plus} that $\xi^{\pm}$ are continuously differentiable on $[0,T]$. From \eqref{k-i}, \eqref{eq-G-integral-before-indicator} and \eqref{eq-ode-xi-plus} we get that 
\begin{equation*}
\begin{aligned}
\frac{d}{dt}(\mathsf{G}\psi)(t)  &= \frac{d}{dt}\left(\xi_{t}^{+}\int^{t}_{0}\xi_{s}^{-}(\mathsf{K}_{1}^{*}\psi)(s)ds\right) \\
&=r^{1}_{t}\left(\xi_{t}^{+}\int^{t}_{0}\xi_{s}^{-}(\mathsf{K}_{1}^{*}\psi)(s)ds\right) + (\mathsf{K}_{1}^{*}\psi)(t)\\
&=r^{1}_{t}(\mathsf{G}\psi)(t) + (\mathsf{K}_{1}^{*}\psi)(t). 
\end{aligned}
\end{equation*}
Since by \eqref{eq-G-rep-K} that the operator $\mathsf{G}$ can be represented in terms of $\mathsf{K}_{1}$ and $\mathsf{K}^{*}_{1}$ and Lemma \ref{lemma-kernel-L2} shows that $\mathcal{K}$ is jointly continuous on $[0,T]^{2}$, it follows from \eqref{eq-operator-T1-T2} and \eqref{eq-adjoint-k1-operator} that $(\mathsf{G}\psi)(t)$ and $(\mathsf{K}_{1}^{*}\psi)(t)$ are continuous on $[0,T]$. Since we have show that $r^{1}_{t}$ is also continuous, it follows that $\frac{d}{dt}(\mathsf{G}\psi)$ is continuous over $[0,T]$.  Finally, note that
\begin{equation*}
\label{eq-K1-zero-psi}
(\mathsf{K}_{1}\psi)(0) = \int^{T}_{0}\mathcal{K}(s,0)\mathbbm{1}_{\{s\leq 0\}}\psi(s)ds  =0 .
\end{equation*}
From \eqref{eq-G-rep-K} we have $(\mathsf{G}\psi)(t)= (\mathsf{K}_{1}\mathsf{K}_{1}^{*}\psi)(t)$. This proves that $(\mathsf{G}\psi)(0) = 0$ and completes the proof. 
\end{proof}

Now we are ready to prove Lemma \ref{lemma-alternative-major}. 
\begin{proof}[Proof of Lemma  \ref{lemma-alternative-major}]
Let $\nu^{0}\in\mathcal{A}_{M}^{q_0}$. Recall that minor agent's strategy is assumed to be $\nu^{1,*}$ in \eqref{eq-minor-agent-optimal}. We define 
\begin{equation}
\label{eq-definition-z-process}
Z^{\nu}_{t} =Y_{t}^{\nu} - \int^{t}_{0}\mu_{s}ds,\quad 0\leq t \leq T.
\end{equation}
Note that from \eqref{eq-permanent-impact} it follows that $Z^{\nu}_{0} =0$.

Using \eqref{eq-major-cash} and \eqref{eq-execution-price-major} we get 
\begin{equation*}
 \mathbb{E}\left[X^{0,\nu^{0}}_{T}\right] = x_{0} + \mathbb{E}\left[\int^{T}_{0}(P_{t}^{\nu}-\lambda_{0}\nu^{0}_{t})\nu^{0}_{t}dt\right] .
\end{equation*}
Together with  \eqref{eq-P-nu-definition}, \eqref{eq-major-inventory-def} and \eqref{eq-definition-z-process} we arrive at
\begin{equation}
\label{eq-pre-major-cash-expectation-2}
\begin{aligned}
\mathbb{E}\left[X^{0,\nu^{0}}_{T}\right] &= x_{0} +\mathbb{E}\left[\int^{T}_{0}M_{t}\nu^{0}_{t}dt\right] +\mathbb{E}\left[\int^{T}_{0}Z_{t}^{\nu}dQ^{0,\nu^{0}}_{t}\right]-\int^{T}_{0}\lambda_{0}(\nu^{0}_{t})^{2}dt.
\end{aligned}
\end{equation}
Recall that $Z_{0}^{\nu}=0$ and $Q^{0,\nu^{0}}_{T}=0$. Using integration by parts, \eqref{eq-permanent-impact}, \eqref{eq-definition-z-process} and Fubini's theorem we obtain 
\begin{equation}
\label{eq-It\^o-major-impact}
\begin{aligned}
\mathbb{E}\left[\int^{T}_{0}Z_{t}^{\nu}dQ^{0,\nu^{0}}_{t}\right] &=-\mathbb{E}\left[\int^{T}_{0}Q^{0,\nu^{0}}_{t}dZ_{t}^{\nu}\right] \\
&= -\int^{T}_{0}Q^{0,\nu}_{t}\left(\kappa_{0}\nu^{0}_{t} + \kappa_{1}\mathbb{E}\left[\nu^{1,*}_{t}(\nu^{0})\right]-\bar{\mu}_{t}\right)dt.
\end{aligned}
\end{equation}
Moreover, it follows from \eqref{eq-major-inventory-def} that 
\begin{equation}
\label{eq-identity-q0-squared}
\int^{T}_{0}Q^{0,\nu^{0}}_{t}\nu^{0}_{t} dt =\frac{q_0^{2}}{2}.
\end{equation}
By substituting \eqref{eq-F-identity}, \eqref{eq-It\^o-major-impact} and \eqref{eq-identity-q0-squared} into \eqref{eq-pre-major-cash-expectation-2} we get
\begin{equation}
\label{eq-expected-cash-major}
 \mathbb{E}\left[X^{0,\nu^{0}}_{T}\right]  
=x_{0} + M_{0}q_0  -\kappa_{0}\frac{q_0^{2}}{2} - \int^{T}_{0}Q^{0,\nu^{0}}_{t}(\kappa_{1}\mathbb{E}[\nu^{1,*}_{t}(\nu^{0})]-\bar{\mu}_{t})dt   -\int^{T}_{0}\lambda_{0}(\nu^{0}_{t})^{2}dt .
 \end{equation}
Notice that from Proposition \ref{proposition-operator-differential-eq} and Lemma \ref{lemma-expectation-v1} we have \begin{equation}
\label{eq-expectation-nu1-star-derivative}
\mathbb{E}[\nu_{t}^{1,*}(\nu^{0})] = \frac{\kappa_{0}}{2\lambda_{1}}\frac{d}{dt}(\mathsf{G}\nu^{0})(t) - \frac{1}{2\lambda_{1}}\frac{d}{dt}(\mathsf{G}\bar{\mu})(t),\quad \textrm{for all } 0\leq t\leq T.
\end{equation}
Plugging in \eqref{eq-expectation-nu1-star-derivative} into \eqref{eq-expected-cash-major} gives
\begin{equation}
\label{eq-expected-cash-major-2}
\begin{aligned}
\mathbb{E}\left[X^{0,\nu^{0}}_{T}\right]   
&=x_{0} + M_{0}q_0  -\kappa_{0}\frac{q_0^{2}}{2}- \frac{\kappa_{1}\kappa_{0}}{2\lambda_{1}}\int^{T}_{0}Q^{0,\nu^{0}}_{t}\frac{d}{dt}(\mathsf{G}\nu^{0})(t)   dt   \\ &\quad+\int^{T}_{0}Q^{0,\nu^{0}}_{t}\left(\frac{\kappa_{1}}{2\lambda_{1}}\frac{d}{dt}(\mathsf{G}\bar{\mu})(t)  +  \bar{\mu}_{t}\right)dt -\lambda_{0}\int^{T}_{0}(\nu^{0}_{t})^{2}dt .
\end{aligned}
\end{equation}
Since by \eqref{eq-mu-square-integrable} and \eqref{exp-s}, $\bar{\mu}\in L^{2}([0,T])$ and by \eqref{eq-major-admissible-set} also $\nu^{0}\in L^{2}([0,T])$, we get that from Proposition \ref{proposition-operator-differential-eq} that $(\mathsf{G}\bar{\mu})(0)=(\mathsf{G}\nu^{0})(0)=0$. Then by additional integration by parts and recalling that $Q^{0,\nu^{0}}_{T}=0$ we get
\begin{equation}
\label{eq-identity-derivative-G-Q}
 \int^{T}_{0} Q^{0,\nu^{0}}_{t}\frac{d}{dt}(\mathsf{G}\nu^{0})(t) dt = \int^{T}_{0} \nu^{0}_{t} (\mathsf{G}\nu^{0})(t) dt, 
 \end{equation}
as well as
\begin{equation}
\label{eq-identity-derivative-G-Q-mu}
 \int^{T}_{0} Q^{0,\nu^{0}}_{t}\frac{d}{dt}(\mathsf{G}\bar{\mu})(t) dt  =\int^{T}_{0} \nu^{0}_{t}(\mathsf{G}\bar{\mu})(t) dt .
\end{equation}
Hence, by plugging in \eqref{eq-identity-derivative-G-Q} and \eqref{eq-identity-derivative-G-Q-mu} into \eqref{eq-expected-cash-major-2} we obtain
\begin{equation}
\label{eq-expected-cash-major-3}
\begin{aligned}
\mathbb{E}\left[X^{0,\nu^{0}}_{T}\right]   &=x_{0} + M_{0}q_0  -\kappa_{0}\frac{q_0^{2}}{2}- \frac{\kappa_{1}\kappa_{0}}{2\lambda_{1}}\int^{T}_{0}\nu_{t}^{0}(\mathsf{G}\nu^{0})(t)  dt   \\ &\quad+\int^{T}_{0}\left(\frac{\kappa_{1}}{2\lambda_{1}}\nu^{0}_{t}(\mathsf{G}\bar{\mu})(t)  +  Q^{0,\nu^{0}}_{t}\bar{\mu}_{t}\right)dt -\lambda_{0}\int^{T}_{0}(\nu^{0}_{t})^{2}dt,
\end{aligned}
\end{equation}
which together \eqref{eq-major-functional}, proves the result.
\end{proof}

\section{Proofs of the Numerical Results in Section \ref{sec-numerics}}\label{proof-convergence-numerics}
Throughout this section we assume that Assumption \ref{ass-alpha-lambda} holds.  
Our first goal is to prove Proposition \ref{prop-convergence-Gn}, but before getting to the proof we introduce an auxiliary lemma. 
\begin{lemma}
\label{lemma-Gn-kernels-l2}
Let $(\mathcal{G}_{n})_{n \geq 1}^{}$ be defined as in \eqref{eq-convergence-kernel-series-G}. Then, $\mathcal{G}_{n}$ is in $L^{2}([0,T]^{2})$ for any $n\geq 1$.
\end{lemma}
\begin{proof}
Recall that $\mathsf{G}$ was defined in \eqref{eq-G-rep-integral}. From Proposition \ref{prop-G-adjoint-compact} it follows that $\mathsf{G}$ is an operator in $B(L^{2}([0,T]))$. Recall $(a_{i})_{n\geq 1}$ is a complete orthonormal basis in $L^{2}([0,T])$, hence from \eqref{eq-bj-integral} we get that $(b_{i})_{n\geq 1}$ are in $L^{2}([0,T])$. We therefore get from \eqref{eq-convergence-kernel-series-G} that
\begin{equation}
\begin{aligned}
\int^{T}_{0}\int^{T}_{0}\mathcal{G}_{n}(t,s)^{2}dsdt&\leq n\sum_{i=1}^{n}\left(\int^{T}_{0} a^{2}_{i}(t)dt \right)\left(\int^{T}_{0}b_{i}^{2}(s)ds\right) \\
&< \infty, 
\end{aligned}
\end{equation}
and the result follows. 
 \end{proof}
\begin{proof}[Proof of Proposition \ref{prop-convergence-Gn}]
The result follows directly from Lemma \ref{lemma-Gn-kernels-l2} and \eqref{eq-convergence-degenerate-kernel}.
\end{proof}
Before we prove Proposition \ref{prop-degenerate-kernel-approximation} we introduce the following theorem from \cite{Atkinson1997}. 
\begin{theorem}[{\cite[Theorem 2.1.1]{Atkinson1997}}]
\label{thm-theorem-atkinson}
Let $\mathsf{G}$ be in $B(L^{2}([0,T]))$ and let $\lambda\in\mathbb{R}$. Assume that $\mathsf{I} -\lambda\mathsf{G}$ is invertible on $L^{2}([0,T])$. Furthermore, assume that $(\mathsf{G}_{n})_{n\geq 1}$ is a sequence of operators in $B(L^{2}([0,T]))$ with
\begin{equation*}
\lim_{n\to\infty} || \mathsf{G}-\mathsf{G}_{n}||=0.
\end{equation*}
Then the following holds: 
\begin{itemize}
\item[\bf{(i)}] there exists an $N\geq 1$  such that for all $n\geq N$ the operators $\left(\mathsf{I} -\lambda\mathsf{G}_{n}\right)^{-1}$ exists and are in $ B(L^{2}([0,T]))$;
\item [\bf{(iii)}] $\left(\mathsf{I} -\lambda\mathsf{G}_{n}\right)^{-1}$ converges to $\left(\mathsf{I} -\lambda\mathsf{G}\right)^{-1}$ in $B(L^{2}([0,T]))$, that is 
\begin{equation*}
\lim_{n\to\infty} \left|\left|\left(\mathsf{I} -\lambda\mathsf{G}_{n}\right)^{-1}-\left(\mathsf{I} -\lambda\mathsf{G}\right)^{-1}\right|\right|=0;
\end{equation*}
\item [\bf{(iii)}] $ \left|\left|\left(\mathsf{I} -\lambda\mathsf{G}_{n}\right)^{-1}\right|\right|$ converges to $\left|\left|\left(\mathsf{I} -\lambda\mathsf{G}\right)^{-1}\right|\right|$, that is
\begin{equation*}
\lim_{n\to \infty}\left|\left|\left(\mathsf{I} -\lambda\mathsf{G}_{n}\right)^{-1}\right|\right| = \left|\left|\left(\mathsf{I} -\lambda\mathsf{G}\right)^{-1}\right|\right|.
\end{equation*}
\end{itemize}
\end{theorem}
We define
\begin{equation}
\label{eq-def-operators-Rn}
\mathsf{R}_{n}:=\left(\mathsf{I} + \frac{\kappa_{1}\kappa_{0}}{2\lambda_{0}\lambda_{1}}\mathsf{G}_{n}\right)^{-1}, \quad n\geq 1. 
\end{equation}
\begin{proof}[Proof of Proposition \ref{prop-degenerate-kernel-approximation}]
Recall that $(\mathsf{G}_{n})_{n\geq 1}$ was defined in \eqref{eq-operator-degenerate-kernel} and that $\mathsf{G}$ was defined as in  \eqref{eq-G-rep-integral}.   From Propositions \ref{prop-convergence-Gn}, \ref{prop-G-adjoint-compact} and \ref{prop-properties-resolvent} it follows that the assumptions of Theorem \ref{thm-theorem-atkinson} hold, hence there exists an $N\geq 1$ such that for all $n\geq N$ the operators $\mathsf{I} + \frac{\kappa_{1}\kappa_{0}}{2\lambda_{0}\lambda_{1}}\mathsf{G}_{n}$ are invertible. Since for any $n\geq N$ the corresponding kernels $\mathcal{G}_{n}$ in \eqref{eq-convergence-kernel-series-G} are degenerate, hence it follows from Theorem 2.1.2 of \cite{Atkinson1997} that the matrices $I_{n} + \frac{\kappa_{1}\kappa_{0}}{2\lambda_{0}\lambda_{1}} \mathbb{G}_{n}$ are invertible (recall \eqref{eq-matrix-G-frak} for the definition of $ \mathbb{G}_{n}$). 

Let $g,\psi\in L^{2}([0,T])$ and define  
\begin{equation*}
\gamma_{i} = - \frac{\kappa_{1}\kappa_{0}}{2\lambda_{0}\lambda_{1}}\sum_{j=1}^{n}\left( I_{n} + \frac{\kappa_{1}\kappa_{0}}{2\lambda_{0}\lambda_{1}} \mathbb{G}_{n}\right)^{-1}_{ij} \langle \psi, b_{j}\rangle_{L^{2}}, \quad i=1,\ldots,n,
\end{equation*}
 for any $n\geq N$. As shown in Chapter 3 of \cite{Porter1990} (see equations (3.5) -- (3.7) therein) the unique solution to
\begin{equation*}
\left(\mathsf{I} + \frac{\kappa_{1}\kappa_{0}}{2\lambda_{0}\lambda_{1}}\mathsf{G}_{n}\right) g = \psi, 
\end{equation*}
is given by
\begin{equation*}
\label{eq-inverse-frak-G}
g(t) = \psi(t) + \sum_{i=1}^{n}\gamma_{i}a_{i,n}(t),\quad \textrm{for all } 0\leq t\leq T, \quad n\geq N, 
\end{equation*}
and \eqref{eq-inverse-matrix-degenerate} follows. 
\end{proof}
Before proving Proposition \ref{prop-convergence-numerics-mean}, we need to present two intermediate lemmas.
\begin{lemma}
\label{lemma-convergence-inverse-inner-product}
Let $\mathsf{R}$ be defined as in \eqref{eq-resolvent-operator} and let $(\mathsf{R}_{n})_{n\geq 1}$ be defined as in \eqref{eq-def-operators-Rn}. Then the following holds:
\begin{itemize} 
\item[\textbf{(i)}] $$\liminf_n \langle\mathsf{R}_{n}\mathit{1},\mathit{1}\rangle_{L^{2}}>0, $$
\item [\textbf{(ii)}]  
\begin{equation*}
\lim_{n\to\infty} \frac{1}{\left\langle\mathsf{R}_{n}\mathit{1},\mathit{1}\right\rangle_{L^{2}}} = \frac{1}{\left\langle\mathsf{R}\mathit{1},\mathit{1}\right\rangle_{L^{2}}}.
\end{equation*}
\end{itemize} 
\end{lemma}
\begin{proof}
 (i) We have shown in the proof of Proposition \ref{prop-degenerate-kernel-approximation} that the assumptions of Theorem \ref{thm-theorem-atkinson} hold, hence there exists an $N\geq 1$ such that for all $n\geq N$ the operators $\mathsf{R}_{n}$ exist.  From the Cauchy-Schwarz inequality and since $\|\mathit{1}\|_{L^{2}}=T$ we get 
\begin{equation}
\label{eq-estimate-Rn}
\begin{aligned}
\left|\langle \mathsf{R}_{n}\mathit{1},\mathit{1}\rangle_{L^{2}} - \langle \mathsf{R}\mathit{1},\mathit{1}\rangle_{L^{2}}\right| 
&\leq \|\mathsf{R}_{n}-\mathsf{R}\| \|\mathit{1}\|^{2}_{L^{2}} \\
&\leq  \|\mathsf{R}_{n}-\mathsf{R}\| T^{2}.
\end{aligned}
\end{equation} 
From Proposition \ref{prop-properties-resolvent}(ii) it follows that the operator $\mathsf{R}$ is bounded from below, therefore, by Definition \ref{def-positive-operator}, there exists $\varepsilon>0$ such that 
\begin{equation}
\label{eq-R-positive-epsilon}
\langle \mathsf{R}\mathit{1},\mathit{1}\rangle_{L^{2}} > \varepsilon.
\end{equation}
From Theorem \ref{thm-theorem-atkinson}(iii) we get that there exists an $N_{1}\geq N$ such that for all $n\geq N_{1}$ we have that 
\begin{equation}
\label{eq-Rn-R-epsilon-prime}
||\mathsf{R}_{n}-\mathsf{R}|| < \varepsilon T^{-2}/2.
\end{equation}
From \eqref{eq-estimate-Rn}--\eqref{eq-Rn-R-epsilon-prime} we get (i). 

(ii) follows directly from \eqref{eq-estimate-Rn}, \eqref{eq-R-positive-epsilon} and (i). 
 \end{proof}

Next, we prove the convergence of the sequence of constants $(\eta_{n})_{n\geq 1}$ from \eqref{eq-eta-n}.
\begin{lemma}
\label{lemma-constant-convergence}
Let $\eta$ and $\eta_{n}$ be defined as in \eqref{eq-constant-major-solution} and \eqref{eq-eta-n}, respectively.
Then, there exists $N\geq 1$ such that for all $n\geq N$ the constants $\eta_{n}$ are well-defined. Moreover, 
\begin{equation} \label{con-eta} 
\lim_{n\to\infty} \eta_{n} = \eta .
\end{equation}
\end{lemma}
\begin{proof}[Proof of Lemma \ref{lemma-constant-convergence}]
From Lemma \ref{lemma-convergence-inverse-inner-product}(i), \eqref{eq-constant-major-solution} and \eqref{eq-def-operators-Rn} if follows that for all $n$ sufficiently large $\eta_{n}$ is well defined. The same claim holds for $\eta$ by  
\eqref{eq-eta-n}, \eqref{eq-R-positive-epsilon} and Proposition \ref{prop-properties-resolvent}(i). 
From Cauchy-Schwarz inequality we get
$$
\left|\left\langle\mathsf{R}_{n}\mathsf{S}\bar{\mu},\mathit{1}\right\rangle_{L^{2}} - \left\langle\mathsf{R}_{}\mathsf{S}\bar{\mu},\mathit{1}\right\rangle_{L^{2}} \right|  \leq \|\mathsf{R}-\mathsf{R}_{n}\|\|\mathsf{S}\bar{\mu}\|_{L^{2}}||\mathit{1}\|_{L^{2}}, 
$$
and together with Theorem \ref{thm-theorem-atkinson}(iii) and Lemma \ref{lemma-chi-continuous} it follows that 
\be\label{n-n-con} 
 \lim_{n\rr \infty} \left|\left\langle\mathsf{R}_{n}\mathsf{S}\bar{\mu},\mathit{1}\right\rangle_{L^{2}} - \left\langle\mathsf{R}_{}\mathsf{S}\bar{\mu},\mathit{1}\right\rangle_{L^{2}} \right| =0.
\ee
%
From \eqref{eq-constant-major-solution} and \eqref{eq-eta-n} we have 
\begin{equation} \label{eq-estimate-2-etan}
|\eta_n-\eta| = \left|\frac{\left\langle\mathsf{R}_{n}\mathsf{S}\bar{\mu},\mathit{1}\right\rangle_{L^{2}}}{\left\langle\mathsf{R}_{n}\mathit{1},\mathit{1}\right\rangle_{L^{2}}}  -\frac{\left\langle\mathsf{R}\mathsf{S}\bar{\mu},\mathit{1}\right\rangle_{L^{2}}}{\left\langle\mathsf{R}\mathit{1},\mathit{1}\right\rangle_{L^{2}}} \right|, 
\end{equation}
hence \eqref{con-eta} follows from Lemma \ref{lemma-convergence-inverse-inner-product} and \eqref{n-n-con}. 
 \end{proof}

We are now ready to prove Proposition \ref{prop-convergence-numerics-mean}.
\begin{proof}[Proof of Proposition \ref{prop-convergence-numerics-mean}] 
From \eqref{eq-major-solution}, \eqref{eq-integral-approximate-operator} and \eqref{eq-def-operators-Rn} we have 
\begin{equation*}
 \nu^{0,*} - \nu^{0,(n)} =\left( \mathsf{R}-\mathsf{R}_{n} \right)(\mathsf{S}\bar{\mu})  +\frac{1}{2\lambda_{0}} (\eta\mathsf{R} \mathit{1}- \eta_n\mathsf{R}_{n}\mathit{1}). 
 \end{equation*}
It follows that 
\begin{equation*}
\label{eq-convergence-mean-pre}
\begin{aligned}
||\nu^{0,*} - \nu^{0,(n)}||_{L^{2}} &\leq \left|\left|  \left(\mathsf{R}-\mathsf{R}_{n}\right) (\mathsf{S}\bar{\mu})\right|\right|_{L^{2}} +\frac{\eta}{2\lambda_{0}} \left|\left|  \left(\mathsf{R}-\mathsf{R}_{n}\right) \mathit{1}\right|\right|_{L^{2}} \\
&\quad   +\frac{1}{2\lambda_{0}} |\eta_n-\eta| \|\mathsf{R}_{n} \mathit{1} \|_{L^2}. 
\end{aligned}
\end{equation*}
Hence by Lemma \ref{lemma-constant-convergence}, Theorem \ref{thm-theorem-atkinson}(iii) and following similar lines as  in the proof of \eqref{n-n-con} we get 
\begin{equation*}
\lim_{n\to\infty}\left|\left|\nu^{0,*} - \nu^{0,(n)}\right|\right|_{L^{2}}=0. 
\end{equation*}
\end{proof}
Finally, we we prove Theorem  \ref{thm-uniform-convergence-numerics}.
\begin{proof}[Proof of Theorem \ref{thm-uniform-convergence-numerics}]
Throughout the proof we consider $n$'s large enough such that the results of Lemma \ref{lemma-constant-convergence} and Proposition \ref{prop-convergence-numerics-mean} hold, even if it is not stated explicitly. 

(i) From \eqref{eq-uniform-numerical-solution} and \eqref{eq-integral-eq-major} we get
\begin{equation*}
\label{eq-pre-kernel-subtraction}
\begin{aligned}
\nu^{0,*}_{t}-  \hat{\nu}^{0,(n)}_{t}  
&=  -\frac{\kappa_{1}\kappa_{0}}{2\lambda_{0}\lambda_{1}} \left(\mathsf{G}\left(\nu^{0,*}-\nu^{0,(n)}\right)\right)(t) + \frac{\eta -\eta_{n}}{2\lambda_{0}} \\
&= -\frac{\kappa_{1}\kappa_{0}}{2\lambda_{0}\lambda_{1}} \int^{T}_{0}\mathcal{G}(t,s)\left(\nu^{0,*}_{s} -\nu^{0,(n)}_{s}\right)ds + \frac{\eta -\eta_{n}}{2\lambda_{0}}, 
\end{aligned}
\end{equation*}
where we used \eqref{eq-G-rep-integral} in the second equality. 
 
From Cauchy-Schwarz inequality we get for all $0\leq t\leq T$ and $n$ sufficiently large, 
\begin{equation*}
\begin{aligned}
\left|\nu^{0,*}_{t}-  \hat{\nu}^{0,(n)}_{t}\right| 
&\leq  \frac{\kappa_{1}\kappa_{0}}{2\lambda_{0}\lambda_{1}}\left(\int^{T}_{0}|\mathcal{G}(t,s)|^{2}ds\right)^{1/2} \left|\left| \nu^{0,*} -\nu^{0,(n)}\right|\right|_{L^{2}} + \frac{|\eta-\eta_{n}|}{2\lambda_{0}}.
\end{aligned}
\end{equation*}
Hence, from Proposition \ref{prop-kernel-caligraphic-properties}, Lemma  \ref{lemma-constant-convergence} and Proposition \ref{prop-convergence-numerics-mean}, we get (i). 

(ii) From Proposition \ref{prop-results-ode} it follows that $r^{1}$ is bounded over $[0,T]$ and from Lemma \ref{lemma-kernel-L2} we get that the kernel $\mathcal{K}$ is bounded over $[0,T]^{2}$. Together with \eqref{eq-minor-agent-optimal} we get that there exists a constant $C>0$ such that for all $0\leq t\leq T$ and $n$ sufficiently large we have 
\begin{equation}
\label{eq-estimate-nu1-pre-sup-1}
\begin{aligned}
\left|\nu^{1,*}_{t}\left(\nu^{0,*}\right) - \nu^{1,*}_{t}\left(\hat{\nu}^{0,(n)}\right) \right|  &\leq \left| r^{0}_{t} \left(\hat{\nu}^{0,(n)}\right)- r^{0}_{t} \left(\nu^{0,*}\right)\right| \\ &\quad+C\left| \int^{t}_{0}  \left(r^{0}_{s} \left(\hat{\nu}^{0,(n)}\right)- r^{0}_{s} \left(\nu^{0,*}\right)\right)ds\right|. 
\end{aligned}
\end{equation}
By plugging in \eqref{eq-v0} into \eqref{eq-estimate-nu1-pre-sup-1} and observing that the stochastic part in the right hand side of \eqref{eq-v0} cancels, we conclude that  
\begin{equation}\label{eq-estimate-nu1-pre-sup-2}
\begin{aligned}
\left|\nu^{1,*}_{t}\left(\nu^{0,*}\right) - \nu^{1,*}_{t}\left(\hat{\nu}^{0,(n)}\right) \right|  &\leq C_1 \int_t^T \left| \hat{\nu}^{0,(n)}_s-  \nu^{0,*}_s\right|ds \\ &\quad+C_2\int^{t}_{0}  \left(\int_s^T \left| \hat{\nu}^{0,(n)}_r-  \nu^{0,*}_r\right|dr \right)ds, 
\end{aligned}
\end{equation}
for some constants $C_1,C_2>0$ independent from $n$ and $t$. 
Then (ii) follows from \eqref{eq-estimate-nu1-pre-sup-2} and (i). 
 \end{proof}

\appendix

\section{An Example of Spectral Decomposition of $\mathsf{G}$} \label{sec-spectral-explicit}
In this section we give an example of the spectral decomposition of $\mathsf{G}$ in Lemma  \ref{lemma-spectral-theorem-G} for the case where $\phi^{1}= 0$. We continue to assume that Assumption \ref{ass-alpha-lambda} holds. 
\begin{lemma}
\label{prop-T1-star-ode}
Let $\psi\in C([0,T])$ and recall that $\mathsf{K}_{1}$ be defined as in \eqref{eq-operator-T1-T2}. Then, $\mathsf{K}_{1}$ satisfies for $0\leq t \leq T$, 
\begin{equation*}
\label{eq-K1-star-ODE}
\frac{d}{dt} (\mathsf{K}^{*}_{1}\psi)(t) = - r^{1}_{t} (\mathsf{K}^{*}_{1}\psi)(t)  - \psi(t), \quad (\mathsf{K}^{*}_{1}\psi)(T)=0.
\end{equation*}
 In particular, $(\mathsf{K}^{*}_{1}\psi)(t)$ is continuously differentiable on $[0,T]$. 
\end{lemma}
\begin{proof} The proof follows the same lines of Proposition \ref{proposition-operator-differential-eq} hence we just just give the outlines. We now take a derivative of $\mathsf{K}_{1}^{*}\psi$ with respect to time using \eqref{eq-adjoint-k1-operator}, \eqref{eq-def-xi-pm} and \eqref{eq-kernel-minor} to get 
\begin{equation*}
\label{eq-derivative-K1-star-proof}
\begin{aligned}
\frac{d}{dt} (\mathsf{K}^{*}_{1}\psi)(t)
&=\frac{d}{dt}\left(\xi^{-}_{t}\int^{T}_{t} \xi^{+}_{s} \psi(s)ds\right) \\
&= -r^{1}_{t} (\mathsf{K}_{1}^{*}\psi)(t) - \psi(t), \quad \textrm{for all } 0\leq t\leq T. 
\end{aligned}
\end{equation*} 
Note that from \eqref{eq-adjoint-k1-operator} it follows that 
\begin{equation*}
(\mathsf{K}_{1}^{*}\psi)(T) =0. 
\end{equation*}
\end{proof}

\begin{proposition}
\label{prop-sufficient-spectral}
Let $\mathsf{G}$ be defined as in  \eqref{eq-G-rep-integral} and assume that $\phi^{1}~\equiv~0$. Let $(z_{n})_{n=1}^{\infty}$ be the increasing sequence of real positive roots of the following equation, 
\begin{equation}
\label{eq-transcendental}
\cot (z) = -\frac{(2\alpha-\kappa_{1})}{\lambda_{1}}\frac{T}{z}.
\end{equation}
Then the eigenvalues $(\zeta_{n})_{n=1}^{\infty}$ and the eigenfunctions $(\psi_{n})_{n=1}^{\infty}$ of $\mathsf{G}$ are given by 
\begin{equation}
\label{eq-zeta-eigenvalue}
\psi_{n}(t) = \frac{2}{\sqrt{\zeta_{n}}}\frac{\sin\left(\frac{t}{\sqrt{\zeta_{n}}}\right)}{\sqrt{\frac{2T}{\sqrt{\zeta_{n}}} -\sin\left(\frac{2T}{\sqrt{\zeta_{n}}}\right)}}, 
 \qquad \zeta_{n} = \frac{T^{2}}{z^{2}_{n}}.  
\end{equation}
\end{proposition}
\begin{proof}
We first show that the eigenvalues $(\zeta_{n})_{n=1}^{\infty}$ and eigenfunctions $(\psi_{n})_{n=1}^{\infty}$ are arising from  solutions to an ODE. Then we show that the solutions of the ODE can be determined in terms of the roots to \eqref{eq-transcendental}.

Let $\zeta$ be an eigenvalue of $\mathsf{G}$ and $\psi$ the corresponding eigenfunction, i.e. $\zeta$ and $\psi$ satisfy
\begin{equation}
\label{eq-eigen-G}
(\mathsf{G}\psi) (t) = \zeta \psi(t), \quad 0\leq t \leq T. 
\end{equation}
From Lemma \ref{lemma-spectral-theorem-G} it follows that $\zeta>0$ and $\psi\in L^{2}([0,T])$. Proposition \ref{proposition-operator-differential-eq} shows that $(\mathsf{G}\psi)(t)$ is continuously differentiable over $[0,T]$, therefore, it follows from \eqref{eq-eigen-G} that $\psi(t)$ is continuously differentiable over $[0,T]$. We take a derivative on both sides of \eqref{eq-eigen-G} to obtain that $(\zeta,\psi)$ must satisfy
\begin{equation}
\label{eq-G-derivative}
\frac{d}{dt}(\mathsf{G}\psi) (t) = \zeta \psi'(t), \quad \textrm{for all } 0\leq t \leq T.
\end{equation}
Proposition \ref{proposition-operator-differential-eq} shows that $(\mathsf{G}\psi)$ is the solution to \eqref{eq-G-ODE}, therefore we can substitute \eqref{eq-G-ODE} in \eqref{eq-G-derivative} to get that $\psi$ must satisfy
\begin{equation}
\label{eq-G-derivative-2}
r^{1}_{t} (\mathsf{G}\psi)(t) + (\mathsf{K}_{1}^{*}\psi)(t) = \zeta \psi'(t) , \quad \textrm{for all } 0\leq t \leq T.
 \end{equation}
Proposition \ref{prop-results-ode} proves that $r^{1}$ is the solution to \eqref{eq-v1-ode-proof}. When $\phi^{1}\equiv0$ it can be computed explicitly as follows, 
\begin{equation*}
r^{1}_{t} = \frac{2\alpha -\kappa_{1}}{(t-T)(2\alpha-\kappa_{1}) -2\lambda_{1}},\quad t\in[0,T]. 
\end{equation*}
Note that under Assumption \ref{ass-alpha-lambda}, $r^{1}$ is continuously differentiable on $[0,T]$.  

Since we have proved that $\psi(t)$ is continuous, it follows from Proposition \ref{prop-T1-star-ode} that $(\mathsf{K}_{1}^{*}\psi)(t)$ is continuously differentiable on $[0,T]$. We take a derivative on both sides of \eqref{eq-G-derivative-2} to get
\begin{equation*}
\label{eq-ode-psi-1}
\frac{dr^{1}_{t}}{dt} (\mathsf{G}\psi)(t) + r^{1}_{t}  \frac{d }{dt}(\mathsf{G}\psi)(t) + \frac{d}{dt}(\mathsf{K}_{1}^{*}\psi)(t) = \zeta  \psi''(t),
\end{equation*}
and then use \eqref{eq-G-ODE} to get that $\psi$ satisfies
\begin{equation}
\label{eq-ode-psi-2}
\frac{dr_{t}^{1}}{dt} (\mathsf{G} \psi)(t) +  (r^{1}_{t})^{2}(\mathsf{G}\psi)(t)  + r^{1}_{t}(\mathsf{K}_{1}^{*}\psi)(t) + \frac{d}{dt}(\mathsf{K}_{1}^{*}\psi)(t) = \zeta  \psi''(t), \quad 0\leq t \leq T. 
\end{equation}
By applying \eqref{eq-v1-ode-proof}, \eqref{eq-G-derivative-2} and \eqref{eq-K1-star-ODE} to \eqref{eq-ode-psi-2} we get  that $\psi$ must satisfy
\begin{equation}
\label{eq-ode-psi-2-2}
-\psi(t)= \zeta \psi''(t), \quad 0\leq t \leq T. 
\end{equation}
Recall that $\zeta>0$, hence it follows from \eqref{eq-eigen-G} and Proposition \ref{proposition-operator-differential-eq} that $\psi$ satisfies the initial condition $\psi(0)=0$. The terminal condition $\psi'(T)=-\left(\frac{2\alpha-\kappa_{1}}{\lambda_{1}}\right)\psi(T) $ follows by combining \eqref{eq-G-derivative-2} with \eqref{eq-v1-ode-proof}, \eqref{eq-eigen-G} and $(\mathsf{K}_{1}^{*}\psi)(T)=0$ (see Proposition \ref{prop-T1-star-ode}). It follows that $(\zeta,\psi)$ satisfy 
\begin{equation}
\label{eq-spectral-ODE}
 \psi''(t) =  -\frac{1}{\zeta}\psi(t), \quad 0< t < T,  \quad \psi(0)=0 , \quad \psi '(T)= -\left(\frac{2\alpha-\kappa_{1}}{\lambda_{1}}\right)\psi(T).
\end{equation}
We show that \eqref{eq-transcendental} has an infinite number of positive roots. To see this, note that since $2\alpha-\kappa_{1}\geq 0$ by Assumption \ref{ass-alpha-lambda}, then for any $n\geq 1$, 
\begin{equation*}
\begin{aligned}
\lim_{z\to (n-1)\pi^{+}} \cot(z) + \frac{2\alpha-\kappa_{1}}{\lambda_{1}}\frac{T}{z} &= +\infty, \\
\lim_{z\to n\pi^{-}} \cot(z) + \frac{2\alpha-\kappa_{1}}{\lambda_{1}}\frac{T}{z} &= -\infty. 
\end{aligned}
\end{equation*}
Since $\cot(z) + \frac{2\alpha-\kappa_{1}}{\lambda_{1}}\frac{T}{z}$ is continuous over the intervals $((n-1)\pi,  n\pi)$ for any $n\geq 1$, then it follows by the intermediate value theorem that \eqref{eq-transcendental} has a root in the interval $((n-1)\pi,  n\pi)$ for any $n\geq 1$.

Next we identify $\zeta_n$ as in \eqref{eq-zeta-eigenvalue}. Let $z_{n}$ be the $n^{\text{th}}$ positive root of \eqref{eq-transcendental} and let $\zeta_{n}$, $\psi_{n}$ be defined as in \eqref{eq-zeta-eigenvalue}. First, note that since $z_{n}>0$, then 
\begin{equation}
\label{eq-z-sinz-inequality}
2 z_{n} - \sin\left(2z_{n}\right)>0.
\end{equation}
From \eqref{eq-zeta-eigenvalue} and \eqref{eq-z-sinz-inequality} it follows that
\begin{equation*}
\frac{2T}{\sqrt{\zeta_{n}}} -\sin\left(\frac{2T}{\sqrt{\zeta_{n}}}\right)>0,
\end{equation*}
therefore, the function $\psi_{n}(t)$ is well-defined for all $t\in[0,T]$ and $||\psi_{n}||_{L^{2}}=1$. Using the following identity which arises from \eqref{eq-spectral-ODE}, 
\begin{equation*}
\label{eq-trascendental-proof}
\cos(z_{n}) = -\frac{(2\alpha-\kappa_{1})}{\lambda_{1}}\frac{T}{z_{n}}\sin(z_{n}), 
\end{equation*}
it is easy to verify that $\psi_{n}$ in \eqref{eq-zeta-eigenvalue} solves \eqref{eq-spectral-ODE} with $\zeta=\zeta_{n}$ for any $n\geq 1$. This completes the proof. 
\end{proof}

\section{Proof of Lemmas \ref{lemma-gat} and \ref{lemma-fbsde-minor}} 
\label{sec-gat-pf} 
\begin{proof}[Proof of Lemma \ref{lemma-gat}] 
Let $\epsilon>0$ and $\nu^{1},\omega \in\mathcal{A}_{m}$. 
We note that from \eqref{eq-minor-inventory} it follows that 
\begin{equation}
\label{eq-q1-epsilon-omega}
Q^{1,\nu^{1}+\epsilon \omega}_{t} = Q^{1,\nu^{1}}_{t}-\epsilon\int^{t}_{0}\omega_{s}ds, \quad  \textrm{for all } 0\leq t \leq T. 
\end{equation}
We use the alternative representation of $H^{1}$ in \eqref{eq-alternative-minor} and \eqref{eq-q1-epsilon-omega} to get
\begin{equation}
\label{eq-difference-G\^ateaux}
\begin{aligned}
H^{1}(\nu^{1}+\epsilon\omega)-H^{1}(\nu^{1}) &=\epsilon \mathbb{E}\Bigg[ \int^{T}_{0}\omega_{t}(-2\lambda_{1}\nu_{t}^{1} +2\alpha Q^{1,\nu^{1}}_{T} -\kappa_{1}Q^{1,\nu^{1}}_{t}) dt \\ 
&\quad+ \int^{T}_{0}\left(\int^{t}_{0}\omega_{s}ds\right)(2\phi^{1}_{t}Q_{t}^{1,\nu^{1}}dt + \kappa_{0}\nu^{0}_{t}dt+ \kappa_{1}\nu_{t}^{1}dt-dA_{t}) \Bigg]\\ &\quad+ \epsilon^{2}\mathbb{E}\Bigg[-\lambda_{1}\int^{T}_{0}\omega_{s}^{2}ds -\alpha \left(\int^{T}_{0}\omega_{s}ds\right)^{2} \\ &\quad- \int^{T}_{0}\phi^{1}_{t}\left(\int^{t}_{0}\omega_{s}ds\right)^{2}dt + \kappa_{1}\int^{T}_{0}\omega_{t}\left(\int^{t}_{0}\omega_{s}ds\right)dt\Bigg].
\end{aligned}
\end{equation}
From \eqref{eq-def-minor-G\^ateaux} and \eqref{eq-difference-G\^ateaux} we get
\begin{equation}
\label{eq-minor-G\^ateaux-2}
\begin{aligned}
\langle\mathcal{D}H^{1}(\nu^{1}),\omega\rangle &=\mathbb{E}\Bigg[ \int^{T}_{0}\omega_{t}(-2\lambda_{1}\nu_{t}^{1} +2\alpha Q^{1,\nu^{1}}_{T} -\kappa_{1}Q^{1,\nu^{1}}_{t}) dt \\ 
&\quad+ \int^{T}_{0}\left(\int^{t}_{0}\omega_{s}ds\right)(2\phi^{1}_{t}Q_{t}^{1,\nu^{1}}dt + \kappa_{0}\nu^{0}_{t}dt+ \kappa_{1}\nu_{t}^{1}dt -dA_{t} )\Bigg].
\end{aligned}
\end{equation}
Since $\nu^{1},\omega\in\mathcal{A}_{m}$, $\nu^{0}\in\mathcal{A}_{M}^{q_0}$ and $\mathbb{E}[(\int^{T}_{0}|dA_{t}|)^{2}]<\infty$, then use Fubini's theorem in \eqref{eq-minor-G\^ateaux-2} to get
\begin{multline*}
\langle\mathcal{D}H^{1}(\nu^{1}),\omega\rangle = \mathbb{E}\Bigg[ \int^{T}_{0}\omega_{t}\Big(-2\lambda_{1}\nu_{t}^{1} +2\alpha Q^{1,\nu^{1}}_{T} -\kappa_{1}Q^{1,\nu^{1}}_{t} + A_{t} \\ 
+ \int^{T}_{t}\left(2\phi^{1}_{s}Q_{s}^{1,\nu^{1}} + \kappa_{0}\nu^{0}_{s} + \kappa_{1}\nu_{s}^{1}\right)ds - A_{T}\Big)dt\Bigg],
\end{multline*}
which concludes the result.
\end{proof}

\begin{proof}[Proof of Lemma \ref{lemma-fbsde-minor}] 
In Lemma \ref{lemma-minor-strictly-concave} we have shown that under Assumption \ref{ass-alpha-lambda},  the functional $H^{1}(\nu^1)$ is strictly concave of any $\nu^1 \in \mathcal{A}_{m}$. Therefore, we may apply Proposition 2.1 of \cite[Chapter II]{EkelTem:99} to obtain that
\begin{equation}
\label{eq-iff-G\^ateaux-minor}
\langle \mathcal{D}H^{1}(\nu^{1,*}),\omega\rangle =0\quad \textrm{for all } \omega \in\mathcal{A}_{m}  \iff \nu^{1,*}=\arg\sup_{\nu\in\mathcal{A}_{m}}H^{1}(\nu).
\end{equation}
The strict concavity of $H^{1}$ guarantees that the optimiser $\nu^{1,*}$ is unique.
\paragraph{Necessity:}We assume that 
\begin{equation*}
\nu^{1,*}=\arg\sup_{\nu\in\mathcal{A}_{m}}H^{1}(\nu).
\end{equation*}
Then,  \eqref{eq-iff-G\^ateaux-minor} and \eqref{eq-minor-G\^ateaux} imply that for all $\omega\in \mathcal{A}_{m}$ it holds
\begin{equation*}
\label{eq-gateaux-zero-minor}
\begin{aligned}
\langle\mathcal{D}H^{1}(\nu^{1,*}),\omega\rangle &= \mathbb{E}\Bigg[ \int^{T}_{0}\omega_{t}\Big(-2\lambda_{1}\nu_{t}^{1,*} +2\alpha Q^{1,\nu^{1,*}}_{T} -\kappa_{1}Q^{1,\nu^{1,*}}_{t} + A_{t}  \\ 
&\quad+ \int^{T}_{t}\left(2\phi^{1}_{s}Q_{s}^{1,\nu^{1,*}} + \kappa_{0}\nu^{0}_{s}  + \kappa_{1}\nu_{s}^{1,*}\right)ds- A_{T}\Big)dt\Bigg]=0.
\end{aligned}
\end{equation*}
By applying the optional projection theorem we get 
\begin{multline}
\label{eq-pre-first-order}
\mathbb{E}\Bigg[ \int^{T}_{0}\omega_{t}\Bigg(-2\lambda_{1}\nu_{t}^{1,*} +\mathbb{E}_{t}\left[2\alpha Q^{1,\nu^{1,*}}_{T}  - A_{T}\right]-\kappa_{1}Q^{1,\nu^{1,*}}_{t}  + A_{t} \\ 
+\mathbb{E}_{t}\left[ \int^{T}_{t}\left(2\phi^{1}_{s}Q_{s}^{1,\nu^{1,*}} + \kappa_{0}\nu^{0}_{s}  + \kappa_{1}\nu_{s}^{1,*}\right)ds\right]\Bigg)dt\Bigg]=0, \end{multline}
As \eqref{eq-pre-first-order} holds for all $\omega\in\mathcal{A}_{m}$ we deduce the following first order condition holds
\begin{multline}
\label{eq-first-order}
-2\lambda_{1}\nu_{t}^{1,*} +\mathbb{E}_{t}\left[2\alpha Q^{1,\nu^{1,*}}_{T}- A_{T}\right] -\kappa_{1}Q^{1,\nu^{1,*}}_{t}  + A_{t}\\ 
+\mathbb{E}_{t}\left[ \int^{T}_{t}\left(2\phi^{1}_{s}Q_{s}^{1,\nu^{1,*}} + \kappa_{0}\nu^{0}_{s} + \kappa_{1}\nu_{s}^{1,*}\right)ds\right]=0,
\end{multline} 
$d\mathbb{P}\otimes dt$-a.e. on $\Omega\times[0,T]$.

We define the following martingales
\begin{equation}
\label{eq-optimal-martingale}
\begin{aligned}
\mathcal{M}_{t} &:=\mathbb{E}_{t}\left[ \int^{T}_{0}\left(2\phi^{1}_{s}Q_{s}^{1,\nu^{1,*}} + \kappa_{0}\nu^{0}_{s}+ \kappa_{1}\nu_{s}^{1,*}\right)ds\right],\\
\mathcal{N}_{t} &:= \mathbb{E}_{t}\left[2\alpha Q^{1,\nu^{1,*}}_{T}- A_{T}\right].
\end{aligned}
\end{equation}
Note that $\mathcal{M}$ and $\mathcal{N}$ are square-integrable since $\mathbb{E}[(\int^{T}_{0}|dA_{t}|)^{2}]<\infty$, $\nu^{1,*},\omega\in\mathcal{A}_{m}$ and $\nu^{0}\in\mathcal{A}_{M}^{q_0}$.

We plug $\mathcal{M}$ and $\mathcal{N}$ in \eqref{eq-first-order} to get
\be\label{eq-representation-optimal-minor}
\begin{aligned}
-2\lambda_{1}\nu_{t}^{1,*} +\mathcal{N}_{t}-\kappa_{1}Q^{1,\nu^{1,*}}_{t} + A_{t}  
+\mathcal{M}_{t}-\int^{t}_{0}\left(2\phi^{1}_{s}Q_{s}^{1,\nu^{1,*}} + \kappa_{0}\nu^{0}_{s} + \kappa_{1}\nu_{s}^{1,*}\right)ds=0,
\end{aligned}
\ee
From \eqref{eq-minor-inventory} and \eqref{eq-representation-optimal-minor} it follows that $\nu^{1,*}$ solves the following BSDE
\begin{equation*}
\begin{cases}
d\nu_{t}^{1,*}&= \frac{1}{2\lambda_{1}}d\mathcal{N}_{t}  
+ \frac{1}{2\lambda_{1}}d\mathcal{M}_{t}-\frac{1}{\lambda_{1}}\phi^{1}_{t}Q_{t}^{1,\nu^{1,*}}dt -\frac{\kappa_{0}}{2\lambda_{1}}\nu^{0}_{t}dt + \frac{1}{2\lambda_{1}}dA_{t},\\
\nu_{T}^{1,*}&=\frac{2\alpha -\kappa_{1}}{2\lambda_{1}}Q_{T}^{1,\nu^{1,*}},
\end{cases}
\end{equation*}
this gives \eqref{eq-minor-fbsde}.

\paragraph{Sufficiency:} Assume that $(Q^{1,\nu^{1,*}},\nu^{1,*})$ solves \eqref{eq-minor-fbsde} $d\mathbb{P}\otimes dt$-a.e. and that $\nu^{1,*}\in\mathcal{A}_{m}$. 
We will show that $\langle \mathcal{D}H^{1}(\nu^{1,*}),\omega\rangle$ vanishes for all $\omega\in\mathcal{A}_{m}$ which, once combined with \eqref{eq-iff-G\^ateaux-minor}, implies that $\nu^{1,*}$ is the solution to the minor agent's problem.
Since $(Q^{1,\nu^{1,*}},\nu^{1,*})$ solves \eqref{eq-minor-fbsde} we get that 
\begin{equation}
\begin{aligned}
2\lambda_{1}\nu_{t}^{1,*}
 &=\mathbb{E}_{t}\left[\left(2\alpha-\kappa_{1}\right)Q^{1,\nu^{1,*}}_{T}\right] -  \mathbb{E}_{t}\left[\int^{T}_{t}dA_{s}\right]  +  \mathbb{E}_{t}\left[\int^{T}_{t}\left(\kappa_{0}\nu^{0}_{s} + 2\phi^{1}_{t}Q^{1,\nu^{1,*}}_{s}\right)ds\right]  \\
 &= \mathbb{E}_{t}\left[2\alpha Q^{1,\nu^{1,*}}_{T}- A_{T}\right] -\kappa_{1}Q^{1,\nu^{1,*}}_{t}  + A_{t} \\ & \ \quad+  \mathbb{E}_{t}\left[\int^{T}_{t}\left(\kappa_{0}\nu^{0}_{s} + 2\phi^{1}_{t}Q^{1,\nu^{1,*}}_{s} + \kappa_{1}\nu^{1,*}_{s}\right)ds\right], \quad d\mathbb{P}\otimes dt-\textrm{a.e.},
\end{aligned}
\end{equation}
where we used \eqref{eq-minor-inventory} in the second equality. Hence, $\nu^{1,*}_{t}$ satisfies \eqref{eq-first-order}, therefore the left-hand side of \eqref{eq-iff-G\^ateaux-minor} hold.
\end{proof}

\section{Proof of Proposition \ref{prop-results-ode} and Lemma \ref{lemma-kernel-L2}}
Before we prove Lemma \ref{lemma-kernel-L2} we introduce the following lemma. 
\begin{lemma} \label{lem-sol-ric} 
Under Assumption \ref{ass-alpha-lambda}, the Riccati equation \eqref{eq-v1-ode-proof} has a unique continuous solution. 
\end{lemma} 
\begin{proof} 
Let $\hat{r}^{1}$ be the solution to the following equation
\begin{equation}
\label{eq-hat-v-ode}
\begin{cases}
\partial_{t}\hat{r}^{1}_{t} &= -\frac{1}{\lambda_{1}}\phi^{1}_{t} + (\hat{r}^{1}_{t})^{2}, \quad 0\leq t \leq T,\\
\hat{r}^{1}_{T} &= \frac{2\alpha-\kappa_{1}}{2\lambda_{1}}.
\end{cases}
\end{equation} 
Since, $\hat{r}^{1}_{T} \geq 0$ and $\phi^{1}$ is a piecewise continuous, locally bounded non-negative function over $[0,T]$, then by Theorem 2.1 of \cite{Wonham1968} there exists a unique solution $\hat{r}^{1}$, which is absolutely continuous on $[0,T]$ (see also Theorem 3.5 in \cite{Freiling:02} for a more recent reference). As stated by \cite{Wonham1968}, the function $\hat{r}^{1}$ satisfies \eqref{eq-hat-v-ode} only $dt$ almost everywhere. Note that by taking $ r^{1}_{t}=-\hat{r}^{1}_{t}$, it follows that $r^{1}$ is an absolutely continuous solution to \eqref{eq-v1-ode-proof}. Uniqueness of the solution to \eqref{eq-v1-ode-proof} then follows by the uniqueness for \eqref{eq-hat-v-ode}. 
\end{proof} 

\begin{proof}[Proof of Lemma \ref{lemma-kernel-L2}] 

Lemma \ref{lem-sol-ric} proves that $r^{1}$ is continuous on $[0,T]$.  Then, from \eqref{eq-def-xi-pm} it follows that the functions $\xi^{\pm}_{t}$ are continuous. Therefore, from \eqref{eq-kernel-minor} we get that $\mathcal{K}$ is jointly continuous on $[0,T]^{2}$ hence it $\mathcal{K}$ is bounded on $[0,T]^{2}$ and  
\begin{equation*}
\int^{T}_{0}\int^{T}_{0}\left|\mathcal{K}(t,s)\right|^{2}dsdt <\infty,
\end{equation*}
that is $\mathcal{K}$ is in $L^{2}([0,T]^{2})$. 
\end{proof}

Now we are ready to prove Proposition \ref{prop-results-ode}. 
\label{sec-proof-prop-ode}
\begin{proof}[Proof of Proposition \ref{prop-results-ode}]
In Lemma \ref{lem-sol-ric} we have established that \eqref{eq-v1-ode-proof} has a unique continuous solution $r^{1}$. We prove the rest of the claims in the following two steps. 

\textit{Step 1.} We show that $r^{0}_{t}$ given by \eqref{eq-v0} solves the BSDE \eqref{eq-v0-bsde};
Note that since $r^{1}$ is continuous, the function $\xi^{+}$ in \eqref{eq-def-xi-pm} is the unique solution of the ODE
\begin{equation}
\label{eq-xi-ode}
\frac{d \xi_{t}^{+}}{dt} = r^{1}_{t}\xi_{t}^{+}, \quad  \xi_{0}^{+} = 1. 
\end{equation}
Since $r^{1}_{t}$ is continuous on $[0,T]$, it holds that
\begin{equation}
\label{eq-square-integral-xi}
\int^{T}_{0}(\xi_{t}^{+})^{2}dt  < \infty.
\end{equation}
Since $\xi_{t}^{+}$ satisfies \eqref{eq-square-integral-xi}, then the process  
\begin{equation}
\label{eq-xi-v0}
\xi_{t}^{+}r^{0}_{t} := \frac{1}{2\lambda_{1}}\mathbb{E}_{t}\left[\int^{T}_{t} \xi_{s}^{+}(dA_{s}-\kappa_{0}\nu^{0}_{s}ds)\right] , \quad 0\leq t \leq T, 
\end{equation}
is the unique strong solution to the following linear BSDE
\begin{equation}
\label{eq-xi-v0-bsde}
\begin{cases}
d(\xi_{t}^{+}r^{0}_{t})&=\frac{\xi_{t}^{+}}{2\lambda_{1}}(\kappa_{0}\nu^{0}_{t}dt -dA_{t}) - \frac{1}{2\lambda_{1}}\xi_{t}^{+}d\mathcal{M}_{t}- \frac{1}{2\lambda_{1}}\xi_{t}^{+}d\mathcal{N}_{t}, \\
\xi_{T}^{+}r^{0}_{T} &=0. 
\end{cases}
\end{equation}
We multiply both sides in \eqref{eq-xi-v0} by $\xi_{t}^{-}$ from \eqref{eq-def-xi-pm} and use the identity $\xi_{t}^{-}\xi_{t}^{+}=1$. By doing so, we can obtain an expression for $r^{0}_{t}$, that is
\begin{equation}
\label{eq-v0-solution}
r^{0}_{t} = \frac{1}{2\lambda_{1}} \mathbb{E}_{t}\left[\int^{T}_{t} \xi_{t}^{-}\xi_{s}^{+}(dA_{s}-\kappa_{0}\nu^{0}_{s}ds)\right].
\end{equation} 
We now show that $r^{0}$ from \eqref{eq-v0-solution} is the solution to \eqref{eq-v0-bsde}.  
From \eqref{eq-xi-v0-bsde} and It\^o's product rule we get
 \begin{equation}
\label{eq-v0-solution-2}
d\xi_{t}^{+}r^{0}_{t} + \xi_{t}^{+}dr^{0}_{t}=  \frac{\xi_{t}^{+}}{2\lambda_{1}}(\kappa_{0}\nu^{0}_{t}dt -dA_{t}) - \frac{1}{2\lambda_{1}}\xi_{t}^{+}d\mathcal{M}_{t} - \frac{1}{2\lambda_{1}}\xi_{t}^{+}d\mathcal{N}_{t}.
\end{equation}
Next, we use \eqref{eq-xi-ode} and \eqref{eq-v0-solution-2} to get
\begin{equation}
\xi_{t}^{+}\left(r^{1}_{t}r^{0}_{t}dt + dr^{0}_{t}-  \frac{1}{2\lambda_{1}}(\kappa_{0}\nu^{0}_{t}dt -dA_{t})+ \frac{1}{2\lambda_{1}}d\mathcal{M}_{t}+ \frac{1}{2\lambda_{1}}d\mathcal{N}_{t}\right) = 0.
\end{equation}
Since $\xi_{t}^{+}>0$ for all $t\in[0,T]$, we have 
\begin{equation*}
r^{1}_{t}r^{0}_{t}dt + dr^{0}_{t}-  \frac{1}{2\lambda_{1}}(\kappa_{0}\nu^{0}_{t}dt -dA_{t}) + \frac{1}{2\lambda_{1}}d\mathcal{M}_{t}+ \frac{1}{2\lambda_{1}}d\mathcal{N}_{t} =0.
\end{equation*}
with terminal condition $r^{0}_{T}=0$ which follows from \eqref{eq-v0-solution}. By comparing this with \eqref{eq-v0-bsde}, it follows that $r^{0}_{t}$
 is the solution to the BSDE \eqref{eq-v0-bsde}. Recall the definition of $\mathcal{K}$ in \eqref{eq-kernel-minor}. We substitute the expression for $\mathcal{K}(t,s)$ from \eqref{eq-kernel-minor} into \eqref{eq-v0-solution} from which it follows that $r^{0}_{t}$ given by \eqref{eq-v0} solves \eqref{eq-v0-bsde}.

\textit{Step 2.} We show \eqref{sq-r0}. From \eqref{eq-kernel-minor} it follows that $\mathcal{K}(t,s)>0$ for all $t,s\in[0,T]^{2}$. Moreover by Lemma \ref{lemma-kernel-L2}, $\mathcal{K}$ is bounded on $[0,T]^{2}$. Since $A$ is of bounded variation and $\mathcal{K}$ is bounded, then by using the conditional Jensen's inequality and the tower property we get
\begin{equation}
\label{eq-sup-kernel-A-2}
\begin{aligned}
\sup_{t \in [0,T]} \mathbb E\left[\left(\mathbb{E}_{t}\left[\int^{T}_{t} \mathcal{K}(t,s) dA_{s}\right]\right)^2 \right] &\leq C \mathbb{E}\left[\left(\int^{T}_{0} |dA_{s}|\right)^2\right] \\
&<\infty, 
\end{aligned}
\end{equation} 
where we used \eqref{ass:P} in the last inequality.  Similarly, we can obtain that
\begin{equation}
\label{eq-bound-kernel-nu0}
\begin{aligned}
\sup_{t \in [0,T]} \mathbb E\left[\left(\mathbb{E}_{t}\left[\int^{T}_{t} \mathcal{K}(t,s) \nu^0_sds\right]\right)^2 \right]  & \leq {C}\mathbb{E}\left[\int^{T}_{0}(\nu^{0}_{s})^2ds\right] \\
&< \infty, 
\end{aligned}
\end{equation}
where we used the fact that $\nu^{0}\in\mathcal{A}_{M}^{q_0}$ and \eqref{eq-major-admissible-set}. 
From \eqref{eq-v0}, \eqref{eq-sup-kernel-A-2} and \eqref{eq-bound-kernel-nu0} we get
\begin{equation*}
 \sup_{t\in[0,T]}\mathbb{E}[(r^{0}_{t})^{2}] <\infty, 
\end{equation*}
and \eqref{sq-r0} follows. 
 \end{proof}



\end{document}